\documentclass[a4paper,USenglish,cleveref,pdfa]{lipics-v2021}

\pdfoutput=1
\hideLIPIcs
\nolinenumbers

\usepackage{preamble}

\title{Parameterized Capacitated Vertex Cover Revisited}

\titlerunning{Parameterized Capacitated Vertex Cover Revisited}

\author{Michael Lampis}
{Universit\'{e} Paris Dauphine -- PSL, CNRS UMR7243, LAMSADE, Paris, France
\and \url{https://www.lamsade.dauphine.fr/~mlampis/index.html}}
{michail.lampis@dauphine.fr}
{https://orcid.org/0000-0002-5791-0887}{}

\author{Manolis Vasilakis}
{University of Southern Denmark, Odense, Denmark
\and \url{https://manolisvasilakis.github.io/}}
{vasilakis@imada.sdu.dk}
{https://orcid.org/0000-0001-6505-2977}
{\flag{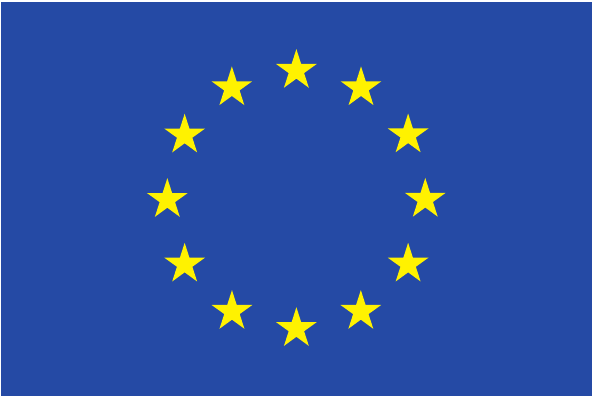}\flag{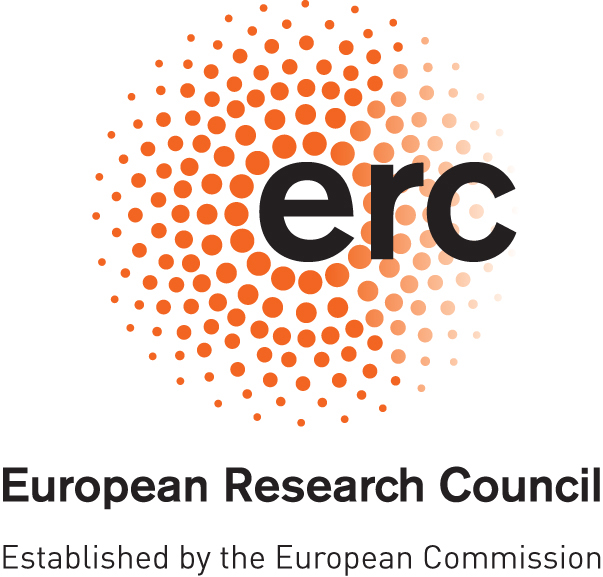}Supported by ERC grant PARAMLP (project no.~101220563) funded by the European Union.
Views and opinions expressed are however those of the author(s) only and do not necessarily reflect those of the European Union or the European Research Council Executive Agency.
Neither the European Union nor the granting authority can be held responsible for them.}

\authorrunning{M. Lampis and M. Vasilakis}

\Copyright{Michael Lampis and Manolis Vasilakis}



\category{} %optional, e.g. invited paper

\funding{This work is partially supported by ANR project ANR-21-CE48-0022 (S-EX-AP-PE-AL).}

\acknowledgements{Part of this work was done while the second author was at Universit\'{e} Paris Dauphine -- PSL.}

\EventEditors{Neeldhara Misra and Magnus Wahlstr\"{o}m}
\EventNoEds{2}
\EventLongTitle{18th International Symposium on Parameterized and Exact Computation (IPEC 2023)}
\EventShortTitle{IPEC 2023}
\EventAcronym{IPEC}
\EventYear{2023}
\EventDate{September 6--8, 2023}
\EventLocation{Amsterdam, The Netherlands}
\EventLogo{}
\SeriesVolume{285}
\ArticleNo{22}
\begin{document}

\maketitle

\begin{abstract}
\textsc{Capacitated Vertex Cover} is a generalization of \textsc{Vertex Cover}
where we seek a set of $k$ vertices covering all edges, but with the added
constraint that for each vertex we are given a capacity, that is, a maximum
number of incident edges the vertex can cover if selected. This problem is
well-studied in parameterized complexity and is known to be in XP but W[1]-hard
parameterized by treewidth. Natural dynamic programs lead to complexities of
the form $C^{\mathcal{O}(\mathrm{tw})}n^{\mathcal{O}(1)}$,
$\Delta^{\mathcal{O}(\mathrm{tw})}n^{\mathcal{O}(1)}$, or
$k^{\mathcal{O}(\mathrm{tw})}n^{\mathcal{O}(1)}$, where $C, \Delta$ are
respectively the maximum capacity and maximum degree. These results imply an
FPT algorithm for parameter $k$ with complexity
$k^{\mathcal{O}(k)}n^{\mathcal{O}(1)}$, as $\mathrm{tw}\le k$.

These simple dynamic programming algorithms essentially constitute the state of
the art.  Our main question is therefore whether these algorithms can be
improved. We address this question from numerous angles, most notably giving
the following three sets of results:

\begin{enumerate}

\item Regarding the natural parameter $k$ we show that the treewidth-based
algorithm above is likely essentially best-possible. More precisely,  we show
using a reduction relying on $d$-detecting families that the parameter
dependence $k^{\mathcal{O}(k)}$ is likely optimal, as an algorithm with
dependence $k^{o(k)}$ would falsify the ETH.

\item Focusing on width-based parameters and starting from pathwidth, we offer
a fine-grained lower bound answering our main question regarding parameters $C$
and $\Delta$: We establish using a single reduction that for all fixed $C\ge
3$, $\Delta\ge 3$, and $\varepsilon>0$, obtaining an algorithm with complexity
$(C+1-\varepsilon)^{\mathrm{pw}}$ or
$\left(\left\lfloor\frac{\Delta}{2}\right\rfloor+1-\varepsilon\right)^{\mathrm{pw}}$
is both necessary and sufficient for refuting the $\mathrm{pw}$-SETH.

\item Turning to algorithms, we show that the bound given above is exactly
tight for both treewidth and pathwidth. For the case of $C$ this follows from a
known fast-convolution based algorithm with dependence $(C+1)^{\mathrm{tw}}$.
However, for the case of $\Delta$ the ``easy'' treewidth-based DP leads to an
algorithm with dependence
$\left(\left\lceil\frac{\Delta}{2}\right\rceil+1\right)^{\mathrm{tw}}$, leaving
a small gap for odd values of $\Delta$. We show how this can be improved by
employing a somewhat unusual modular counting trick, thus replacing the
$\left\lceil\frac{\Delta}{2}\right\rceil$ term by
$\left\lfloor\frac{\Delta}{2}\right\rfloor$ and obtaining tight fine-grained
complexity bounds under the $\mathrm{pw}$-SETH, for both treewidth and
pathwidth, for all fixed values of $C$ and $\Delta$.

\end{enumerate}

In addition to the above, we show numerous other tight or almost-tight bounds
for related parameters painting a complete picture of the parameterized
complexity of the problem.  In particular, we show that the ``correct''
dependence for parameter cutwidth is $2^{\mathrm{ctw}}$, as improving the base
of the exponential by $\varepsilon$ is again equivalent to falsifying the
$\mathrm{pw}$-SETH; while the problem cannot be solved in time
$n^{o(\mathrm{td})}$ parameterized by tree-depth under the ETH.  Furthermore,
we show that the problem is NP-complete for graphs of constant clique-width;
and that it admits an algorithm with running time
$\mathrm{vi}^{\mathcal{O}(\mathrm{vi}^{2})} n^{\mathcal{O}(1)}$ using $N$-fold
integer programming, where $\mathrm{vi}$ is the vertex integrity. We observe
that the latter may not be far from optimal, as a
$2^{\mathcal{O}(\mathrm{vc}^{2-\varepsilon})} n^{\mathcal{O}(1)}$ algorithm,
where $\mathrm{vc}$ is the vertex cover, would imply corresponding progress for
a broader class of integer-programming-type problems. Finally, we show that the
problem admits an EPTAS of running time
$(\frac{1}{\varepsilon})^{\mathcal{O}(\frac{1}{\varepsilon})}n^{\mathcal{O}(1)}$
on planar graphs.

\end{abstract}

% \newpage

\section{Introduction}\label{sec:introduction}

\CVC\ ({\CVCshort}) is the hard-capacitated variant of \textsc{Vertex Cover} where  the input
consists of a graph $G=(V,E)$, a capacity function $\capacity \colon V \to
\mathbb{N}$, and an integer $k$.  The task is to determine whether there exists
a set $S \subseteq V$ of size at most $k$ that covers all edges and such that
every selected vertex $v \in S$ is assigned at most $\capacity(v)$ incident
edges.  Equivalently, one may orient every edge towards the endpoint that
covers it and ask that each selected vertex has in-degree at most its capacity.

Capacitated covering problems are a very well-studied topic, particularly in
the context of approximation algorithms~\cite{siamcomp/ChuzhoyN06}. In this
paper, however, we mostly study \CVC\ through the lens of structural
parameterized complexity.  \CVC\ has historically played a prominent role in
this domain as it was one of the first natural problems known to be in XP but
shown to be W[1]-hard parameterized by treewidth, a result established in 2008
by Dom, Lokshtanov, Saurabh, and Villanger~\cite{iwpec/DomLSV08}.

To understand the context of the result of~\cite{iwpec/DomLSV08}, it is useful
to think how one could naturally try to solve this problem using treewidth. The
standard approach would be to perform dynamic programming where for each vertex
of the bag we keep track of how many of its incident edges this vertex has been
used to cover, or conversely how many have been covered by its neighbors. This
approach naturally leads to algorithms with complexities of the form
$C^{\bO(\tw)} $ or $\Delta^{\bO(\tw)}$ for the former approach, where
$C,\Delta$ are the maximum capacity and degree respectively, or $k^{\bO(\tw)}$
for the latter approach. Unfortunately, all three parameters $C,\Delta,k$ can
attain large values, meaning this simple approach can only lead to an XP
algorithm for treewidth. The main result of~\cite{iwpec/DomLSV08} was then to
show that this running time cannot be improved to FPT by treewidth.
Nevertheless, one bright spot of the above natural approach is that it directly
leads to an FPT algorithm for the natural parameter $k$, by using the fact that
we may assume that $\tw\le k$~\cite{mst/GuoNW07}.

\subparagraph{Main Results.} The question of whether the basic algorithmic
approach sketched above can be improved upon or avoided has been studied
repeatedly in the FPT literature (we review some previous works below). Here,
we revisit this topic starting from two basic questions which have remained open:

\begin{enumerate}

\item Can the $k^{\bO(k)}n^{\bO(1)}$ algorithm sketched above be improved?
Even though this algorithm establishes fixed-parameter tractability for the
natural parameter, the complexity compares rather unfavorably to the
uncapacitated version of the problem, which can of course be solved in
$2^{\bO(k)}n^{\bO(1)}$ time. It is therefore natural to ask whether such a
running time is achievable for \CVC.

\item What is the precise complexity of \CVC\ for small values of $C$ or
$\Delta$?  As observed above, \CVC\ does become tractable parameterized by
treewidth if we additionally assume that either $C$ or $\Delta$ are small. Can
we obtain an algorithm that is even a little faster than the natural dynamic
programs in this case?

\end{enumerate}

Our principal contribution in this paper is to give (mostly negative) answers
to these two questions. With respect to the natural parameter $k$, for which
\CVC\ admits slightly super-exponential algorithms on $k$
(\cite{iwpec/DomLSV08,mst/GuoNW07,sofsem/RooijR19}),  we show that this is
essentially unavoidable, as a $k^{o(k)} n^{\bO(1)}$ algorithm would refute the
ETH (\cref{thm:k:lb}).
Our proof uses $d$-detecting families via a technique introduced
into structural parameterized complexity by Lampis and
Vasilakis~\cite{toct/LampisV24}, building on earlier ideas of Bonamy, Kowalik,
Pilipczuk, Socala, and Wrochna~\cite{toct/BonamyKPSW19}.

With respect to the second question, we begin by taking a closer look at
algorithmic upper bounds for parameters $\tw,\pw,C$, and $\Delta$. On the one
hand, it is fairly straightforward to obtain an algorithm running in time
$(C+1)^{\pw}n^{\bO(1)}$ and fast subset convolution techniques allow this to be
extended to treewidth keeping the complexity
$(C+1)^{\tw}n^{\bO(1)}$~\cite{sofsem/RooijR19}. On the other hand, focusing on
parameter $\Delta$ the natural DP algorithm has table size
$\left(\left\lceil\frac{\Delta}{2}\right\rceil+1\right)^{\mathrm{tw}}$.
Intuitively, this is achieved by storing in the DP (i) for each low-capacity
vertex (for which $\capacity(v)\le \left\lfloor \frac{\Delta}{2}\right\rfloor$)
the number of edges $v$ has been used to cover (ii) for each high-capacity
vertex $v$ the number of edges incident on $v$ covered by another vertex.  Both
integers take at most $\left\lfloor \frac{\Delta}{2}\right\rfloor+1$ values,
but because for high-capacity vertices we need an extra state to remember if
the vertex is selected this scheme leads to tables of size
$\left(\left\lceil\frac{\Delta}{2}\right\rceil+1\right)^{\mathrm{tw}}$.

Having established these baseline algorithms, our main contribution for this
question is then to present a single lower bound argument showing that
complexities better than $(C+1)^{\pw}n^{\bO(1)}$ or
$\left(\left\lfloor\frac{\Delta}{2}\right\rfloor+1\right)^{\mathrm{pw}}$ are
likely to be impossible for all non-trivial\footnote{By non-trivial we mean
greater or equal to $3$.  The problem is obviously in P for $\Delta\le 2$ or
$C\le 1$, while we show it to also be in P for $C=2$.} fixed values of $C$ and
$\Delta$.  More precisely, we show that decreasing the base of the exponential
by an arbitrarily small $\varepsilon>0$ is impossible in both cases under the
$\pw$-SETH and in fact that obtaining an algorithm with dependence
$(C+1-\varepsilon)^{\pw}$ or
$\left(\left\lfloor\frac{\Delta}{2}\right\rfloor+1-\varepsilon\right)^{\mathrm{pw}}$
is possible \emph{if and only if} the $\pw$-SETH is false. This adds \CVC\ to
the growing list of problems for which improving the current best algorithm is
equivalent to disproving the
$\pw$-SETH~\cite{soda/Lampis25,esa/LampisL26,LampisV25}. Of course, since lower
bounds for pathwidth also apply to treewidth, this tightly characterizes the
complexity of the problem for both treewidth and pathwidth for all $C\ge 3$ and
for all even $\Delta\ge 4$.

The above discussion then leaves a small gap in the case of odd $\Delta$,
because the upper bound contains a $\left\lceil\frac{\Delta}{2}\right\rceil $
term while the lower bound a $\left\lfloor\frac{\Delta}{2}\right\rfloor $ term.
We close this gap by improving the DP algorithm sketched above using a somewhat
unusual modular-counting trick. More precisely, we observe that when $\Delta$
is odd the problematic case stems from vertices $v$ with $\deg(v)=\Delta$ and
$\capacity(v)=\frac{\Delta+1}{2}$, so we treat these vertices in a special way.
In particular, for such vertices $v$ we only keep track of the number of
incident edges covered by other vertices \emph{modulo} $\frac{\Delta+1}{2}$
(and do not store whether $v$ is selected). This ensures that there are $
\frac{\Delta+1}{2} = \left\lfloor \frac{\Delta}{2}\right\rfloor+1$ states per
vertex, as desired, but makes it more challenging to verify if we have a
correct solution. The key idea now is that if we assume that selected
problematic vertices must use their full capacity, the number of edges oriented
away from a problematic vertex $v$ will be either $\Delta$ (if $v$ is not
selected) or $\frac{\Delta-1}{2}$. However, $\Delta \bmod \frac{\Delta+1}{2} =
\frac{\Delta-1}{2}$, so our algorithm checks when we forget $v$ that the number
stored for $v$ is $\frac{\Delta-1}{2}$. This seems ill-advised as it merges the
state where $v$ is not selected with the state where $v$ is selected and
exhausts its capacity, but we are able to separate these two by keeping track
of a global budget. In particular, each time an edge incident on a problematic
vertex $v$ is covered by its other endpoint we obtain a profit of $1$, meaning
that each problematic $v$ contributes either $\Delta$ or $\frac{\Delta-1}{2}$
units of profit. By multiplying the number of problematic vertices by
$\frac{\Delta-1}{2}$ and subtracting the result from the total profit we can
infer how many problematic vertices are not selected (namely, the remaining
profit divided by $\frac{\Delta+1}{2}$). In the end this somewhat unorthodox
trick (combined with fast subset convolution methods) allows us to close the
remaining gap and obtain an algorithm of running time exactly matching our
lower bound. This completes our tight characterization for treewidth and
pathwidth and all remaining values of $\Delta$.

\subparagraph{Further Results.} Moving further from the above, we then consider
the complexity of the problem from several related angles. In particular, we
show the following results:

\begin{enumerate}

\item Our study of the parameterization by $\tw+\Delta$ motivates us to look at
structural parameters that indirectly bound the degree. We therefore consider
cutwidth, which is restrictive enough to recover exact fixed-parameter
tractability.  In fact, \CVC\ is already FPT for the more general
parameterization by tree-cut width~\cite{siamdm/GanianKS22}; our contribution
is to determine the correct single-exponential dependence for cutwidth.  We
give an algorithm running in time $2^{\ctw} n^{\bO(1)}$ when a linear
arrangement of cutwidth $\ctw$ is provided (\cref{cor:cutwidth:algorithm}).
More importantly, we show that we have the ``correct'' complexity, as improving
the running time to $(2-\varepsilon)^{\ctw}n^{\bO(1)}$ is once again
\emph{equivalent} to falsifying the $\pw$-SETH.

\item Having studied treewidth and pathwidth we also take another look at the
more restricted parameter tree-depth.  The reduction of~\cite{iwpec/DomLSV08}
already implies that an algorithm running in time $f(\td)\, n^{o(\sqrt{\td})}$
for any computable function $f$ would refute the ETH, where $\td$ denotes the
tree-depth of the input graph.  In \cref{thm:td:lb} we strengthen this to
$f(\td)\, n^{o(\td)}$, thereby rendering the standard $n^{\bO(\tw)}$ algorithm
essentially optimal, even under the more restrictive parameterization by
tree-depth.  Our reduction relies on the recursive construction
of~\cite{toct/LampisV24}, which allows us to appropriately adapt the reduction
of~\cite{iwpec/DomLSV08} while keeping the tree-depth linear.

\item A natural question going in the opposite direction is whether the XP
algorithms for treewidth can be extended to a more general parameter, such as
clique-width.  At a high level, algorithms parameterized by clique-width rely
on aggregating information over label classes.  For \CVC, however, capacities
and assigned demands interact at the level of individual vertices.  It is
therefore far from clear how one could summarize this information inside a
label class without losing essential structure.  We show that this obstacle is
genuine by proving that the problem remains NP-hard already on graphs of linear
clique-width equal to $6$ (\cref{thm:cw:np-hard}).

\item Considering again a more restricted parameter, we take a look at vertex
integrity.%
\footnote{The \emph{vertex integrity} of a graph $G$ is the minimum
value $k$ such that there exists a set $S \subseteq V(G)$ with every connected
component of $G-S$ having size at most $k - |S|$.  The vertex integrity of a
graph is upper-bounded by its vertex cover number plus 1.}
Gima, Hanaka, Kiyomi, Kobayashi, and Otachi~\cite{tcs/GimaHKKO22} showed that \CVC\ is FPT
for this parameter by formulating it as an integer linear program with a small
number of variables, yielding an algorithm with double-exponential dependence
on $\vi$. We exponentially improve over this algorithm and obtain one of
running time $\vi^{\bO(\vi^2)} n^{\bO(1)}$ based on the machinery of $N$-fold
Integer Programming. Complementing this, in \cref{thm:vc:lb} we give a simple
reduction showing evidence against a $2^{\bO(\vc^{2 - \varepsilon})}
n^{\bO(1)}$ algorithm, even for the more restricted parameter of
(uncapacitated) vertex cover. Indeed, we show that such an improvement would,
via the fine-grained equivalences of Rohwedder and
W\k{e}grzycki~\cite{innovations/RohwedderW25}, imply analogous improvements for
a broader class of integer-programming-type problems.  In particular, it would
yield a $2^{\bO(m^{2 - \varepsilon})} n^{\bO(1)}$ algorithm for \textsc{Integer
Linear Programming} with bounded entries, parameterized by the number $m$ of
constraints, when each variable $x_i$ is additionally subject to bounds $\ell_i
\le x_i \le r_i$ that do not count towards $m$.  Whether such an algorithm
exists is a major open
question~\cite{mor/JansenR23,toct/KnopPW20,esa/RandolphW24}.  To the best of
our knowledge, this is the first barrier of this type for a \emph{structural}
parameterization; previous barriers were known only for \emph{natural}
parameterizations, for example \textsc{Closest String} parameterized by the
number of strings or {\SMC} parameterized by the size of the universe. Hence,
we give evidence that for both vertex integrity and vertex cover the correct
complexity of \CVC\ is roughly  ``exponential in the square'' of the parameter.

\item Finally, we show that \CVC\ admits an EPTAS on planar graphs, that is, a
$(1+\varepsilon)$-approximation algorithm running in time
$(\frac{1}{\varepsilon})^{\bO(1/\varepsilon)}n^{\bO(1)}$. In order to obtain
this result we need to combine several non-trivial ideas, building upon the
standard technique of Baker~\cite{jacm/Baker94}, its extension by Demaine and
Hajiaghayi~\cite{soda/DemaineH05}, but also using the EPTAS of
Becker~\cite{waoa/Becker17} which works when $C$ is bounded, and the FPT-AS for
\CVC\ parameterized by treewidth given in \cite{isaac/ChuL23}.

\end{enumerate}

\subparagraph{Related Work.} \CVC\ is well-studied in the approximation
setting, due in part to its applications \cite{jal/GuhaHKO03}.  Our focus in
this paper is exact complexity in the standard \emph{hard-capacitated} setting,
where each vertex may be selected at most once.  This should be contrasted with
the \emph{soft-capacitated} variant, in which a vertex may be selected multiple
times and each copy contributes the same capacity.  This distinction is central
in the approximation literature.  For soft capacities, Guha, Hassin, Khuller,
and Or~\cite{jal/GuhaHKO03} introduced \CVC\ and gave a $2$-approximation,
while Bar-Yehuda, Flysher, Mestre, and Rawitz~\cite{siamdm/Bar-YehudaFMR10}
studied the partial version.  For hard capacities, Chuzhoy and
Naor~\cite{siamcomp/ChuzhoyN06} initiated the approximation study, Gandhi,
Halperin, Khuller, Kortsarz, and Srinivasan~\cite{jcss/GandhiHKKS06} improved
the approximation ratio to $2$ in the unweighted case, and Grandoni, Koenemann,
Panconesi, and Sozio~\cite{siamcomp/GrandoniKPS08} gave a bicriteria
distributed algorithm.  The problem has also been studied with respect to its
approximability on planar graphs~\cite{waoa/Becker17}, albeit only obtaining an
EPTAS for bounded capacities -- we generalize this result by removing this
assumption. Later work has also obtained tight
approximation guarantees for multigraphs and
hypergraphs~\cite{soda/CheungGW14,algorithmica/Kao21,tcs/KaoSLL19,soda/Wong17}.
We also mention the recent work of Lokshtanov, Sahu, Saurabh, Surianarayanan,
and Xue~\cite{soda/LokshtanovS0S025}, which develops parameterized
approximation algorithms for weighted \textsc{Capacitated $d$-Hitting Set} with
hard capacities.

As mentioned, the problem is also well-studied in the (structurally)
parameterized setting. The W[1]-hardness by treewidth was recently refined to
XNLP-completeness for pathwidth~\cite{algorithmica/BodlaenderGJJL25} and to
XALP-completeness for either treewidth~\cite{iwpec/BodlaenderGJPP22} or
outerplanarity~\cite{dam/BodlaenderS26}.  In addition, the problem admits
efficient FPT approximation schemes parameterized by
treewidth~\cite{isaac/ChuL23,icalp/Lampis14}.

\subparagraph{Disclosure on the Use of AI Tools.}
We used OpenAI Codex in conducting and writing this research.
The authors had already obtained the simple DP sketched at the start of \cref{sec:degree}
and \cref{lem:degree:pathwidth}.
The modular-counting idea behind \cref{thm:degree:treewidth-algorithm}
originated from Codex during interactive exchanges with the authors.
Codex also suggested adapting ideas from~\cite{soda/DemaineH05} to obtain the PTAS
in \cref{sec:planar}, which the authors subsequently improved to an EPTAS.
We also used Codex for literature search, drafting and revising the
text, exploring possible approaches, developing proof arguments,
and checking proofs for gaps. The authors reviewed all definitions,
statements, and proofs, and simplified and improved the arguments.
They take full responsibility for the correctness and integrity of
the paper, including all content developed with the assistance of AI.

\section{Preliminaries}\label{sec:preliminaries}
Throughout the paper we use standard graph notation~\cite{books/Diestel25},
and we assume familiarity with the basic notions of parameterized complexity~\cite{books/CyganFKLMPPS15}.
All graphs considered are undirected without loops, unless explicitly stated otherwise.
Let $\N$ denote the set of non-negative integers.
For $x, y \in \Z$, let $[x, y] = \setdef{z \in \Z}{x \leq z \leq y}$,
while $[x] = [1,x]$.
Standard $\sO(\cdot)$ notation is used to suppress polynomial factors.
Due to space constraints, the proofs of statements marked with $(\appsymb)$ are deferred to the appendix.

We also use an equivalent orientation formulation of \CVC.
Let $G = (V,E)$ be the input graph with capacity function $\capacity \colon V \to \N$.
An \emph{orientation} of $G$ is a directed graph obtained by replacing each edge $\{u,v\} \in E$ with exactly one of the arcs $(u,v)$ or $(v,u)$.
For an orientation $O$ of $G$, let $\indeg_O(v)$ and $\outdeg_O(v)$ denote the in-degree and out-degree
of a vertex $v \in V$ in $O$ respectively.
We say that an orientation $O$ is \emph{feasible} if $\indeg_O(v) \leq \capacity(v)$ for all $v \in V$.
The \emph{size} of an orientation $O$, denoted by $\size(O)$, is defined as $|\setdef{v \in V}{\indeg_O(v) > 0}|$.
In this formulation, {\CVC} asks for a feasible orientation of minimum size.
We use both formulations interchangeably.

\subparagraph{Treewidth and pathwidth.}
A \emph{tree decomposition} of $G$ consists of a tree $T$ and bags
$X_t\subseteq V(G)$ for $t\in V(T)$, such that every edge has both
endpoints in some bag and, for every vertex $v$, the nodes whose bags
contain $v$ induce a nonempty connected subtree of $T$.
Its \emph{width} is $\max_{t\in V(T)}|X_t|-1$, and the
\emph{treewidth} $\tw(G)$ is the minimum width of a tree decomposition
of $G$. A \emph{path decomposition} is a tree decomposition in which
$T$ is a path; the minimum width of such a decomposition is the
\emph{pathwidth} $\pw(G)$.

We also use \emph{nice} decompositions, rooted with empty root and leaf bags
and at most two children per node.
A node with one child either adds or removes one vertex from its child's bag
(an \emph{introduce-vertex} or \emph{forget-vertex} node), or keeps the bag
and introduces an edge of $G$ between two of its vertices
(an \emph{introduce-edge} node). A node with two children is a \emph{join}
node, and both children have the same bag as their parent.
Every edge is introduced exactly once.
Any given decomposition can be made nice in polynomial time without increasing
its width; nice path decompositions have no join nodes.

\subparagraph{pw-SETH.} The \emph{primal pathwidth SETH} (pw-SETH) states that,
for every $\varepsilon>0$, there is no algorithm solving 3-SAT on $n$
variables in time $(2-\varepsilon)^{\pw}n^{\bO(1)}$, even when a path
decomposition of the formula's primal graph of width $\pw$ is supplied.%
\footnote{The \emph{primal} (or \emph{Gaifman}) graph of a SAT or CSP instance
has one vertex per variable, with two variables adjacent when they appear
in a common clause or constraint.}
Thus, pw-SETH asserts that the standard $\sO(2^{\pw})$ DP algorithm
for 3-SAT cannot be improved in its exponential base.
It is \emph{equivalent} to many tight lower bounds for problems parameterized
by pathwidth or linear clique-width~\cite{soda/Lampis25}, and is implied by
the Set Cover Conjecture and the SETH for circuits of depth
$\varepsilon n$~\cite{soda/Lampis26a}.

We also use a consequence for \textsc{MaxW-CSP}.
Given a CSP over a domain $D$ of size $B\ge2$, a weight function
$w\colon D\to\N$, and a target $W$, this problem asks for a satisfying
assignment whose total weight, summed over all variables, is at least $W$.
Standard DP takes time $\sO(B^{\pw})$ given a path decomposition of the
primal graph of width $\pw$.
For fixed $D$, $w$, and an upper bound on the number of variables per constraint,
Lampis~\cite[Theorems~5~and~11]{soda/Lampis25} shows that if the pw-SETH
is false, this running time can be improved to
$\sO((B-\varepsilon)^{\pw})$ for some $\varepsilon>0$.

\subparagraph{$N$-fold IP.}
An \emph{$N$-fold integer program}~\cite{mfcs/BlazejKPS24,disopt/GavenciakKK22,swat/KnopMV26}
minimizes an objective over $N$ blocks, called \emph{bricks}, of bounded
integer variables $x^{(i)}\in\Z^{t_i}$, subject to
\[
    \sum_{i\in[N]}D_i x^{(i)}=\mathbf{b}_0,\qquad
    A_i x^{(i)}=\mathbf{b}_i\quad(i\in[N]).
\]
The matrices $D_i\in\Z^{r\times t_i}$ specify the $r$ \emph{linking constraints},
while $A_i\in\Z^{s_i\times t_i}$ specifies the \emph{local constraints} of brick $i$;
inequalities are also allowed. Set
\[
    s=\max_{i\in[N]}s_i,\qquad
    t=\max_{i\in[N]}t_i,\qquad
    d=\sum_{i\in[N]}t_i,\qquad
    a=rs\max_{i\in[N]}\{\|D_i\|_\infty,\|A_i\|_\infty\},
\]
where $\|\cdot\|_\infty$ denotes the largest absolute entry of a matrix.

\begin{proposition}[{\cite[Corollary~97]{corr/EisenbrandHKKLO19}}]\label{prop:nfold}
    An $N$-fold IP with a linear objective can be solved in time
    $a^{\bO(r^2s+rs^2)}(dL)^{\bO(1)}$, where $L$ is the maximum bit length
    of any numerical input.
\end{proposition}

\section{Natural Parameter}\label{sec:natural_parameter}

In this section we consider {\CVC} under its natural parameterization.
We show that, unless the ETH fails, there is no $k^{o(k)} n^{\bO(1)}$-time algorithm for \CVC, where $k$ denotes the solution size.
To show this, we give a reduction from a variant of \OneInThreeSAT, called \SD, which is defined as follows.

\problemdef{\SD}
{A CNF formula $\psi$, every clause of which contains exactly 3 distinct variables and each variable appears in at most 4 clauses.}
{Determine whether there exists an assignment to the variables of $\psi$ such that each clause has exactly one True literal.}

It is known that there is no $2^{o(n)}$ time algorithm for \SD, where $n$ is the number of variables of the input formula, unless the ETH fails~\cite[Theorem~46]{toct/LampisV24}.
Our reduction relies on two ingredients from~\cite{toct/LampisV24}.
The first is a grouping lemma that partitions the variables and clauses of the formula into sufficiently small groups, while ensuring that within each clause group there is at most one occurrence of a variable from any fixed variable group.
This lets us encode, for each variable group, a choice of a partial assignment and then reason about its effect on an entire clause group at once.
The second ingredient is a compression step based on $d$-detecting families:
rather than testing each clause individually for satisfaction by a single literal,
we will test only a small set of carefully chosen subsets of clauses instead.

\begin{definition}
    Let $U$ be a finite set and let $d \in \N$.
    A family $\mathcal{F} \subseteq 2^U$ is a \emph{$d$-detecting family} for $U$ if, for every two distinct functions
    $f,g \colon U \to \{0,\ldots,d-1\}$, there exists a set $S \in \mathcal{F}$ such that
    $\sum_{x \in S} f(x) \neq \sum_{x \in S} g(x)$.
\end{definition}

In other words, the vector of subset sums over $\mathcal{F}$ uniquely determines any function $U \to \{0,\ldots,d-1\}$.
Thus, one may recover per-element information by checking only aggregated sums over the sets of $\mathcal{F}$.
We use the following result due to Lindstr\"om~\cite{cmb/Lindstrom65}.

\begin{theorem}[{\cite{cmb/Lindstrom65}}]\label{thm:lindstrom}
    For every constant $d \in \N$ and finite set $U$,
    one can construct in time polynomial in $|U|$
    a $d$-detecting family $\mathcal{F}$ on $U$ of size
    $\frac{2 |U|}{\log_d |U|} \cdot (1 + o(1))$.
\end{theorem}

In our reduction, the relevant function assigns to each clause the number of its literals that are made true by the selected partial assignments, which is a value in $\{0,1,2,3\}$.
Hence we will use $4$-detecting families.
Combined with the grouping lemma of~\cite[Lemma~47]{toct/LampisV24}, \Cref{thm:lindstrom} allows us to replace per-clause checks by only $\bO(\sqrt{n}/\log n)$ aggregate tests for each clause group.
The proof below follows this plan: for each variable group we select one partial assignment using a choice gadget, and for each clause group we use a small detecting family to verify, through aggregate constraints, that every clause receives exactly one true literal.

% \manolis{Reviewer: ``precise statement of the grouping lemma should be included in the paper''}

\begin{theorem}\label{thm:k:lb}
    There is no $k^{o(k)} n^{\bO(1)}$-time algorithm for \CVC,
    unless the ETH fails.
\end{theorem}

\begin{proof}
    Let $\psi$ be an instance of {\SD} of $n$ variables.
    Assume without loss of generality that $n$ is a power of $4$
    (this can be achieved by adding dummy variables to the instance if needed).
    We begin by applying the grouping lemma of~\cite[Lemma~47]{toct/LampisV24}.
    In time $n^{\bO(1)}$, we obtain:
    \begin{itemize}
        \item a partition of $\psi$'s variables into $n_V = \bO(n / \log n)$ subsets
         $V_1, \ldots, V_{n_V}$, where for all $p \in [n_V]$, $|V_p| \leq \log n$, and
        \item a partition of $\psi$'s clauses into $n_C = \bO(\sqrt{n})$ subsets
        $C_1, \ldots, C_{n_C}$, where for all $i \in [n_C]$, $|C_i| \leq \sqrt{n}$,
    \end{itemize}
    such that for every pair $(V_p,C_i)$, at most one variable of $V_p$ appears in the clauses of $C_i$,
    and if such a variable exists, then it appears only once in $C_i$.

    Next, for every clause group $C_i$, we apply \Cref{thm:lindstrom} with $d=4$ and $U=C_i$.
    Since $|C_i| \leq \sqrt{n}$, this yields in polynomial time a $4$-detecting family
    $\{ C_{i,1}, \ldots, C_{i,s_i} \}$ of subsets of $C_i$,
    where $s_i = \bO\parens*{\frac{|C_i|}{\log |C_i|}} = \bO\parens*{\frac{\sqrt{n}}{\log n}}$.

    Define $k = 2 n_V + 2 \sum_{i \in [n_C]} s_i$.
    We will construct a capacitated graph $G = (V, E)$ with a capacity function $\capacity \colon V \to \mathbb{N}$
    such that $(G, \capacity, k)$ is a Yes instance of {\CVC} if and only if $\psi$ is a Yes instance of \SD.
    Intuitively, the construction has two parts:
    for every variable group $V_p$ we create a gadget that forces the solution to choose exactly one partial assignment of $V_p$,
    and for every set of the detecting family we create a pair of vertices that checks the corresponding aggregate constraint.

    \proofsubparagraph*{Choice Gadget.}
    For each variable subset $V_p$ we introduce a vertex $u_p$,
    as well as vertices $\mathcal{V}_p = \setdef{v^p_q}{q \in [n]}$.
    We add edges $\{u_p, v^p_q\}$ for all $q \in [n]$.
    We further attach $k+1$ unnamed leaves to $u_p$.
    We set the capacity of $u_p$ to $\capacity(u_p) = \deg_{G} (u_p) - 1$,
    and the capacity of each $v^p_q$ to $\capacity(v^p_q) = \deg_{G} (v^p_q)$.
    We fix an arbitrary one-to-one mapping so that every vertex of $\mathcal{V}_p$ corresponds to a different assignment for the variables of $V_p$.
    Since $2^{|V_p|} \leq n$, there are enough vertices to encode all assignments of $V_p$.
    Let $\mathcal{V} = \mathcal{V}_1 \cup \ldots \cup \mathcal{V}_{n_V}$ denote the set of all such vertices.
    The role of $u_p$ is to force the solution to choose one vertex of $\mathcal{V}_p$.

    \proofsubparagraph*{Clause Gadget.}
    For $i \in [n_C]$, let $C_i$ be a clause subset and $\{ C_{i, 1}, \ldots, C_{i, s_i} \}$ be its 4-detecting family.
    For every set $C_{i,j}$ in this detecting family we introduce vertices $a_{i,j}$ and $a'_{i,j}$,
    and attach to each of them $k+1$ unnamed leaves.
    We connect $a_{i,j}$ to a vertex $v^p_q \in \mathcal{V}_p$ if the partial assignment represented by $v^p_q$
    satisfies the unique clause of $C_{i,j}$ that contains a variable of $V_p$.
    By the grouping lemma, each clause of $C_{i,j}$ contributes exactly three such variables, all from different variable groups.
    Hence there are exactly $3|C_{i,j}|$ relevant variable-group/clause incidences, and for each of them exactly half of the assignments of the corresponding variable group satisfy the clause.
    Therefore $a_{i,j}$ has exactly $|C_{i,j}| \cdot \frac{3n}{2}$ neighbors.
    Set the capacity of $a_{i,j}$ to $\capacity(a_{i,j}) = \deg_G (a_{i,j}) - |C_{i,j}|$.
    As for $a'_{i,j}$, we make its neighborhood on $\mathcal{V}$ complementary to that of $a_{i,j}$:
    for every $v \in \mathcal{V}$, we add the edge $\{ a'_{i,j}, v \}$ if $v \notin N(a_{i,j})$.
    Notice that $N(a_{i,j}) \cup N(a'_{i,j}) = \mathcal{V}$,
    while $N(a_{i,j}) \cap N(a'_{i,j}) = \varnothing$.
    Lastly, set the capacity of $a'_{i,j}$ to $\capacity(a'_{i,j}) = \deg_G (a'_{i,j}) - (n_V - |C_{i,j}|)$.

    This completes the construction of $G$.
    Let $\mathcal{I} = (G, \capacity, k)$ be the constructed instance of \CVC.

    \begin{claim}\label{claim:cvc_lb:sat->cvc}
        If $\psi$ is a Yes instance of \SD, then $\mathcal{I}$ is a Yes instance of \CVC.
    \end{claim}

    \begin{claimproof}
        Let $f \colon V_1 \cup \ldots \cup V_{n_V} \to \{ \true, \false \}$ be an assignment of the variables of $\psi$
        so that all of its clauses have exactly 1 True literal.
        Let $S \subseteq V(G)$ contain from each $\mathcal{V}_p$ the vertex corresponding to this assignment
        restricted to $V_p$; notice that this implies that $|S \cap \mathcal{V}_p| = 1$ for all $p \in [n_V]$.
        We further add into $S$ all vertices $u_p$, for $p \in [n_V]$, as well as all vertices $a_{i,j}$ and $a'_{i,j}$
        introduced due to the clause subsets.
        It holds that $|S| = 2 n_V + 2 \sum_{i \in [n_C]} s_i = k$.

        We argue that $S$ is a capacitated vertex cover of $G$.
        First, every chosen vertex of $S \cap \mathcal{V}$ has capacity equal to its degree,
        so we may orient all its incident edges towards it.
        Next, since $\capacity(u_p) = \deg_G (u_p) - 1$ for all $p \in [n_V]$,
        the only edge incident with $u_p$ not covered by $u_p$ is the one to the selected vertex of $\mathcal{V}_p$;
        all edges between $u_p$ and $\mathcal{V}_p \setminus S$ are covered by $u_p$.

        It remains to check the vertices $a_{i,j}$ and $a'_{i,j}$.
        We claim that for every set $C_{i,j}$, the chosen vertices $S \cap \mathcal{V}$ contain exactly $|C_{i,j}|$ neighbors of $a_{i,j}$.
        Indeed, every clause of $C_{i,j}$ has exactly one true literal under $f$,
        and by construction the chosen vertex of each variable group contributes to $N(a_{i,j})$
        exactly when its partial assignment satisfies the unique clause of $C_{i,j}$ involving that group.
        Hence each clause of $C_{i,j}$ contributes exactly one selected neighbor of $a_{i,j}$, and therefore
        $|(S \cap \mathcal{V}) \cap N_G(a_{i,j})| = |C_{i,j}|$.
        Since $|S \cap \mathcal{V}| = n_V$ and sets $N(a_{i,j}), N(a'_{i,j})$ partition $\mathcal{V}$,
        this also gives $|(S \cap \mathcal{V}) \cap N_G(a'_{i,j})| = n_V - |C_{i,j}|$.
        Thus exactly $|C_{i,j}|$ edges incident with $a_{i,j}$ and exactly $n_V-|C_{i,j}|$ edges incident with $a'_{i,j}$
        are covered by vertices of $S \cap \mathcal{V}$, and the remaining incident edges are within the capacities of $a_{i,j}$ and $a'_{i,j}$ by construction.
        This completes the proof.
    \end{claimproof}

    \begin{claim}\label{claim:cvc_lb:cvc->sat}
        If $\mathcal{I}$ is a Yes instance of \CVC, then $\psi$ is a Yes instance of \SD.
    \end{claim}

    \begin{claimproof}
        Let $S \subseteq V(G)$ be a capacitated vertex cover of $G$, with $|S| \leq k$.
        First, every vertex $u_p$, $a_{i,j}$, and $a'_{i,j}$ has $k+1$ pendant neighbors,
        so all these vertices must belong to $S$.
        There are exactly $k-n_V$ such vertices.
        Moreover, since $\capacity(u_p) = \deg_G(u_p)-1$, the set $S$ must contain at least one neighbor of each $u_p$.
        Because the neighborhoods $N(u_p)$ are pairwise disjoint and $|S| \leq k$,
        it follows that $S$ contains exactly one neighbor of each $u_p$.
        We may assume that this chosen neighbor lies in $\mathcal{V}_p$ rather than being a leaf,
        since replacing a chosen leaf by a vertex of $\mathcal{V}_p$ preserves feasibility.

        Now consider the assignment $f \colon V_1 \cup \ldots \cup V_{n_V} \to \{ \true, \false \}$
        of the variables of $\psi$ such that the vertex of $S \cap \mathcal{V}_p$ corresponds to the
        assignment $f$ restricted to the variables of $V_p$.
        Let $C_i$, where $i \in [n_C]$, be a clause subset resulting from the
        partition due to~\cite[Lemma~47]{toct/LampisV24}.
        Let $h \colon C_i \to \{ 0,1,2,3 \}$ map each clause $c \in C_i$ to the number of selected partial assignments that satisfy $c$.
        Since every clause has three literals and we selected one assignment from each variable group, indeed $h(c) \in \{0,1,2,3\}$.

        We argue that $\sum_{c \in C_{i,j}} h(c) = |C_{i,j}|$, for all $j \in [s_i]$.
        Indeed, since $\capacity(a_{i,j}) = \deg_G(a_{i,j}) - |C_{i,j}|$,
        at least $|C_{i,j}|$ edges incident with $a_{i,j}$ must be covered by vertices of $S \cap \mathcal{V}$.
        Likewise, since $\capacity(a'_{i,j}) = \deg_G(a'_{i,j}) - (n_V - |C_{i,j}|)$,
        at least $n_V - |C_{i,j}|$ edges incident with $a'_{i,j}$ must be covered by vertices of $S \cap \mathcal{V}$.
        Since $N(a_{i,j})$ and $N(a'_{i,j})$ partition $\mathcal{V}$,
        these lower bounds must in fact be tight,
        so $|(S \cap \mathcal{V}) \cap N_G(a_{i,j})| = |C_{i,j}|$.
        By construction, each selected neighbor of $a_{i,j}$ corresponds to one unit contributed to the sum $\sum_{c \in C_{i,j}} h(c)$.
        Therefore $\sum_{c \in C_{i,j}} h(c) = |C_{i,j}|$.

        Let $g \colon C_i \to \{0,1,2,3\}$ be the constant function $g \equiv 1$.
        Notice that $\sum_{c \in C_{i,j}} h(c) = \sum_{c \in C_{i,j}} g(c)$, for all $j \in [s_i]$.
        Since $\{ C_{i,1}, \ldots, C_{i, s_i} \}$ is a 4-detecting family, this implies that $h \equiv g$.
        Thus $h(c)=1$ for every clause $c \in C_i$, meaning that exactly one selected partial assignment satisfies $c$.
        Since $i \in [n_C]$ was arbitrary, every clause of $\psi$ has exactly one true literal under $f$.
    \end{claimproof}
    Correctness follows by \Cref{claim:cvc_lb:sat->cvc,claim:cvc_lb:cvc->sat}.
    Finally, assuming there exists a $k^{o(k)} n^{\bO(1)}$ algorithm for \CVC,
    one could decide {\SD} in time
    \[
        k^{o(k)} n^{\bO(1)} = \left( \frac{n}{\log n} \right)^{o(n / \log n)} n^{\bO(1)} =
        2^{(\log n - \log \log n) o(n / \log n)} =
        2^{o(n)},
    \]
    which contradicts the ETH due to~\cite[Theorem~46]{toct/LampisV24}.
\end{proof}

\section{pw-SETH Equivalences}\label{sec:seth}

We establish tight running-time bounds for {\CVC} parameterized by pathwidth
and either the maximum degree $\Delta$ or the maximum capacity $C$.
The degree algorithm in \cref{thm:degree:treewidth-algorithm} runs in time
$\sO\bigl((\floor{\Delta/2}+1)^{\tw}\bigr)$, while the capacity algorithm
of van Rooij and van Rooij~\cite[Theorem~4]{sofsem/RooijR19} runs in time
$\sO\bigl((C+1)^{\tw}\bigr)$.

The main result of this section is \cref{thm:pw-seth:equivalences}, which states that improving over either
of those bases is \emph{equivalent} to falsifying the $\pw$-SETH~\cite{soda/Lampis25}:
the statement holds for every fixed $C \ge 3$ and $\Delta \ge 3$, even when a path
decomposition is supplied.

We start with the case $C \le 2$, which we prove to be polynomial-time solvable via a reduction to \textsc{Maximum-Weight Matching}.
Notice that since $C \le \Delta$, this also shows polynomial-time solvability for $\Delta \le 2$.

\begin{theoremrep}[\appsymb]\label{thm:capacity:two}
    {\CVC} is polynomial-time solvable when every vertex has capacity at most $2$.
\end{theoremrep}

\begin{proof}
    Let $(G,\capacity,k)$ be the input instance, where $G=(V,E)$, $n=|V|$, $m=|E|$, and $k \le n$.
    Set $B=n+1$ and construct an edge-weighted graph $H$ as follows.
    For every vertex $v$ of capacity $1$, introduce one \emph{slot} $v'$
    and a private dummy vertex $d_v$, joined by a \emph{closing edge}.
    For every vertex $v$ of capacity $2$, introduce two slots $v'_1,v'_2$
    joined by a closing edge. Give every closing edge weight $1$.
    For every edge $e=uv$, introduce a vertex $x_e$ and join it to every slot
    of $u$ and $v$ by edges of weight $B$.

    We show that $G$ has a capacitated vertex cover of size at most $k$ if and
    only if $H$ has a matching of weight at least $Bm+n-k$.
    Given a feasible orientation $O$ of size at most $k$, match each $x_e$ to a distinct slot of
    the head of $e$. The capacities ensure that this is possible.
    Add the closing edge of each vertex of in-degree zero, whose slots are
    all unused. The resulting matching has weight $Bm+n-\size(O)\geq Bm+n-k$.

    Conversely, let $M$ be a matching of weight at least $Bm+n-k$.
    If some $x_e$ were unmatched, $M$ would have weight at most
    \[
        B(m-1)+n=Bm-1<Bm+n-k,
    \]
    a contradiction. Thus every $x_e$ is matched, and at least $n-k$ closing
    edges belong to $M$. Orient each original edge $e$ toward the endpoint
    whose slot is matched to $x_e$. No slot is used twice, so the orientation
    respects the capacities. A vertex whose closing edge belongs to $M$ has
    no used slot and hence in-degree zero. The orientation therefore has
    size at most $k$.

    The graph $H$ has $\bO(n+m)$ vertices and edges, with weights at most $n+1$.
    The result follows from the polynomial-time algorithm for
    \textsc{Maximum-Weight Matching}~\cite[Section~11]{books/KorteV18}.
\end{proof}

We now state the main theorem of this section.

\begin{theorem}\label{thm:pw-seth:equivalences}
    For all fixed $C\ge 3, \Delta\ge 3$, the three following statements are equivalent:
    \begin{itemize}
        \item There exists $\varepsilon>0$ such that {\CVC} with capacities at most $C$ can be solved in time
            $\sO\bigl((C+1-\varepsilon)^\pw\bigr)$.

        \item There exists $\varepsilon>0$ such that {\CVC} on graphs of maximum degree at most $\Delta$ can be solved in time
            $\sO\bigl((\lfloor \Delta/2 \rfloor +1-\varepsilon)^\pw\bigr)$.

        \item The $\pw$-SETH is false.
    \end{itemize}
\end{theorem}

If the $\pw$-SETH is false, the two faster algorithms follow from
\cref{lem:capacity:pathwidth-speedup,lem:degree:pathwidth}, with the latter being shown in \cref{sec:degree}.
The converse implications follow from
\cref{lem:pw-seth:csp-to-cvc,lem:pw-seth:degree-three}.

\begin{lemmarep}[\appsymb]\label{lem:capacity:pathwidth-speedup}
    If the $\pw$-SETH is false, then for every fixed $C\geq3$ there is an
    $\varepsilon>0$ such that {\CVC} with capacities at most $C$ can be solved
    in time $\sO\bigl((C+1-\varepsilon)^\pw\bigr)$.
\end{lemmarep}

\begin{proof}
    Fix $C \ge 3$ and suppose that the $\pw$-SETH is false.
    Let $(G,\capacity,k)$ be the input instance, where $G=(V,E)$, $n=|V|$,
    and $m=|E|$, and let $\mathcal{B}$ be a path decomposition of $G$ of width $\pw$.
    We reduce to an instance $\psi$ of \textsc{MaxW-CSP} of primal pathwidth $\pw + \bO(1)$,
    alphabet size $C+1$, and constant arity.

    \proofsubparagraph{Construction.}
    Make $\mathcal{B}$ nice, and assign each edge to a distinct bag containing its endpoints
    that does not introduce or forget any vertices, potentially duplicating bags in the process.
    Order the edges by their assigned bags. For each vertex $v$, let $e_{v,1},\ldots,e_{v,d_v}$ be its
    incident edges in this order, where $d_v=\deg_G(v)$.

    Construct an instance $\psi$ over the alphabet $\{0,\ldots,C\}$, with variables
    $x_{v,0},\ldots,x_{v,d_v}$ for each vertex $v$, which we call \emph{counters}.
    The counter $x_{v,i}$ records how many of the first $i$ incident edges are oriented toward $v$.
    Impose
    \[
        x_{v,0}=0
        \qquad\text{and}\qquad
        x_{v,d_v}\leq\capacity(v).
    \]
    For each edge $uv=e_{u,i}=e_{v,j}$, impose
    \[
        (x_{u,i}-x_{u,i-1},\,x_{v,j}-x_{v,j-1})\in\{(1,0),(0,1)\}.
    \]
    Thus each edge increases exactly one of its endpoint counters.

    Set $w(0)=1$ and $w(a)=0$ for $a \in \{1, \ldots, C\}$.
    We want only the final counters to affect the objective.
    To this end, each \emph{nonfinal} counter $x_{v,i}$, $i<d_v$,
    receives a private \emph{companion} $\bar{x}_{v,i}$, constrained so that
    exactly one of the two equals zero. Each such pair contributes weight one,
    regardless of the counter value. There are $\sum_v d_v=2m$ pairs; the $n$
    final counters remain unpaired.
    Set the target weight to $W=2m+n-k$.
    This completes the construction of $\psi$.

    \begin{claim}
        The primal graph of $\psi$ has pathwidth at most $\pw+2$.
    \end{claim}

    \begin{claimproof}
        Sweep through $\mathcal{B}$, keeping one current counter
        for each vertex in the current bag. Introduce $x_{v,0}$ when $v$ first
        appears. At the assigned bag of $uv=e_{u,i}=e_{v,j}$, the current counters
        are $x_{u,i-1},x_{v,j-1}$. Introduce $x_{u,i},x_{v,j}$, cover the update
        constraint, and forget the two old counters. Forget $x_{v,d_v}$ when $v$ is forgotten.

        Immediately before an update, introduce and forget the companions of
        its two old counters one at a time, in separate bags.
        Every constraint is covered and every counter occurs in an interval.
        The maintained bag has at most $\pw+1$ counters; an update adds two
        variables, and a companion step adds one. The width is therefore at
        most $\pw+2$.
    \end{claimproof}

    \begin{claim}
        $G$ has a capacitated vertex cover of size at most $k$ if and only if
        $\psi$ has a satisfying assignment of weight at least $W$.
    \end{claim}

    \begin{claimproof}
        Every satisfying assignment $\sigma$ defines an orientation $O$: orient each edge toward
        the endpoint whose counter increases. The update constraints give
        $x_{v,d_v}=\indeg_O(v)$, so the final bounds make $O$ feasible.
        Conversely, every feasible orientation supplies the counter values;
        choose each companion to be $1$ if its counter is zero, and $0$ otherwise.
        In either direction, the assignment has weight
        \[
            \operatorname{wt}(\sigma)
            =2m+\bigl|\setdef{v\in V}{\indeg_O(v)=0}\bigr|
            =2m+n-\size(O).
        \]
        Thus $\operatorname{wt}(\sigma)\geq W$ if and only if $\size(O)\leq k$.
    \end{claimproof}

    There are $n+2m$ counters and $2m$ companions, and the constraints have arity
    at most four. Since $C$ is fixed, $\psi$ has polynomial size and a fixed
    weight function. The claimed running time follows from
    Lampis~\cite[Theorem~11]{soda/Lampis25}.
\end{proof}

For the converse, we use the following results from~\cite{soda/Lampis25}.

\begin{definition}[{\cite[Definition~15]{soda/Lampis25}}]
    Let $\psi$ be a \textnormal{\textsc{CSP}} instance over $[B]$, with variable
    set $X$, a path decomposition $(B_1,\ldots,B_\tau)$ of its primal graph,
    and an injective map $b$ assigning each constraint to a bag containing its scope.
    A \emph{multi-assignment} assigns a value $\sigma(x,t)\in[B]$ to each
    occurrence $x\in B_t$. It is \emph{satisfying} if every constraint $c$ is
    satisfied by the values $\sigma(x,b(c))$ on its scope.
    It is \emph{monotone} if, for each variable $x$, its values are nondecreasing
    along the decomposition, and \emph{consistent on $x$} if these values are constant.
\end{definition}

\begin{corollary}[{\cite[proof of Corollary~18]{soda/Lampis25}}]\label{cor:pw-seth:monotone-csp}
    Fix $\varepsilon>0$ and $B \ge 3$.  Consider the following promise problem.  The input is
    a $4$-\textnormal{\textsc{CSP}} instance $\psi$ over an alphabet of size $B$, a partition of its
    variables into $V_1,V_2$, a width-$p$ path decomposition of its primal graph in which
    every bag contains $\bO(B\log p)$ variables of $V_2$, and an injective function $b$
    mapping each constraint to a bag containing all its variables.  It is promised that one
    of the following holds:
    \begin{itemize}
        \item \emph{Yes:} $\psi$ has an ordinary satisfying assignment;
        \item \emph{No:} $\psi$ has no monotone satisfying multi-assignment that is
        consistent for every variable of $V_2$.
    \end{itemize}
    Distinguishing these cases in time
    $\bO\bigl((B-\varepsilon)^p|\psi|^{\bO(1)}\bigr)$ would falsify the $\pw$-SETH.
\end{corollary}

\begin{lemmarep}[\appsymb]\label{lem:pw-seth:csp-to-cvc}
    Assume either of the following holds:
    \begin{enumerate}
        \item There exist an integer $C \ge 3$ and $\varepsilon > 0$ such that {\CVC} with capacities at most
        $C$ can be solved in time $\sO\bigl((C+1-\varepsilon)^{\pw}\bigr)$.
        \item There exist an integer $\Delta \ge 4$ and $\varepsilon > 0$ such that {\CVC} on graphs of
        maximum degree at most $\Delta$ can be solved in time
        $\sO\bigl((\floor{\Delta/2}+1-\varepsilon)^{\pw}\bigr)$.
    \end{enumerate}
    Then the $\pw$-SETH is false.
\end{lemmarep}

\begin{proof}
    Suppose one of the stated algorithms exists for some $\varepsilon>0$.
    Set $r=C$ in the capacity case and $r=\floor{\Delta/2}$ in the degree
    case. In either case, $r\geq2$ and the assumed running time is
    $\sO\bigl((r+1-\varepsilon)^{\pw}\bigr)$.
    We give a single reduction producing graphs with maximum capacity $\max\{3,r\}$
    and maximum degree $2r$, so the instances satisfy the restrictions of the respective case.

    Let $\psi$ be an instance of the promise problem in \cref{cor:pw-seth:monotone-csp} over the alphabet $[r+1]$.
    Its variable set is $X=V_1\uplus V_2$, and it comes with a
    width-$p$ path decomposition $\mathcal B=(B_1,\ldots,B_\tau)$ in which every bag
    contains $\bO(\log p)$ variables of $V_2$, and an injective map $b$ assigning every
    constraint to a bag containing its scope.
    For $x\in X$, let $I_x\coloneq\setdef{t\in[\tau]}{x\in B_t}$
    denote the interval of bags containing $x$.
    We may assume that every constraint has at least one allowed tuple,
    since otherwise we can immediately reject.

    We relabel the alphabet, in order, by
    \[
        0\prec\cdots\prec r-2\prec r\prec r-1.
    \]
    For $r=2$, this is $0\prec2\prec1$.
    Monotonicity refers to $\prec$, not to the usual numerical order.

    We construct an instance $(G,\capacity,k)$ of \CVC.
    Each variable has a row of \emph{blocks}, one per bag containing it.
    The rows for $V_1$ will enforce monotonicity, while those for $V_2$ will
    enforce consistency. Constraint gadgets connect to the blocks in their assigned bags.
    All non-hub vertices receive capacity equal to their final degree;
    the hubs receive capacity $r$.

    \proofsubparagraph{Block gadgets.}
    For every $x\in X$ and $t\in I_x$, introduce a block $A_{x,t}$ containing
    $r$ disjoint edges $\ell_i r_i$, $i\in[r]$.
    The intended selection for a value $a\in\{0,\ldots,r\}$ is
    \[
        S_a(A_{x,t})=\{\ell_1,\ldots,\ell_a,r_{a+1},\ldots,r_r\}.
    \]
    Vertex names are local to each block.
    We distinguish two types of rows, as illustrated in \cref{fig:cvc-blocks}.

    If $x\in V_1$, the block also has vertices $\alpha,\beta$ of capacity $r$,
    with edges $\alpha\ell_i$ and $r_i\beta$ for every $i\in[r]$.
    Identify $\beta$ of each block with $\alpha$ of the next block in the row.
    We call the resulting $\alpha,\beta$ vertices the \emph{hubs} of the row.
    Finally, attach one leaf $f_x$ of capacity $1$ to the last $\beta$.

    If $x\in V_2$, there are no vertices $\alpha,\beta$.
    Join $r_i$ of each block to $\ell_i$ of the next block, and join $r_i$ of the
    last block to $\ell_i$ of the first, for every $i\in[r]$.
    For a row with at least two blocks, each index gives an even cycle.
    If the row has only one block, keep just the edge $\ell_i r_i$.

    \begin{figure}[htb]
        \centering
        \begin{subfigure}[t]{\textwidth}
            \centering
            \begin{tikzpicture}[
                scale=.9, transform shape, every node/.style={font=\small},
                port/.style={circle, draw, fill=white, minimum size=8mm,
                             inner sep=1pt},
                hub/.style={port, very thick, fill=black!14, minimum size=9mm},
                edge/.style={draw=black!78, line width=.65pt}]
                \path[use as bounding box] (-5.9,-1.65) rectangle (5.9,2.05);
                \node[hub] (a) at (-5.2,0) {$\alpha^t$};
                \node[hub] (h) at (0,0) {$h$};
                \node[hub] (b) at (5.2,0) {$\beta^{t+1}$};
                \node at (-2.6,1.85) {$A_{x,t}$};
                \node at (2.6,1.85) {$A_{x,t+1}$};
                \foreach \i/\y in {1/1.1,r/-1.1} {
                    \node[port] (l\i) at (-3.7,\y) {$\ell_{\i}^{t}$};
                    \node[port] (r\i) at (-1.5,\y) {$r_{\i}^{t}$};
                    \node[port] (ll\i) at (1.5,\y) {$\ell_{\i}^{t+1}$};
                    \node[port] (rr\i) at (3.7,\y) {$r_{\i}^{t+1}$};
                    \draw[edge] (a)--(l\i)--(r\i)--(h)--(ll\i)--(rr\i)--(b);
                }
                \foreach \x in {-3.7,-1.5,1.5,3.7}
                    \node at (\x,0) {$\vdots$};
                \node[font=\scriptsize, align=center] at (0,-1.2)
                    {$\capacity(h)=r$\\$\deg(h)=2r$};
            \end{tikzpicture}
            \caption{A $V_1$-row: $h=\beta^t=\alpha^{t+1}$ is one shared vertex,
            adjacent to all $r_i^t$ and $\ell_i^{t+1}$.}
            \label{fig:cvc-v1-row}
        \end{subfigure}

        \medskip
        \begin{subfigure}[t]{\textwidth}
            \centering
            \begin{tikzpicture}[
                scale=.9, transform shape, every node/.style={font=\small},
                port/.style={circle, draw, fill=white, minimum size=8mm,
                             inner sep=1pt},
                edge/.style={draw=black!78, line width=.65pt}]
                \path[use as bounding box] (-5.9,-1.65) rectangle (5.9,2.05);
                \node at (-2.6,1.85) {$A_{x,t}$};
                \node at (2.6,1.85) {$A_{x,t+1}$};
                \foreach \i/\y in {1/1.1,r/-1.1} {
                    \node[port] (l\i) at (-3.7,\y) {$\ell_{\i}^{t}$};
                    \node[port] (r\i) at (-1.5,\y) {$r_{\i}^{t}$};
                    \node[port] (ll\i) at (1.5,\y) {$\ell_{\i}^{t+1}$};
                    \node[port] (rr\i) at (3.7,\y) {$r_{\i}^{t+1}$};
                    \draw[edge] (l\i)--(r\i);
                    \draw[edge] (ll\i)--(rr\i);
                    \draw[edge, very thick] (-5.2,\y)--(l\i)
                        (r\i)--(ll\i) (rr\i)--(5.2,\y);
                }
                \foreach \x in {-3.7,-1.5,1.5,3.7}
                    \node at (\x,0) {$\vdots$};
            \end{tikzpicture}
            \caption{A $V_2$-row: thick edges join $r_i^t$ to $\ell_i^{t+1}$.
            Outer stubs indicate continuation of the row.}
            \label{fig:cvc-v2-row}
        \end{subfigure}
        \caption{Two consecutive blocks of each row type; superscripts identify
        their bags. Only the first and last pairs are drawn.
        Gray hubs have capacity $r$ and degree at most $2r$, with equality
        at a shared hub. Attachment paths at the white vertices are omitted;
        each adds at most one incident edge, keeping their degree at most $3$.}
        \label{fig:cvc-blocks}
    \end{figure}

    \proofsubparagraph{Constraint gadgets.}
    Let $c$ have scope $\{x_1,\ldots,x_q\}$, where $1\leq q\leq4$, and let
    $\Sigma_c=\{\mathbf{a}^1,\ldots,\mathbf{a}^m\}$ be its nonempty set of
    allowed tuples over the relabeled alphabet, in an arbitrary order.
    Put $d=qr$. Introduce the path
    \[
        s_c,u_1,v_1,u_2,v_2,\ldots,u_m,v_m,t_c.
    \]
    For each $\mathbf{a}^h=(a^h_1,\ldots,a^h_q)\in\Sigma_c$,
    introduce a path on $2d-1$ vertices
    \[
        P_{\mathbf{a}^h}=z_1,w_1,z_2,w_2,\ldots,z_{d-1},w_{d-1},z_d,
    \]
    and add the edges $u_hz_1$ and $z_dv_h$. Vertex names on these paths
    are local to each tuple. Thus each tuple path, together with $u_h,v_h$,
    forms a cycle on $2d+1$ vertices; consecutive cycles are joined by the
    edges $v_hu_{h+1}$, and $s_c,t_c$ are leaves at the two ends.
    For each $h\in[m]$, $j\in[q]$, and $i\in[r]$, assign
    $z_{(j-1)r+i}$ of $P_{\mathbf{a}^h}$ to the block endpoint $\ell_i$ of
    $A_{x_j,b(c)}$ if $i\leq a^h_j$, and to $r_i$ otherwise.
    Thus tuple $\mathbf{a}^h$ prescribes the endpoint selections
    $\bigcup_{j=1}^q S_{a^h_j}(A_{x_j,b(c)})$.
    We realize these assignments through attachment paths, defined below:
    each $z$-vertex is joined to a vertex $q_j^v$ on the path of its
    assigned block endpoint $v$. Each $z$-vertex receives just one such
    edge, so every vertex of the constraint gadget has degree at most $3$.
    Set the budget of the gadget to $k_c=m(d+1)$.
    For an illustration, see \cref{fig:cvc-constraint}.

    \begin{figure}[htb]
        \centering
        \begin{tikzpicture}[
            every node/.style={font=\scriptsize},
            port/.style={circle, draw, fill=white, minimum size=5mm, inner sep=0pt},
            chosen/.style={port, fill=black, text=white},
            edge/.style={draw=black!78, line width=.65pt}]
            \foreach \h/\x in {1/-4,2/0,3/4} {
                \node[port] (u\h) at (\x-1.4,0) {$u_{\h}$};
                \node[port] (v\h) at (\x+1.4,0) {$v_{\h}$};
                \foreach \id/\dx/\lab in {z1/-1.4/z_1,w1/-.7/w_1,z2/0/z_2,w2/.7/w_2,z3/1.4/z_3} {
                    \node[port] (p\h\id) at (\x+\dx,1.3) {$\lab$};
                }
                \draw[edge] (u\h)--(p\h z1)--(p\h w1)--(p\h z2)
                    --(p\h w2)--(p\h z3)--(v\h)--(u\h);
                \foreach \id/\dx in {z1/-1.4,z2/0,z3/1.4}
                    \draw[edge] (p\h\id)--(\x+\dx,2.05);
                \node at (\x,2.5) {$P_{\mathbf{a}^{\h}}$};
            }
            \node[port] (s) at (-6.4,0) {$s_c$};
            \node[port] (t) at (6.4,0) {$t_c$};
            \draw[edge] (s)--(u1) (v1)--(u2) (v2)--(u3) (v3)--(t);
            \foreach \h in {1,3} {
                \foreach \id/\lab in {z1/z_1,z2/z_2,z3/z_3}
                    \node[chosen] at (p\h\id) {$\lab$};
            }
            \node[chosen] at (p2w1) {$w_1$};
            \node[chosen] at (p2w2) {$w_2$};
            \foreach \id/\lab in {u1/u_1,u2/u_2,v2/v_2,v3/v_3}
                \node[chosen] at (\id) {$\lab$};
            \foreach \x in {-1.4,0,1.4}
                \node[circle, fill=black, inner sep=2pt] at (\x,2.05) {};
        \end{tikzpicture}
        \caption{The constraint gadget for $m=d=3$, with tuple $\mathbf{a}^2$
        chosen. Black vertices are selected. Upper stubs lead to the
        attachment-path vertices $q_j^v$ defined below; the black upper
        endpoints cover these edges for the chosen tuple.
        The selected connectors are $u_1,u_2,v_2,v_3$,
        and the two leaves are unselected.}
        \label{fig:cvc-constraint}
    \end{figure}

    \proofsubparagraph{Attachment paths.}
    These paths distribute the constraint edges while keeping the
    degree of each block endpoint at most $3$.
    For each block endpoint $v\in\{\ell_i,r_i\}$, let $\delta_v$ be the
    number of $z$-vertices assigned to it. Add the path
    \[
        Q_v=q_0^v,w_1^v,q_1^v,\ldots,w_{\delta_v}^v,q_{\delta_v}^v,
        \qquad q_0^v=v.
    \]
    Keep the row edges at $v$, and join its assigned $z$-vertices to
    distinct vertices $q_1^v,\ldots,q_{\delta_v}^v$.
    If $\delta_v=0$, the path consists only of $v$.
    Every non-hub vertex now has degree at most $3$, and we set its capacity equal to its degree.
    The hubs have capacity $r$ and degree at most $2r$, so the maximum
    degree of $G$ is at most $2r$ and all capacities are at most $\max\{3,r\}$.

    The intended choices on $Q_v$ are all its $q$-vertices when $v$ is
    selected, and all its $w$-vertices otherwise. Their sizes are
    $\delta_v+1$ and $\delta_v$, respectively; see
    \cref{fig:cvc-attachments}. We regard the new vertices of $Q_v$ as
    belonging to the block and row of $v$.

    \begin{figure}[htb]
        \centering
        \begin{tikzpicture}[
            every node/.style={font=\scriptsize},
            port/.style={circle, draw, fill=white, minimum size=6mm, inner sep=0pt},
            chosen/.style={port, fill=black, text=white},
            edge/.style={draw=black!78, line width=.65pt}]
            \foreach \selected/\height in {1/0,0/-3.1} {
                \begin{scope}[yshift=\height cm]
                    \node[port] (q0) at (-4.5,0) {$\ell_i$};
                    \foreach \j in {1,2,3} {
                        \node[port] (w\j) at (-6+3*\j,0) {$w_{\j}$};
                        \node[port] (q\j) at (-4.5+3*\j,0) {$q_{\j}$};
                    }
                    \draw[edge] (q0)--(w1)--(q1)--(w2)--(q2)--(w3)--(q3);
                    \draw[edge] (q0.north west)--++(-.75,.45)
                        node[above left] {$\alpha$};
                    \draw[edge] (q0.south west)--++(-.75,-.45)
                        node[below left] {$r_i$};
                    \foreach \h in {1,2,3} {
                        \node[chosen] (z\h) at (-4.5+3*\h,1.3) {$z_s^{\h}$};
                        \draw[edge] (z\h)--(q\h);
                        \draw[edge] (z\h.west)--++(-.65,0);
                        \draw[edge] (z\h.east)--++(.65,0);
                        \node at (-4.5+3*\h,1.95) {$P_{\mathbf{a}^{\h}}$};
                    }
                    \ifnum\selected=1
                        \node[chosen] at (q0) {$\ell_i$};
                        \foreach \j in {1,2,3}
                            \node[chosen] at (q\j) {$q_{\j}$};
                        \node[port] at (z2) {$z_s^2$};
                        \node at (0,-.75) {(a) Selected endpoint: $\delta_v+1=4$ path vertices.};
                    \else
                        \foreach \j in {1,2,3}
                            \node[chosen] at (w\j) {$w_{\j}$};
                        \node at (0,-.75) {(b) Unselected endpoint: $\delta_v=3$ path vertices.};
                    \fi
                \end{scope}
            }
        \end{tikzpicture}
        \caption{The two canonical choices on an attachment path at
        $v=\ell_i$, for $\delta_v=3$. Black vertices are selected.
        The upper vertices belong to three tuple paths of the same constraint;
        superscripts identify the tuples. Horizontal stubs are edges within
        those paths, and the left stubs are the row edges at $\ell_i$.
        The path and tuple vertices have capacity equal to their degree, at most $3$.
        In (a), tuple $\mathbf{a}^2$ is chosen, so its attachment is covered
        by $q_2$. In (b), all three adjacent $z$-vertices are forced into the cover.}
        \label{fig:cvc-attachments}
    \end{figure}

    \proofsubparagraph{Budget.}
    Let $L_x=|I_x|$ and $D=\sum_v\delta_v$, where the sum ranges over all
    block endpoints. Either intended choice on $Q_v$ selects exactly
    $\delta_v$ new vertices, so the attachment paths contribute a fixed
    cost $D$, in addition to the selected block endpoints. Set
    \[
        k=D+\sum_{x\in V_1}\bigl(L_x(r+1)+1\bigr)
          +\sum_{x\in V_2}L_x r
          +\sum_c k_c.
    \]
    This completes the construction of $(G,\capacity,k)$.

    \proofsubparagraph{Completeness.}
    Let $\psi$ have a satisfying assignment $\sigma \colon X \to \{0,\ldots,r\}$.
    We argue that $G$ has a capacitated vertex cover of size $k$.

    Let $x \in X$ be a variable where $\sigma(x) = a \in \{0,\ldots,r\}$.
    Select $S_a$ in every block of its row, and extend this selection to
    each attachment path by choosing all its $q$-vertices or all its
    $w$-vertices according to whether its endpoint is selected.
    In a $V_1$-row, also select its $L_x+1$ distinct vertices $\alpha,\beta$.
    If $a=r$, replace the first $\alpha$ by the leaf $f_x$;
    otherwise leave $f_x$ unselected. Excluding the new attachment vertices,
    a $V_1$-row uses $L_x(r+1)+1$ vertices and a $V_2$-row uses $L_xr$.
    The attachment paths add exactly $D$ vertices in total.

    In the gadget of a constraint, choose the path for its satisfying tuple.
    If this is $P_{\mathbf{a}^h}$, select its $d-1$ vertices $w_i$, the $d$
    vertices $z_i$ of every other tuple path, and the connectors
    \[
        u_1,\ldots,u_h,\quad v_h,\ldots,v_m.
    \]
    Leave $s_c,t_c$ unselected. This covers all edges within the gadget using
    \[
        (d-1)+(m-1)d+(m+1)=m(d+1)=k_c
    \]
    vertices. The attachment edges of unchosen tuple paths have selected
    $z$-endpoints. Those of the chosen tuple path have selected $q$-endpoints,
    because its tuple agrees with the selected sets $S_a$ in its scope blocks.
    All edges within the attachment paths are also covered.

    In a $V_1$-block, each selected pair endpoint covers its two edges within the row;
    the remaining edges go to $\alpha$ or $\beta$.
    A shared vertex covers $a+(r-a)=r$ edges, and the first $\alpha$ covers
    $r-a$ edges. The last $\beta$ covers $a$ block edges and, if $a<r$,
    its leaf edge; if $a=r$, that edge goes to the selected leaf instead.
    Every remaining edge has only non-hub endpoints, so assign it to any
    selected endpoint. Their capacities equal their degrees.
    Thus all capacities are respected, and the selections sum to $k$.

    \proofsubparagraph{Soundness.}
    Let $S$ be a capacitated vertex cover of $G$ of size at most $k$.
    We argue that $\psi$ has a monotone satisfying multi-assignment consistent on $V_2$.

    First normalize the attachment paths. The unique minimum vertex cover
    of $Q_v$ is its $\delta_v$ vertices $w_j^v$. Any other cover of this path
    has size at least $\delta_v+1$ and can be replaced by all its $q$-vertices.
    This does not increase the size or uncover an external edge: only the
    $q$-vertices have neighbors outside $Q_v$, and none of their selected
    endpoints is removed. Reassign edges incident with removed $w$-vertices
    to the selected $q$-vertices, whose capacities equal their degrees.
    Hence we may assume that every $Q_v$ uses one of its two intended choices.

    In particular, a constraint edge with an unselected $z$-endpoint forces
    its $q$-endpoint, and therefore its block endpoint $v$, to be selected.
    The new attachment vertices contribute exactly $D$ to the cover.
    Subtract this contribution from both its size and the budget.
    Below, row counts refer only to block endpoints, hubs, and the leaf,
    with the row budgets given by the corresponding summands of $k-D$.

    \begin{claim}\label{cl:cvc-lb-budgets}
        Every row and constraint gadget uses exactly its budget.
        In each constraint gadget, some tuple $\mathbf{a}$ has all the block
        endpoints prescribed by its sets $S_{a_j}$ selected.
    \end{claim}

    \begin{claimproof}
        A constraint gadget contains $m$ vertex-disjoint odd cycles of
        $2d+1$ vertices, each requiring at least $d+1$ selected vertices. Hence it needs
        at least $m(d+1)=k_c$ vertices.
        Each tuple path itself requires at least $d-1$ selected vertices.
        If every tuple path used at least $d$ vertices, these paths would
        contribute at least $md$ selected vertices. Moreover, the edges
        \[
            s_cu_1,\quad v_iu_{i+1}\ (i\in[m-1]),\quad v_mt_c
        \]
        form a matching of size $m+1$ outside the tuple paths, requiring
        another $m+1$ selected vertices. This exceeds $k_c$.
        Thus equality requires a tuple path with exactly $d-1$ selected
        vertices, necessarily its $w_i$. All its $z_i$ are unselected,
        so its attachment edges force all the prescribed block endpoints
        to be selected by normalization.

        A $V_2$-row needs at least $L_xr$ vertices
        because its block edges are disjoint. For a $V_1$-row, write $L=L_x$.
        Choose one selected endpoint from each pair, accounting for $rL$ vertices.
        Exactly $rL+1$ row edges are not incident with these vertices:
        one hub edge per pair, and the leaf edge. They must be covered by
        additional selected vertices. A hub can cover at most $r$ of these
        edges, and any other vertex at most one. Thus at least $L+1$
        additional vertices are needed.

        If the bound is tight, at most one of these $L+1$ additional vertices
        can be a non-hub: two non-hubs would cover, together with the remaining
        vertices, at most $r(L-1)+2<rL+1$ edges, since $r\geq2$.
        Hence either all hubs are selected, or one hub is replaced by an
        extra pair endpoint or the leaf. In the latter case, the non-hub
        covers one edge and the $L$ selected hubs must cover all $rL$ others,
        whose other endpoints are unselected. Thus every selected hub is
        saturated by forced edges.

        These lower bounds sum to $k-D$, so all are tight and each constraint
        gadget has a path with the stated property.
    \end{claimproof}

    Choose one tuple path supplied by \cref{cl:cvc-lb-budgets} for each constraint.
    Assign each block the value prescribed by its chosen tuple, or the number
    of selected $\ell_i$ if no chosen tuple involves it. Since $b$ is injective,
    this is well-defined and satisfies every constraint.
    Whenever exactly one endpoint per pair is selected, the assigned value
    equals the number of selected $\ell_i$: the chosen path, if present,
    forces precisely the $r$ endpoint selections of $S_a$ in this block.

    In a $V_2$-row, exactly one endpoint per pair is selected. Its cyclic edges
    force a selected $\ell_i$ to be selected in every subsequent block, since
    its partner $r_i$ is unselected. Similarly for $r_i$. Thus the assigned values are constant.

    For a $V_1$-row with $L$ blocks, let $a_t$ count the selected $\ell_i$
    in block $t$. If all hubs are selected, there are no extra pair endpoints
    and the leaf is unselected. The shared hubs give
    $a_t+(r-a_{t+1})\leq r$, and the last $\beta$ gives $a_L+1\leq r$.
    Hence the assigned values satisfy $a_1\leq\cdots\leq a_L<r$.

    Otherwise one hub is omitted and every selected hub is saturated.
    The first $\alpha$, if selected, forces $a_1=0$.
    Across a selected shared hub, saturation keeps the count unchanged,
    except for a decrease by one after a block with an extra pair endpoint.
    Thus the counts before the omitted hub are $0$, while the first block
    after it, if any, has count $r$. With at most one extra pair endpoint, the sequence is
    \[
        0,\ldots,0,\quad r,\ldots,r,\quad r-1,\ldots,r-1,
    \]
    with possibly empty parts. The only block whose assigned value may differ
    from its count is the one with an extra pair endpoint. It is the last
    block of the $r$-part and contains all $\ell_i$ and one $r_i$.
    Its assigned value is $r$ or $r-1$ (the latter only if the extra endpoint
    has index $r$). Either value preserves monotonicity in the order $\prec$.
    This proves soundness.

    \begin{claim}\label{cl:cvc-lb-pathwidth}
        $G$ has pathwidth at most $p+\bO_r(\log(p+2))$.
    \end{claim}
    \begin{claimproof}
        Sweep through the bags $B_1,\ldots,B_\tau$, retaining one boundary hub
        per active $V_1$-row. For each active $V_2$-row, retain the $r$ vertices
        $r_i$ of its latest block and the $r$ vertices $\ell_i$ of its first
        block, needed to close the row. Since each bag has at most $p+1$
        variables, only $\bO(\log p)$ of them in $V_2$, these retained vertices
        number at most $p+\bO_r(\log(p+2))$.

        At bag $B_t$, process its variable blocks one at a time, except that
        a constraint assigned to this bag is processed together with its at
        most four scope blocks, including their attachment paths.
        Each block participates in at most one constraint, since $b$ is
        injective, and a constraint has at most $(r+1)^4$ allowed tuples.
        Thus $\delta_v\leq(r+1)^4$, and all vertices added in a step number
        $\bO_r(1)$. No attachment-path vertex other than its block endpoint
        needs to remain in the row boundary.
        Afterwards forget everything except the updated row boundaries.
        Start each row at its first bag and finish it at its last, covering
        its leaf edge or closing edges before forgetting its boundary.
        Every edge is covered and every vertex occurs in an interval, so this
        gives the claimed decomposition in polynomial time.
    \end{claimproof}

    The supposed algorithm would consequently distinguish the two promised
    cases in time
    \[
        (r+1-\varepsilon)^{p+\bO_r(\log(p+2))}|\psi|^{\bO(1)}
        =(r+1-\varepsilon)^p|\psi|^{\bO(1)},
    \]
    which falsifies the $\pw$-SETH by \cref{cor:pw-seth:monotone-csp}.
\end{proof}

To complete the proof of \cref{thm:pw-seth:equivalences}, it remains to handle the case $\Delta=3$.

\begin{lemmarep}[\appsymb]\label{lem:pw-seth:degree-three}
    Assuming the $\pw$-SETH, for every $\varepsilon>0$,
    {\CVC} on graphs of maximum degree $3$ cannot be solved in time
    $\sO\bigl((2-\varepsilon)^{\pw}\bigr)$,
    even if a path decomposition of width $\pw$ is given.
\end{lemmarep}

\begin{proof}
    As the reduction is similar to the one in \cref{lem:pw-seth:csp-to-cvc}, we only sketch the argument.
    We reduce from CSP with boolean variables, and essentially substitute the $V_1$-type gadgets from \cref{lem:pw-seth:csp-to-cvc}
    to the standard gadgets for \textsc{Vertex Cover} presented by Lokshtanov, Marx, and Saurabh~\cite{talg/LokshtanovMS18}.
    Those are identical to our $V_2$-type gadgets, except that they have no closing edges.
    The rest of the construction remains the same and the theorem follows.
\end{proof}

\section{Maximum Degree}\label{sec:degree}

We consider {\CVCshort} parameterized by treewidth or pathwidth on graphs of bounded maximum degree.
Fix a constant maximum-degree bound $\Delta\ge3$.
As in \cref{sec:preliminaries}, orient each edge towards the endpoint covering it.
We start by sketching a natural dynamic program for the parameterization by treewidth,
identify its bottleneck, and then show how to overcome it.

\subparagraph*{A basic dynamic program.}
Consider a nice tree decomposition with introduce-edge nodes.
For a vertex $v$, let $q_v\coloneq \deg(v)-\capacity(v)$.
Observe that a selected vertex respects its capacity exactly when at least $q_v$ incident edges are outgoing;
an unselected vertex must have all incident edges outgoing.
For a vertex of capacity at most $\floor{\Delta/2}$, count its processed
incoming edges, using states $0,\ldots,\capacity(v)$ and rejecting
overflow. Once all its incident edges are processed, select the vertex
exactly when this count is positive. For a vertex of larger capacity,
count processed outgoing edges instead, stopping at $q_v$. These
numeric states mean that the vertex is selected. Add one state $\bot$
for leaving it unselected, which forbids incoming edges.
This gives $q_v+2\leq\ceil{\Delta/2}+1$ states.
At the end, accept only $q_v$ (selected) or $\bot$ (unselected).

For each tuple of bag states, record the attainable numbers of selected
vertices already forgotten. When an edge is introduced, try both
orientations and update its endpoint counters; when a vertex is
forgotten, check its final state and charge for selecting it. At a
join, add the incoming counts, rejecting overflow, and add outgoing
counts with capping at $q_v$; state $\bot$ combines only with itself.
There are at most $\ceil{\Delta/2}+1$ states per vertex, and standard convolution
techniques give time $\sO((\ceil{\Delta/2}+1)^{\tw})$.

\subparagraph*{Exceptional vertices.}
Only outgoing counters can exceed $\floor{\Delta/2}+1$ states.
This requires $q_v\geq\floor{\Delta/2}$, which together with
$\capacity(v)>\floor{\Delta/2}$ forces $\Delta$ to be odd,
$\deg(v)=\Delta$, and $\capacity(v)=\ceil{\Delta/2}$.
Call these vertices \emph{exceptional}, and let $\mathcal C$ be their
set. At such a vertex, either ordinary counter uses
$\floor{\Delta/2}+2$ states, while every other vertex already fits
within $\floor{\Delta/2}+1$ states.

At an exceptional vertex, mark any subset of its outgoing edges and
count the marks modulo $\floor{\Delta/2}+1$. Require final residue
$\floor{\Delta/2}$. Since there are at most
$\Delta=2\floor{\Delta/2}+1$ marks, the accepted totals are exactly
$\floor{\Delta/2}$ and $\Delta$. The first guarantees enough outgoing
edges to select the vertex within its capacity; the second guarantees
that every edge is outgoing, so it may remain unselected. These
\emph{cyclic counters} have the desired number of states and combine
by modular addition at joins.

To recover the selection cost, give one point per mark and
$\floor{\Delta/2}+1$ points per unselected nonexceptional vertex.
Every exceptional vertex then contributes a fixed baseline of
$\floor{\Delta/2}$ points, plus $\floor{\Delta/2}+1$ more when it is
unselected. Thus, for complete solutions, maximizing the score
minimizes the cover size. We record which scores are attainable for each
tuple of bag states. Scores have polynomial range and add at joins,
so these joins are ordinary convolutions.

\begin{theoremrep}[\appsymb]\label{thm:degree:treewidth-algorithm}
    For every fixed $\Delta\ge3$, given a capacitated graph of maximum
    degree at most $\Delta$ and a tree decomposition of width $\tw$,
    {\CVCshort} can be solved in time
    $\sO\bigl((\floor{\Delta/2}+1)^{\tw}\bigr)$.
\end{theoremrep}

\begin{proof}
    Put $n=|V(G)|$.
    Use a nice tree decomposition with empty root and leaves and
    introduce-edge nodes. Each edge is introduced exactly once,
    before either endpoint is forgotten. At node $t$, let $X_t$ be
    the bag, $V_t$ the vertices occurring in its subtree, and $E_t$
    the edges introduced there. We call $E_t$ the processed edges
    and $V_t\setminus X_t$ the forgotten vertices. Every forgotten
    vertex has all its incident edges in $E_t$.

    \proofsubparagraph{States and table entries.}
    Recall that $q_v=\deg(v)-\capacity(v)$.
    Use incoming counters at vertices of capacity at most
    $\floor{\Delta/2}$, cyclic counters at the exceptional vertices
    $\mathcal C$, and outgoing counters at all remaining vertices.
    A \emph{partial solution} at $t$ consists of an orientation of
    $E_t$, a set $M_t\subseteq E_t$ of marked edges whose tails are
    cyclic vertices, and a selection choice for each outgoing-counter
    vertex in $V_t$. All in-degrees, out-degrees, and marks below
    refer to $E_t$.
    The state $\sigma_v$ at a bag vertex $v$ belongs to a set $Q_v$
    and has the following meaning:
    \begin{itemize}
        \item \emph{Incoming:} $Q_v=\{0,\ldots,\capacity(v)\}$,
        and $\sigma_v$ is the in-degree of $v$.
        \item \emph{Outgoing:} $Q_v=\{\bot,0,\ldots,q_v\}$.
        If $v$ is selected, $\sigma_v$ is its out-degree capped at $q_v$;
        otherwise $\sigma_v=\bot$.
        \item \emph{Cyclic:} $Q_v=\{0,\ldots,\floor{\Delta/2}\}$,
        and $\sigma_v$ is the number of marked outgoing edges modulo
        $\floor{\Delta/2}+1$.
    \end{itemize}
    Each $Q_v$ has size at most $\floor{\Delta/2}+1$.
    A \emph{signature} is the tuple of bag states:
    \[
        \sigma=(\sigma_v)_{v\in X_t}\in\Sigma_t,
        \qquad \Sigma_t:=\prod_{v\in X_t}Q_v.
    \]

    A partial solution is \emph{valid} if every incoming-counter vertex
    in $V_t$ respects its capacity and every unselected outgoing-counter
    vertex in $V_t$ has no incoming edges. At forgotten vertices, we
    additionally require that selected outgoing-counter vertices
    respect their capacities and that cyclic vertices have exactly
    $\floor{\Delta/2}$ or $\Delta$ marked outgoing edges.
    A forgotten incoming-counter vertex is selected exactly when its
    in-degree is positive; a forgotten cyclic vertex is selected
    exactly when it has $\floor{\Delta/2}$ marks.
    Let $U_t$ be the set of unselected forgotten vertices outside
    $\mathcal C$. We set the \emph{score} of the partial solution as
    \[
        s \coloneq (\floor{\Delta/2}+1)|U_t|+|M_t|.
    \]
    A noncyclic vertex contributes at most $\floor{\Delta/2}+1$ points,
    and a cyclic vertex at most $\Delta$ marks, so $0\leq s\leq\Delta n$.
    For every node $t$ of the tree decomposition,
    we construct a Boolean table $\DP_t$ with the following invariant:
    \begin{equation}\label{eq:degree:treewidth-invariant}
        \DP_t[\sigma,s]=\true
        \quad\iff\quad
        \begin{gathered}
            \text{there is a valid partial solution at $t$}\\
            \text{with signature $\sigma$ and score $s$}.
        \end{gathered}
    \end{equation}

    Compute the tables from the leaves to the root using the following
    recurrences. At a node with one child, denote the child by $t'$.
    Entries with scores outside $\{0,\ldots,\Delta n\}$ and empty disjunctions are false.

    \proofsubparagraph{Leaf nodes.}
    At an empty leaf, where $\varnothing$ denotes the empty signature, set
    \[
        \DP_t[\varnothing,s]=\begin{cases}
            \true, & s=0,\\
            \false, & s>0.
        \end{cases}
    \]
    The empty partial solution has score zero, establishing
    \eqref{eq:degree:treewidth-invariant}.

    \proofsubparagraph{Introduce-vertex nodes.}
    Suppose $t$ introduces $v$. We allow only states $0$ and $\bot$ to be used
    for the new vertex, with $\bot$ available only for outgoing counters, thus we set%
    \footnote{For a signature $\sigma$ in node $t$, $\sigma|_{X_{t'}}$ denotes the restriction of $\sigma$ to the vertices in $X_{t'}$.}
    \[
        \DP_t[\sigma,s]=\begin{cases}
            \DP_{t'}[\sigma|_{X_{t'}},s], & \sigma_v\in\{0,\bot\},\\
            \false, & \text{otherwise}.
        \end{cases}
    \]

    \begin{claim}\label{cl:degree:introduce-vertex}
        If \eqref{eq:degree:treewidth-invariant} holds at $t'$, then
        it holds at the introduce-vertex node $t$.
    \end{claim}
    \begin{claimproof}
        Let $\sigma\in\Sigma_t$ and $s\in\{0,\ldots,\Delta n\}$ such
        that $\DP_t[\sigma,s]=\true$. We argue that there exists a
        valid partial solution at $t$ with signature $\sigma$ and
        score $s$. By the recurrence, $\sigma_v\in\{0,\bot\}$ and
        $\DP_{t'}[\sigma|_{X_{t'}},s]=\true$. The invariant at $t'$
        therefore gives a valid partial solution with this signature
        and score. The introduced vertex $v$ is absent from $V_{t'}$
        and has no processed incident edges. Add $v$ to this solution;
        if it has an outgoing counter, select it exactly when
        $\sigma_v=0$. Its state is then $\sigma_v$, and all other
        states are unchanged. The solution is valid, and its score
        remains $s$ because no mark is added and no vertex is forgotten.

        Conversely, suppose there is a valid partial solution at $t$
        with signature $\sigma\in\Sigma_t$ and score
        $s\in\{0,\ldots,\Delta n\}$. We argue that
        $\DP_t[\sigma,s]=\true$. Since $v$ has no processed incident
        edges, $\sigma_v$ is $0$ or $\bot$, with $\bot$ possible only
        for an outgoing counter. Removing $v$ and its selection choice,
        if any, gives a valid partial solution at $t'$ with signature
        $\sigma|_{X_{t'}}$ and the same score $s$. By the invariant at
        $t'$, $\DP_{t'}[\sigma|_{X_{t'}},s]=\true$, so the recurrence
        sets $\DP_t[\sigma,s]=\true$.
    \end{claimproof}

    \proofsubparagraph{Introduce-edge nodes.}
    Suppose $t$ introduces the edge $uv$. We consider both orientations
    and both marking choices. Let $b\in\{0,1\}$ indicate whether the
    edge is marked. For the orientation $u\to v$, define
    $\mathcal P_t^{u\to v,b}(\sigma)\subseteq\Sigma_{t'}$ to be the set
    of child signatures that produce $\sigma$ with this choice of $b$.
    This set is empty if $\sigma_v=\bot$, or if $b=1$ and $u$ is not
    cyclic. Otherwise, its signatures are obtained as follows.
    Copy every coordinate outside $\{u,v\}$.

    At the head $v$, the child state is determined by its counter type:
    \[
        \sigma'_v=\begin{cases}
            \sigma_v-1, & \text{incoming},\\
            \sigma_v, & \text{outgoing or cyclic}.
        \end{cases}
    \]
    At the tail $u$, include every child state listed in the applicable row:
    \[
        \begin{array}{l|l}
            \text{counter/state at }u & \text{child states }\sigma'_u\\ \hline
            \text{incoming or }\sigma_u=\bot & \sigma_u\\
            \text{outgoing, numeric }\sigma_u<q_u & \sigma_u-1\\
            \text{outgoing, numeric }\sigma_u=q_u & q_u-1,\ q_u\\
            \text{cyclic} & (\sigma_u-b)\bmod(\floor{\Delta/2}+1)
        \end{array}
    \]
    Discard negative states; in particular, when $q_u=0$, the capped
    outgoing case contributes only child state $0$.
    Define $\mathcal P_t^{v\to u,b}(\sigma)$ symmetrically, and set
    \begin{align*}
        \DP_t[\sigma,s]
        =\quad&\bigvee_{\sigma'\in\mathcal P_t^{u\to v,0}(\sigma)}
            \DP_{t'}[\sigma',s]
          \;\lor\!\bigvee_{\sigma'\in\mathcal P_t^{u\to v,1}(\sigma)}
            \DP_{t'}[\sigma',s-1]\\
        \lor&\bigvee_{\sigma'\in\mathcal P_t^{v\to u,0}(\sigma)}
            \DP_{t'}[\sigma',s]
          \;\lor\!\bigvee_{\sigma'\in\mathcal P_t^{v\to u,1}(\sigma)}
            \DP_{t'}[\sigma',s-1].
    \end{align*}
    Thus, for either orientation with a cyclic tail, we allow both
    leaving the edge unmarked ($b=0$, child score $s$) and marking it
    ($b=1$, child score $s-1$). The parent entry is true if at least
    one of these possibilities succeeds.

    \begin{claim}\label{cl:degree:introduce-edge}
        If \eqref{eq:degree:treewidth-invariant} holds at $t'$, then
        it holds at the introduce-edge node $t$.
    \end{claim}
    \begin{claimproof}
        Let $\sigma\in\Sigma_t$ and $s\in\{0,\ldots,\Delta n\}$ such
        that $\DP_t[\sigma,s]=\true$. We argue that there exists a
        valid partial solution at $t$ with signature $\sigma$ and
        score $s$. The recurrence supplies an orientation of $uv$,
        a bit $b$, and a signature $\sigma'$ in the corresponding
        predecessor set such that $\DP_{t'}[\sigma',s-b]=\true$.
        By the invariant at $t'$, there is a valid child solution
        with signature $\sigma'$ and score $s-b$. Add $uv$ with the
        chosen orientation, marking it exactly when $b=1$.
        The predecessor rules give signature $\sigma$, exclude state
        $\bot$ at the head, and enforce the incoming capacity bounds.
        They also allow a mark only at a cyclic tail. Both endpoints
        are in the bag, so no forgotten vertex is affected. Thus the
        resulting solution is valid and has score $(s-b)+b=s$.

        Conversely, suppose there is a valid partial solution at $t$
        with signature $\sigma\in\Sigma_t$ and score
        $s\in\{0,\ldots,\Delta n\}$. We argue that
        $\DP_t[\sigma,s]=\true$. Delete $uv$ and its mark, if any,
        and let $b$ indicate whether it was marked. The resulting
        child solution is valid: in-degrees can only decrease, and
        forgotten vertices are unaffected. Its score is $s-b$, and
        its signature $\sigma'$ belongs to the predecessor set for
        the deleted orientation and marking choice. By the invariant
        at $t'$, $\DP_{t'}[\sigma',s-b]=\true$. This is a true
        disjunct in the recurrence, so $\DP_t[\sigma,s]=\true$.
    \end{claimproof}

    \proofsubparagraph{Forget-vertex nodes.}
    Suppose $t$ forgets $v$; all edges incident with $v$ have been
    introduced. For $a\in Q_v$, write $\sigma^a$ for $\sigma$ extended
    by state $a$ at $v$. Set
    \[
        \DP_t[\sigma,s]=\begin{cases}
            \DP_{t'}[\sigma^0,s-(\floor{\Delta/2}+1)]
                \lor\displaystyle\bigvee_{a=1}^{\capacity(v)}\DP_{t'}[\sigma^a,s],
                & \text{incoming},\\[1ex]
            \DP_{t'}[\sigma^{q_v},s]
                \lor\DP_{t'}[\sigma^\bot,s-(\floor{\Delta/2}+1)],
                & \text{outgoing},\\[1ex]
            \DP_{t'}[\sigma^{\floor{\Delta/2}},s],
                & \text{cyclic}.
        \end{cases}
    \]
    The score changes only when a noncyclic vertex is left unselected;
    marks have already contributed their points at introduce-edge nodes.

    \begin{samepage}
    \begin{claim}\label{cl:degree:forget}
        If \eqref{eq:degree:treewidth-invariant} holds at $t'$, then
        it holds at the forget-vertex node $t$.
    \end{claim}
    \begin{claimproof}
        Let $\sigma\in\Sigma_t$ and $s\in\{0,\ldots,\Delta n\}$ such
        that $\DP_t[\sigma,s]=\true$. We argue that there exists a
        valid partial solution at $t$ with signature $\sigma$ and
        score $s$. The recurrence supplies a true child entry
        $\DP_{t'}[\sigma^a,s']$ in the applicable row. By the invariant
        at $t'$, there is a valid partial solution with signature
        $\sigma^a$ and score $s'$. Regard it as a solution at $t$,
        with $v$ now forgotten. Only the final
        conditions at $v$ need checking. An incoming state already
        respects the capacity. For an outgoing counter, state $q_v$
        guarantees at least $q_v$ outgoing edges and hence respects
        the capacity, while $\bot$ guarantees no incoming edges.
        For a cyclic counter, there are at most
        $\Delta=2\floor{\Delta/2}+1$ marks, so residue
        $\floor{\Delta/2}$ means exactly $\floor{\Delta/2}$ or
        $\Delta$ marks. Thus the solution is valid and has signature
        $\sigma$. Its score increases by $\floor{\Delta/2}+1$
        exactly when $v$ is noncyclic and unselected, which is precisely
        the score change in the recurrence. Its score at $t$ is therefore $s$.

        Conversely, suppose there is a valid partial solution at $t$
        with signature $\sigma\in\Sigma_t$ and score
        $s\in\{0,\ldots,\Delta n\}$. We argue that
        $\DP_t[\sigma,s]=\true$. Put $v$ back in the bag, leaving the
        orientation, marks, and selection choices unchanged. This
        gives a valid child solution. Its state $a$ is accepted by
        the recurrence: any incoming state is allowed, an outgoing
        counter is in state $q_v$ or $\bot$, and a cyclic counter
        has residue $\floor{\Delta/2}$. Its score is
        $s-(\floor{\Delta/2}+1)$ if $v$ is noncyclic and unselected,
        and $s$ otherwise. By the invariant at $t'$, the corresponding
        child entry is true, so $\DP_t[\sigma,s]=\true$.
    \end{claimproof}
    \end{samepage}

    \proofsubparagraph{Join nodes.}
    At a join with children $t_1,t_2$, we have
    $X_t=X_{t_1}=X_{t_2}$, $V_{t_1}\cap V_{t_2}=X_t$, and
    $E_t=E_{t_1}\cup E_{t_2}$, where the union is disjoint.
    Combine signatures $\alpha,\beta\in\Sigma_t$ coordinate by
    coordinate. For a bag vertex $v$, combine $a=\alpha_v$ and
    $b=\beta_v$ as follows:
    \[
        \begin{array}{c|c}
            \text{counter} & \text{parent state}\\ \hline
            \text{incoming} & a+b\quad\text{if }a+b\leq\capacity(v)\\
            \text{outgoing, numeric} & \min\{q_v,a+b\}\\
            \text{outgoing, both }\bot & \bot\\
            \text{cyclic} & (a+b)\bmod(\floor{\Delta/2}+1)
        \end{array}
    \]
    All other combinations are forbidden. Write $\alpha\oplus\beta$
    for the resulting signature.
    The parent entry is true if compatible true child entries have
    scores summing to $s$:
    \begin{equation}\label{eq:degree:treewidth-join}
        \DP_t[\sigma,s]=
        \bigvee_{\substack{\alpha\oplus\beta=\sigma\\s_1+s_2=s}}
            \bigl(\DP_{t_1}[\alpha,s_1]\land\DP_{t_2}[\beta,s_2]\bigr).
    \end{equation}

    \begin{claim}\label{cl:degree:join}
        If \eqref{eq:degree:treewidth-invariant} holds at both children,
        then it holds at the join node $t$.
    \end{claim}
    \begin{claimproof}
        Let $\sigma\in\Sigma_t$ and $s\in\{0,\ldots,\Delta n\}$ such
        that $\DP_t[\sigma,s]=\true$. We argue that there exists a
        valid partial solution at $t$ with signature $\sigma$ and
        score $s$. By the recurrence, there are signatures
        $\alpha,\beta$ and scores $s_1,s_2$ such that
        $\alpha\oplus\beta=\sigma$, $s_1+s_2=s$, and both
        $\DP_{t_1}[\alpha,s_1]$ and $\DP_{t_2}[\beta,s_2]$ are true.
        The child invariants give valid partial solutions with these
        signatures and scores. Their edge sets and forgotten vertex
        sets are disjoint, and their selection choices for outgoing
        counters agree on $X_t$. Unite their orientations, marks,
        and choices. The combination rules give signature $\sigma$:
        incoming degrees add, capping before addition gives the same
        capped sum as capping afterwards, and residues add modulo
        $\floor{\Delta/2}+1$.
        These rules check the incoming capacities in the bag, and
        an unselected outgoing vertex has no incoming edges in either
        child. Every forgotten vertex has all its incident edges in
        one child, so its final conditions remain satisfied.
        Hence the solution is valid. Both $U_t$ and $M_t$ are disjoint
        unions of their child counterparts, so its score is $s_1+s_2=s$.

        Conversely, suppose there is a valid partial solution at $t$
        with signature $\sigma\in\Sigma_t$ and score
        $s\in\{0,\ldots,\Delta n\}$. We argue that
        $\DP_t[\sigma,s]=\true$. Restrict the solution to the two
        child subtrees. Restriction cannot increase in-degrees and
        preserves all incident edges of every vertex forgotten in
        that child, so both child solutions are valid. Let their
        signatures be $\alpha,\beta$ and their scores be $s_1,s_2$.
        Their selection choices agree on the bag, their counters
        combine to $\sigma$, and their scores sum to $s$.
        By the child invariants, $\DP_{t_1}[\alpha,s_1]$ and
        $\DP_{t_2}[\beta,s_2]$ are true. Their conjunction occurs in
        \eqref{eq:degree:treewidth-join}, so $\DP_t[\sigma,s]=\true$.
    \end{claimproof}

    \proofsubparagraph{Correctness.}
    The leaf case and
    \cref{cl:degree:introduce-vertex,cl:degree:introduce-edge,cl:degree:forget,cl:degree:join}
    establish \eqref{eq:degree:treewidth-invariant} by induction over
    the decomposition. At the empty root $r$, for a budget $k$, accept
    if $\DP_r[\varnothing,s]=\true$ for some
    \[
        s\geq\floor{\Delta/2}|\mathcal C|
            +(\floor{\Delta/2}+1)(n-k).
    \]

    Suppose the algorithm accepts. We argue that there exists a
    capacitated vertex cover of size at most $k$. By the acceptance
    condition, there is a score $s$ satisfying the displayed inequality
    such that $\DP_r[\varnothing,s]=\true$. The invariant at $r$
    gives a valid partial solution with this score. Let $S$ be its
    selected vertices. Since the root bag is empty, all vertices are
    forgotten and satisfy their final conditions. For incoming and
    outgoing counters, these conditions ensure that unselected vertices
    have no incoming edges and selected vertices respect their capacities.
    A selected cyclic vertex $v$ has at least $\floor{\Delta/2}$ outgoing
    edges and hence at most $\ceil{\Delta/2}=\capacity(v)$ incoming
    edges; an unselected one has all $\Delta$ edges marked and outgoing.
    Thus every edge has a selected head, and assigning each edge to
    its head gives a feasible capacitated cover $S$.
    Every cyclic vertex contributes a baseline of $\floor{\Delta/2}$
    points, and every unselected vertex contributes a further
    $\floor{\Delta/2}+1$ points, giving
    \begin{equation}\label{eq:degree:treewidth-root-score}
        s=\floor{\Delta/2}|\mathcal C|
            +(\floor{\Delta/2}+1)(n-|S|).
    \end{equation}
    Comparing this equality with the acceptance inequality gives
    $|S|\leq k$, as required.

    Conversely, suppose there is a capacitated vertex cover of size
    at most $k$. We argue that the algorithm accepts. Orient each edge
    towards the vertex covering it in a feasible assignment, and let
    $S$ be the vertices of positive in-degree. Then $|S|\leq k$, and
    every vertex respects its capacity. Choose outgoing-counter
    selection states according to $S$. At a cyclic vertex in $S$,
    mark any $\floor{\Delta/2}$ outgoing edges; its capacity guarantees
    that many. At a cyclic vertex outside $S$, all $\Delta$ incident
    edges are outgoing; mark all of them. This gives a valid partial
    solution at $r$ whose selected vertices are exactly $S$.
    Its score $s$ satisfies \eqref{eq:degree:treewidth-root-score}
    and hence the acceptance inequality. By the invariant at $r$,
    $\DP_r[\varnothing,s]=\true$, so the algorithm accepts.

    \proofsubparagraph{Running time.}
    Since $\Delta$ is fixed, each table has $\bO(n|\Sigma_t|)$ entries,
    and each non-join recurrence examines $\bO(1)$ child entries per
    entry. Thus non-join nodes take $|\Sigma_t|n^{\bO(1)}$ time.

    At a join, partition $\Sigma_t$ by the set of outgoing-counter
    vertices in state $\bot$; only signatures in the same part combine.
    Within each part, omit fixed coordinates and interpret
    $\false,\true$ as $0,1$. The remaining counters use ordinary,
    capped, or cyclic addition; the score is a single large-range
    addition coordinate with $\bO(n)$ possible values. Schepper's
    general convolution algorithm
    \cite[Main Theorem~I.1]{thesis/Schepper25} computes all outputs
    for a part of $D$ signatures in $Dn^{\bO(1)}$ time.
    An output is positive exactly when the corresponding entry in
    \eqref{eq:degree:treewidth-join} is true. The part sizes sum to
    $|\Sigma_t|$, so each join also takes $|\Sigma_t|n^{\bO(1)}$ time.

    Finally, $|\Sigma_t|\leq(\floor{\Delta/2}+1)^{\tw+1}$ and the
    decomposition has polynomially many nodes. Since $\Delta$ is
    fixed, the total running time is
    $\sO((\floor{\Delta/2}+1)^{\tw})$.
\end{proof}

For pathwidth, we can use a simpler DP that changes what its counters represent,
following an approach similar to that of Jansen and Nederlof~\cite{tcs/JansenN19}.
Each vertex first counts processed incoming edges. At its middle incident
edge, it guesses the remaining incoming count, checks the capacity,
and records whether it is selected. It then counts down to zero. Both counts are at most
$\floor{\Delta/2}$, so $\floor{\Delta/2}+1$ states suffice.
We encode this DP as a weighted CSP to establish the missing direction
in \cref{thm:pw-seth:equivalences}.

\begin{lemma}\label{lem:degree:pathwidth}
    If the $\pw$-SETH is false, then for every fixed $\Delta\geq3$
    there is an $\varepsilon>0$ such that {\CVCshort} on graphs of maximum
    degree at most $\Delta$, given a path decomposition of width $\pw$,
    can be solved in time $\sO\bigl((\floor{\Delta/2}+1-\varepsilon)^{\pw}\bigr)$.
\end{lemma}

\begin{proof}
    Fix $\Delta\geq3$.
    Delete isolated vertices, which do not affect the optimum, and
    let $(G,\capacity,k)$ be the resulting instance, with $n$ vertices
    and $m$ edges.
    Use a nice path decomposition with empty endpoints and
    introduce-edge nodes, and order the edges by their introduction.
    We construct a \textsc{MaxW-CSP} instance over
    $\{0,\ldots,\floor{\Delta/2}\}$ as follows.
    All Boolean variables below are restricted to $\{0,1\}$ by
    unary constraints.

    \proofsubparagraph{Counter constraints.}
    For each vertex $v$, introduce counter variables $x_{v,i}$ for
    $1\leq i<\deg(v)$, and set $x_{v,0}=x_{v,\deg(v)}=0$ as
    constants. The switch occurs at incident edge
    $s_v=\floor{\deg(v)/2}+1$. The counter $x_{v,i}$ records incoming
    edges among the first $i$ incident edges if $i<s_v$, and among
    the remaining edges otherwise. Introduce a Boolean variable $y_v$
    indicating whether $v$ is selected, used only at its switch.

    For each edge, add one constraint on its old and new endpoint
    counters and the selection bits of endpoints switching at that
    edge. A tuple is allowed exactly when one orientation satisfies
    the following rules at both endpoints. At an endpoint $v$, let
    $i$ be the incident-edge index, put $a=x_{v,i-1}$ and $b=x_{v,i}$,
    and let $\delta\in\{0,1\}$ indicate whether the edge is incoming:
    \[
        \begin{array}{c|l}
            i<s_v & b=a+\delta\\
            i=s_v & a+\delta+b\leq\capacity(v),\quad
                y_v=1\ \Longleftrightarrow\ a+\delta+b>0\\
            i>s_v & b=a-\delta
        \end{array}
    \]
    Arithmetic is over the integers, and constants are substituted
    directly into the relation. Each edge constraint has arity at most six.

    \proofsubparagraph{Objective.}
    For the weight function, set $w(0)=1$ and $w(a)=0$ for $a \in \{1,\ldots,\floor{\Delta/2}\}$.
    Each selection bit contributes one exactly when its vertex is
    unselected. Give each counter variable $x$ a private Boolean
    companion $\bar x$ with the binary constraint
    $\bar x=1\ \Longleftrightarrow\ x=0$.
    Exactly one of $x,\bar x$ is zero, so each pair contributes one.
    There are $\sum_v(\deg(v)-1)=2m-n$ counter variables.
    Thus a satisfying assignment has weight $2m-\sum_v y_v$.
    Set the target weight to $2m-k$.

    \proofsubparagraph{Correctness.}
    Suppose the CSP has a satisfying assignment of weight at least
    $2m-k$. Choose for each edge an orientation allowed by its
    constraint, obtaining an orientation $O$.
    Starting from $x_{v,0}=0$, the transitions before the switch
    give the actual past incoming counts. Working backwards from
    $x_{v,\deg(v)}=0$, the transitions after the switch give the
    actual future incoming counts. Hence $a+\delta+b$ at the switch
    equals $\indeg_O(v)$. The switch constraint enforces the capacity
    and sets $y_v=1$ exactly when this in-degree is positive.
    Thus $O$ is feasible, and the weight bound gives
    $\size(O)=\sum_v y_v\leq k$.

    Conversely, suppose $G$ has a feasible orientation $O$ of size
    at most $k$. Set the counters to its past and future incoming
    counts, and set $y_v=1$ exactly when $\indeg_O(v)>0$.
    Each counter counts at most $\floor{\deg(v)/2}$ edges, so its
    value lies in $\{0,\ldots,\floor{\Delta/2}\}$. The counts satisfy the update rules,
    and the total at each switch is $\indeg_O(v)\leq\capacity(v)$.
    Hence all edge constraints are satisfied.
    Set each companion to one exactly when its counter is zero.
    The resulting assignment has weight $2m-\size(O)\geq2m-k$.

    \proofsubparagraph{Pathwidth and running time.}
    Construct CSP bags along the given decomposition. Keep each
    counter variable $x_{v,i}$ from the introduction of $v$'s $i$-th
    incident edge through that of its $(i+1)$-st, and put $y_v$ in
    its switch bag.
    Each bag contains at most one counter per original bag vertex,
    plus at most two extra counters and two selection bits at an
    edge introduction. Insert each companion in a separate adjacent
    copy of a bag containing its counter. All constraints are now
    covered and every variable occurs consecutively. A bag has at
    most $\pw+6$ variables, giving primal pathwidth at most $\pw+5$.

    The construction has polynomial size for fixed $\Delta$.
    If the $\pw$-SETH is false, Lampis's equivalences
    \cite[Theorems~5~and~11]{soda/Lampis25} give, for some
    $\varepsilon>0$, running time $\sO\bigl((\floor{\Delta/2}+1-\varepsilon)^{\pw}\bigr)$.
    This completes the proof.
\end{proof}

\section{Cutwidth}\label{sec:cutwidth}

A \emph{linear arrangement} of a graph $G$ is an ordering $v_1,\ldots,v_n$
of its vertices. Its \emph{cutwidth} is the maximum, over $i\in[0,n]$,
of the number of edges with one endpoint among $v_1,\ldots,v_i$ and
the other among $v_{i+1},\ldots,v_n$. The \emph{cutwidth} $\ctw(G)$
is the minimum of this quantity over all linear arrangements of $G$.

We obtain a $\sO(2^{\ctw})$ algorithm and show that improving upon
the base $2$ is equivalent to falsifying the $\pw$-SETH.
A standard DP along a linear arrangement records the directions of
the edges crossing the current cut and the minimum cost of the
processed vertices. Processing each vertex's left edges before its
right edges requires at most $\ctw$ orientation bits and one counter,
giving a $\sO(2^{\ctw})$ algorithm.

To obtain the $\pw$-SETH equivalence, we encode this computation as
\textsc{SAT}, with no restriction on clause length. The counters need
only $\bO(\log n)$ additional bits, so the reduction preserves the
exponential base. This allows us to take advantage of both the standard
DP algorithm for \textsc{SAT} parameterized by its primal pathwidth and
Lampis's $\pw$-SETH equivalences~\cite{soda/Lampis25}.

\begin{lemmarep}[\appsymb]\label{lem:cutwidth:sat-reduction}
    Given an instance $(G,\capacity,k)$ of {\CVC} and a linear
    arrangement of cutwidth $\ctw$, we can construct in polynomial
    time an equivalent CNF formula and a path decomposition
    of its primal graph of width $\ctw+\bO(\log n)$, where $n=|V(G)|$.
\end{lemmarep}

\begin{proof}
    Assume $n\geq2$ and $0\leq k\leq n$, and write the ordering as
    $v_1,\ldots,v_n$. For each edge $e=v_iv_j$ with $i<j$, use a
    Boolean variable $x_e$, where $x_e=1$ means that $e$ is directed
    toward $v_j$.

    We first describe constraints on integer counters in $\{0,\ldots,n\}$.
    A global counter $s$, initially zero, counts vertices of positive
    in-degree. At $v_i$, start a local counter $a$ at zero.
    First add $x_e$ for each edge $e$ to an earlier vertex, then add
    $1-x_e$ for each edge to a later vertex.
    Require $a\leq\capacity(v_i)$ and increase $s$ by one exactly when
    $a>0$. After the last vertex, require $s\leq k$.
    Each update uses a fresh counter variable; for example,
    $a'=a+x_e$ is a constraint on three variables.

    Encode each counter in binary. Each constraint, including the
    domain restrictions, involves $\bO(\log n)$ bits, so it has only
    polynomially many assignments. For each forbidden assignment,
    add the clause excluding it. This gives a polynomial-size CNF
    formula using only the edge variables and counter bits.

    For the path decomposition, keep each $x_e$ from its first
    endpoint until its second. At each vertex, forget its left edges
    before introducing its right edges, so at most $\ctw$ edge
    variables are present. Keep the current counter values and the
    new value during an update, using $\bO(\log n)$ additional bits.
    Every constraint, and hence every clause encoding it, is contained
    in a bag. Every variable occurs in consecutive bags, giving width
    $\ctw+\bO(\log n)$.

    Finally, $a$ equals the current vertex's in-degree at its capacity
    check, and $s$ equals the orientation's size at the final test.
    Thus the constraints, and hence the resulting formula, are
    satisfiable exactly when the instance has a feasible orientation
    of size at most $k$.
\end{proof}

The standard DP algorithm for \textsc{SAT} parameterized by the primal pathwidth now gives the following.

\begin{corollary}\label{cor:cutwidth:algorithm}
    Given a capacitated graph and a linear arrangement of cutwidth
    $\ctw$, the minimum size of a feasible orientation can be found
    in time $\sO(2^{\ctw})$.
\end{corollary}

For the equivalence to the $\pw$-SETH, Lampis's
lower bound implies that an algorithm for \textsc{Vertex Cover} running in time $\sO((2-\varepsilon)^{\ctw})$
would falsify the $\pw$-SETH~\cite[Lemma~26, case $d=2$]{soda/Lampis25}.
Conversely, if the $\pw$-SETH is false, Lampis's equivalence for unrestricted
\textsc{SAT}~\cite[Theorem~5]{soda/Lampis25} gives an algorithm
running in time $\sO((2-\varepsilon)^{\pw})$ for some
$\varepsilon>0$. Applying it to the formula from \cref{lem:cutwidth:sat-reduction}
gives time $\sO((2-\varepsilon)^{\ctw+\bO(\log n)})
=\sO((2-\varepsilon)^{\ctw})$. Consequently we obtain the following.

\begin{corollary}\label{cor:cutwidth:pw-seth}
    Given a linear arrangement of cutwidth $\ctw$, {\CVC} can be solved
    in time $\sO\bigl((2-\varepsilon)^\ctw\bigr)$ for some $\varepsilon>0$
    if and only if the $\pw$-SETH is false.
\end{corollary}

\section{Tree-depth}\label{sec:td}

The \emph{tree-depth} $\td(G)$ is the minimum height of a rooted forest
on $V(G)$ such that one endpoint of every edge of $G$ is an ancestor
of the other.
In this section we strengthen the hardness result of
Dom, Lokshtanov, Saurabh, and Villanger~\cite{iwpec/DomLSV08} for bounded tree-depth graphs
by combining their reduction with the recursive verification framework of~\cite{toct/LampisV24}.
The reduction of~\cite{iwpec/DomLSV08} starts from an instance $(G,k)$ of {\kMC}
and produces an equivalent instance of {\CVC} whose graph has tree-depth $\bO(k^2)$,
which yields only an $n^{o(\sqrt{\td})}$ lower bound under the ETH.%
\footnote{Dom et al.~\cite{iwpec/DomLSV08} state a treewidth bound of $\bO(k^3)$ for their construction;
however, the same graph in fact has tree-depth $\bO(k^2)$.}
Here we keep the same high-level encoding of one chosen vertex per color class,
but replace the quadratic family of pairwise consistency gadgets with a recursive gadget.
This reduces the tree-depth to $\bO(k)$ and therefore improves the lower bound to $n^{o(\td)}$.
In particular, the standard $n^{\bO(\tw)}$ algorithm is already optimal even on graphs of bounded tree-depth.

\begin{theoremrep}[\appsymb]\label{thm:td:lb}
    For any computable function $f$,
    if there exists an algorithm that solves {\CVC} in time $f(\td) n^{o(\td)}$,
    where $\td$ denotes the tree-depth of the input graph, then the ETH is false.
\end{theoremrep}

\begin{proof}
    Recall that in \kMC, we are given a graph $G = (V, E)$ and a partition of $V$ into $k$ independent sets $V_1, \ldots, V_k$, each of size $n$,
    and we are asked to determine whether $G$ contains a $k$-clique.
    It is well-known that this problem does not admit any $f(k) n^{o(k)}$ algorithm,
    where $f$ is any computable function,
    unless the ETH is false~\cite{books/CyganFKLMPPS15}.

    Let $(G,k)$ be an instance of \kMC,
    such that every vertex of $G$ has a self loop,
    i.e., $\braces{v,v} \in E(G)$ for all $v \in V(G)$.
    We assume that~$G$ is given to us partitioned into $k$ independent sets $V_1, \ldots, V_k$,
    where $V_i = \{ v^i_1, \ldots, v^i_n \}$.
    Assume without loss of generality that $k = 2^z$ for some $z \in \mathbb{N}$
    (one can do so by adding dummy independent sets connected to all the other vertices of the graph).
    Moreover, let $E^{i_1,i_2} \subseteq E(G)$ denote the edges of $G$ with one endpoint in $V_{i_1}$ and the other in $V_{i_2}$.
    We will construct in polynomial time an equivalent instance $(H, \capacity, k')$ of \CVC,
    where $H$ is a graph of tree-depth $\td(H) = \bO(k)$,
    $\capacity \colon V \to \mathbb{N}$ is a capacity function,
    and $k'$ is an integer,
    such that $G$ has a $k$-clique if and only if $H$ has a capacitated vertex cover of size at most $k'$.

    We move on to describe the construction of graph $H$.
    As in~\cite{iwpec/DomLSV08}, when we say that a vertex is \emph{marked}, we attach to it $k'+1$ leaves, each of which of capacity $1$.
    Furthermore, when we say that we connect vertices $u$ and $v$ via an \emph{$A$-edge}, we introduce an independent set on
    $A$ marked vertices, each of capacity $(k'+1) + 1$, and add edges between each vertex of the independent set and vertices $u,v$.

    \begin{figure}[ht]
        \centering
        \begin{tikzpicture}[scale=0.8, transform shape]
            \node[vertex] (u1) at (4,4.5) {};
            \node[] () at (3.6,4.5) {$u_1$};
            \node[vertex] (u2) at (6,4.5) {};
            \node[] () at (6.4,4.5) {$u_2$};

            \node[black_vertex] (v1) at (5,5) {};
            % \node[] () at (5,4.6) {$\vdots$};
            \node[black_vertex] (v2) at (5,4.5) {};
            \node[black_vertex] (v3) at (5,4) {};

            \draw[] (u1)--(v1)--(u2);
            \draw[] (u1)--(v2)--(u2);
            \draw[] (u1)--(v3)--(u2);
        \end{tikzpicture}
        \caption{Example of a $3$-edge connecting vertices $u_1$ and $u_2$.
        Black vertices have $k'+1$ leaves attached and are of capacity $k'+2$.}
        \label{fig:td_lb_A_edge}
    \end{figure}

    \proofsubparagraph*{Choice and Copy Gadgets.}
    For an independent set $V_i$, we construct the \emph{choice gadget} $\hat{C}_i$ as follows.
    We introduce vertices $v^i_j$ for $j \in [n]$, as well as a marked vertex $\hat{x}_i$.
    We add all edges $\{ \hat{x}_i, v^i_j \}$.
    As for the capacities of the introduced vertices, we set $\capacity(\hat{x}_i) = \deg_H (\hat{x}_i) - 1$, while $\capacity(v^i_j) = \deg_H (v^i_j)$.
    Given two instances $\mathcal{I}_1,\mathcal{I}_2$ of a choice gadget $\hat{C}_i$,
    when we say that we add a \emph{copy gadget $(\mathcal{I}_1, \mathcal{I}_2)$},
    we introduce two marked \emph{copy vertices} $\hat{g}_1,\hat{g}_2$,
    and for every vertex $v^i_j$ of $\mathcal{I}_1$ (resp., $\mathcal{I}_2$),
    we introduce a $j$-edge (resp., $(n-j)$-edge) between $v^i_j$ and $\hat{g}_1$,
    as well as a $(n-j)$-edge (resp., $j$-edge) between $v^i_j$ and $\hat{g}_2$.
    We set $\capacity(\hat{g}_1) = \deg_H (\hat{g}_1) - n$ and $\capacity(\hat{g}_2) = \deg_H (\hat{g}_2) - n$.
    See also \Cref{fig:td_lb_copy_choice_gadget} for an illustration.

    \begin{figure}[ht]
        \centering
        \begin{tikzpicture}[scale=0.8, transform shape]

            %%%%%%%%%%%%%%%%%%%%%%%%%%%%%%%% I_1
            \node[] () at (4.5,7) {$\mathcal{I}_1$};

            \node[vertex] (v11) at (5,6) {};
            \node[] () at (4.6,6.1) {$v^i_1$};

            \node[] () at (5,5.26) {$\vdots$};

            \node[black_vertex] (x1) at (4,4.5) {};
            \node[] () at (3.6,4.5) {$\hat{x}_i$};

            \node[vertex] (v12) at (5,4.5) {};
            \node[] () at (4.6,4.8) {$v^i_j$};

            \node[] () at (5,3.76) {$\vdots$};

            \node[vertex] (v1n) at (5,3) {};
            \node[] () at (4.6,3) {$v^i_n$};

            \draw[] (x1)--(v11);
            \draw[] (x1)--(v12);
            \draw[] (x1)--(v1n);

            \node[black_vertex] (g1) at (7,5) {};
            \node[] () at (7,5.4) {$\hat{g}_1$};
            \node[black_vertex] (g2) at (7,4) {};
            \node[] () at (7,3.6) {$\hat{g}_2$};

            \draw[dashed] (g1)--(v12) node[midway, above] {$j$};
            \draw[dashed] (g2)--(v12) node[midway, below] {$n-j$};

            % \draw[dashed] (g1)--(v11);
            % \draw[dashed] (g2)--(v11);
            % \draw[dashed] (g1)--(v1n);
            % \draw[dashed] (g2)--(v1n);

            \begin{scope}[shift={(4,0)}]
                %%%%%%%%%%%%%%%%%%%%%%%%%%%%%%%% I_2

                \node[] () at (5.5,7) {$\mathcal{I}_2$};

                \node[vertex] (v21) at (5,6) {};
                \node[] () at (5.4,6.1) {$v^i_1$};

                \node[] () at (5,5.26) {$\vdots$};

                \node[black_vertex] (x2) at (6,4.5) {};
                \node[] () at (6.4,4.5) {$\hat{x}_i$};

                \node[vertex] (v22) at (5,4.5) {};
                \node[] () at (5.4,4.8) {$v^i_j$};

                \node[] () at (5,3.76) {$\vdots$};

                \node[vertex] (v2n) at (5,3) {};
                \node[] () at (5.4,3) {$v^i_n$};

                \draw[] (x2)--(v21);
                \draw[] (x2)--(v22);
                \draw[] (x2)--(v2n);

                \draw[dashed] (g1)--(v22) node[midway, above] {$n-j$};
                \draw[dashed] (g2)--(v22) node[midway, below] {$j$};

            \end{scope}

        \end{tikzpicture}
        \caption{Addition of a copy gadget $(\mathcal{I}_1, \mathcal{I}_2)$.
        Marked vertices are in black, and dashed edges are used for $A$-edges, with $A$ being written on top of them,
        where only those incident to $v^i_j$ are depicted for the sake of readability.}
        \label{fig:td_lb_copy_choice_gadget}
    \end{figure}

    \proofsubparagraph*{Adjacency Gadget.}
    We now proceed to describe the \emph{adjacency gadgets} of our construction.%
    \footnote{For some high-level intuition regarding the adjacency gadgets,
    we refer the reader to~\cite[Section~4]{toct/LampisV24}.}
    Let $i,i' \in [k]$.
    We construct the adjacency gadget $\hat{A}(i, i, i', i')$ as follows
    (for an illustration see \Cref{fig:td_lb_adj_gadget}).
    \begin{itemize}
        \item We introduce instances $\mathcal{I}_1$ and $\mathcal{I}_2$ of the choice gadgets $\hat{C}_{i}$ and $\hat{C}_{i'}$ respectively.

        \item We introduce a marked vertex $\hat{x}_{i,i'}$, as well as a vertex $v_e$ for every edge $e \in E^{i,i'}$.
        We add all edges $\{ \hat{x}_{i,i'}, v_e \}$.
        As for the capacities of the introduced vertices,
        we set $\capacity(\hat{x}_{i,i'}) = \deg_H (\hat{x}_{i,i'}) - 1$
        as well as $\capacity(v_e) = \deg_H (v_e)$ for all vertices $v_e$.

        \item We introduce marked vertices $\hat{\alpha}_{i,i'}, \hat{\beta}_{i,i'}, \hat{\kappa}_{i,i'}, \hat{\lambda}_{i,i'}$,
        which we refer to as \emph{validation vertices}.
        For all $\hat{z}_{i,i'} \in \{ \hat{\alpha}_{i,i'}, \hat{\beta}_{i,i'}, \hat{\kappa}_{i,i'}, \hat{\lambda}_{i,i'} \}$,
        we set $\capacity(\hat{z}_{i,i'}) = \deg_H (\hat{z}_{i,i'}) - n$.

        \item For every vertex $v^i_j$ in $\mathcal{I}_1$, we add a $j$-edge (resp., a $(n-j)$-edge)
        between $v^i_j$ and $\hat{\alpha}_{i,i'}$ (resp., $\hat{\beta}_{i,i'}$).

        \item For every vertex $v^{i'}_{j'}$ in $\mathcal{I}_2$, we add a $j'$-edge (resp., a $(n-j')$-edge)
        between $v^{i'}_{j'}$ and $\hat{\kappa}_{i,i'}$ (resp., $\hat{\lambda}_{i,i'}$).

        \item For $e = \{ v^i_j, v^{i'}_{j'} \} \in E^{i,i'}$,
        we add
        \begin{multicols}{2}
        \begin{itemize}
            \item a $(n-j)$-edge between $v_e$ and $\hat{\alpha}_{i,i'}$,
            \item a $j$-edge between $v_e$ and $\hat{\beta}_{i,i'}$,
            \item a $(n-j')$-edge between $v_e$ and $\hat{\kappa}_{i,i'}$,
            \item a $j'$-edge between $v_e$ and $\hat{\lambda}_{i,i'}$.
        \end{itemize}
        \end{multicols}
    \end{itemize}

    Now let $i_1, i_2, i'_1, i'_2 \in [k]$, where $i_1 < i_2$ and $i'_1 < i'_2$.
    Then, let the adjacency gadget $\hat{A}(i_1,i_2,i'_1,i'_2)$ contain instances of choice gadgets $\hat{C}_{i}$ and $\hat{C}_{i'}$,
    where $i \in [i_1,i_2]$ and $i' \in [i'_1, i'_2]$,
    which we will refer to as the \emph{original choice gadgets} of $\hat{A}(i_1,i_2,i'_1,i'_2)$,
    as well as the adjacency gadgets
    \begin{multicols}{2}
    \begin{itemize}
        \item $\hat{A}\parens*{i_1, \floor*{\frac{i_1+i_2}{2}}, i'_1, \floor*{\frac{i'_1+i'_2}{2}}}$,
        \item $\hat{A}\parens*{i_1, \floor*{\frac{i_1+i_2}{2}}, \ceil*{\frac{i'_1+i'_2}{2}}, i'_2}$,
        \item $\hat{A}\parens*{\ceil*{\frac{i_1+i_2}{2}}, i_2, i'_1, \floor*{\frac{i'_1+i'_2}{2}}}$,
        \item $\hat{A}\parens*{\ceil*{\frac{i_1+i_2}{2}}, i_2, \ceil*{\frac{i'_1+i'_2}{2}}, i'_2}$.
    \end{itemize}
    \end{multicols}
    Lastly, let $\mathcal{I}$ denote the original choice gadget $\hat{C}_p$, where $p \in [i_1, i_2] \cup [i'_1, i'_2]$.
    Notice that there are exactly two instances of the choice gadget $\hat{C}_p$ appearing as original choice gadgets in the adjacency gadgets just introduced,
    say instances $\mathcal{I}_1$ and $\mathcal{I}_2$.
    Add copy gadgets $(\mathcal{I}, \mathcal{I}_1)$ and $(\mathcal{I}, \mathcal{I}_2)$.

    \begin{figure}[ht]
    \centering
    \begin{tikzpicture}[scale=0.8, transform shape]
        %%%%%%%%%%%%%%%%%%%%%%%%%%%%%%%% I_1
        \node[vertex] (v11) at (5,6) {};
        \node[] () at (4.6,6.1) {$v^i_1$};

        \node[] () at (5,5.26) {$\vdots$};

        \node[black_vertex] (x1) at (4,4.5) {};
        \node[] () at (3.6,4.5) {$\hat{x}_i$};

        \node[vertex] (v12) at (5,4.5) {};
        \node[] () at (4.6,4.8) {$v^i_j$};

        \node[] () at (5,3.76) {$\vdots$};

        \node[vertex] (v1n) at (5,3) {};
        \node[] () at (4.6,3) {$v^i_n$};

        \draw[] (x1)--(v11);
        \draw[] (x1)--(v12);
        \draw[] (x1)--(v1n);

        \node[black_vertex] (a) at (7,5) {};
        \node[] () at (7,5.4) {$\hat{\alpha}_{i,i'}$};
        \node[black_vertex] (b) at (7,4) {};
        \node[] () at (7,3.6) {$\hat{\beta}_{i,i'}$};

        \draw[dashed] (a)--(v12) node[midway, above] {$j$};
        \draw[dashed] (b)--(v12) node[midway, below] {$n-j$};

        \begin{scope}[shift={(4,0)}]
            \node[black_vertex] (k) at (7,5) {};
            \node[] () at (7,5.3) {$\hat{\kappa}_{i,i'}$};
            \node[black_vertex] (l) at (7,4) {};
            \node[] () at (7,3.7) {$\hat{\lambda}_{i,i'}$};
        \end{scope}

        \begin{scope}[shift={(8,0)}]
            \node[vertex] (v21) at (5,6) {};
            \node[] () at (5.4,6.1) {$v^{i'}_1$};

            \node[] () at (5,5.26) {$\vdots$};

            \node[black_vertex] (x2) at (6,4.5) {};
            \node[] () at (6.4,4.5) {$\hat{x}_{i'}$};

            \node[vertex] (v22) at (5,4.5) {};
            \node[] () at (5.4,4.8) {$v^{i'}_{j'}$};

            \node[] () at (5,3.76) {$\vdots$};

            \node[vertex] (v2n) at (5,3) {};
            \node[] () at (5.4,3) {$v^{i'}_n$};

            \draw[] (x2)--(v21);
            \draw[] (x2)--(v22);
            \draw[] (x2)--(v2n);

            \draw[dashed] (k)--(v22) node[midway, above] {$j'$};
            \draw[dashed] (l)--(v22) node[midway, below] {$n-j'$};
        \end{scope}

        \begin{scope}[shift={(2,0)}]
            \node[vertex] (u1) at (7,6) {};
            \node[] () at (7,5.26) {$\vdots$};
            \node[vertex] (u2) at (7,4.5) {};
            \node[] () at (7,3.76) {$\vdots$};
            \node[vertex] (u3) at (7,3) {};
            \draw[dashed] (a)--(u2) node[midway, above] {$n-j$};
            \draw[dashed] (b)--(u2) node[midway, below] {$j$};
            \draw[dashed] (k)--(u2) node[midway, above] {$n-j'$};
            \draw[dashed] (l)--(u2) node[midway, below] {$j'$};

            \node[] () at (8.5,7) {$\hat{x}_{i,i'}$};
            \node[black_vertex] (xx) at (8,7) {};
            \draw[] (xx)--(u1);
            \draw[] (xx)--(u2);
            \draw[] (xx) edge [bend left] (u3);
        \end{scope}

    \end{tikzpicture}
    \caption{Adjacency gadget $\hat{A}(i,i,i',i')$.
    Marked vertices are in black.
    Dashed edges are used for $A$-edges, with $A$ being written on top of them;
    only those incident to $v^i_j$, $v^{i'}_{j'}$, and $v_e$ are depicted for the sake of readability,
    where $e = \{ v^i_j, v^{i'}_{j'} \} \in E^{i,i'}$.}
    \label{fig:td_lb_adj_gadget}
    \end{figure}

    Let graph $H$ be the adjacency gadget $\hat{A}(1,k,1,k)$ and let $M \subseteq V(H)$ denote the set of its marked vertices.
    Set $k' = k^2 + \gamma + \delta$, where $\gamma = 2k (2k - 1)$ and $\delta = |M|$.
    This concludes the construction of the instance $(H, \capacity, k')$.

    \begin{claim}\label{claim:td_lb_helper}
        For all $i,i' \in [k]$, there is exactly one instance of the adjacency gadget $\hat{A}(i,i,i',i')$ present in $H$.
        Moreover, the number of instances of choice gadgets present in $H$ is equal to $\gamma$.
    \end{claim}

    \begin{claimproof}
        First observe that every adjacency gadget $\hat{A}(i_1,i_2,i'_1,i'_2)$ appearing in $H$
        satisfies
        \[
            i_2-i_1+1 = i'_2-i'_1+1 = 2^c
        \]
        for some non-negative integer $c$.
        This is true for the root gadget $\hat{A}(1,k,1,k)$ since $k$ is a power of two,
        and it is preserved by the recursive construction because each interval is always split exactly in half.
        Hence every ordered pair $(i,i') \in [k]^2$ determines a unique root-to-leaf path in the recursion tree,
        and therefore a unique leaf gadget $\hat{A}(i,i,i',i')$.

        For the number of choice gadget instances, let $T(\ell)$ denote the number of instances appearing in
        $\hat{A}(i_1,i_2,i'_1,i'_2)$ when both intervals have size $\ell$.
        Then $T(1)=2$ and, for $\ell>1$,
        \[
            T(\ell)=2\ell+4T(\ell/2),
        \]
        because the gadget contains $2\ell$ original choice gadgets and four recursive subgadgets.
        Consequently,
        \[
            T(k) = \sum_{r=0}^{\log k} 4^r \cdot 2\frac{k}{2^r}
            = 2k \sum_{r=0}^{\log k} 2^r
            = 2k(2k-1)
            = \gamma,
        \]
        and the statement follows.
    \end{claimproof}

    \begin{claim}\label{claim:td_lb_td}
        It holds that $\td(H) = \bO(k)$.
    \end{claim}

    \begin{claimproof}
        Let $T(\ell)$ denote the maximum tree-depth of an adjacency gadget
        $\hat{A}(i_1,i_2,i'_1,i'_2)$ where both intervals have size $\ell$.
        By \Cref{claim:td_lb_helper}, only powers of two occur.

        For the base case $\ell=1$, the gadget is $\hat{A}(i,i,i',i')$.
        After removing the four validation vertices
        $\hat{\alpha}_{i,i'}$, $\hat{\beta}_{i,i'}$, $\hat{\kappa}_{i,i'}$, and $\hat{\lambda}_{i,i'}$,
        every remaining connected component has constant tree-depth,
        hence $T(1)=\bO(1)$.

        Consider now a gadget $\hat{A}(i_1,i_2,i'_1,i'_2)$ with interval size $\ell>1$.
        It contains the four recursive adjacency gadgets
        \begin{multicols}{2}
            \begin{itemize}
                \item $\hat{A}\parens*{i_1, \floor*{\frac{i_1+i_2}{2}}, i'_1, \floor*{\frac{i'_1+i'_2}{2}}}$,
                \item $\hat{A}\parens*{i_1, \floor*{\frac{i_1+i_2}{2}}, \ceil*{\frac{i'_1+i'_2}{2}}, i'_2}$,
                \item $\hat{A}\parens*{\ceil*{\frac{i_1+i_2}{2}}, i_2, i'_1, \floor*{\frac{i'_1+i'_2}{2}}}$,
                \item $\hat{A}\parens*{\ceil*{\frac{i_1+i_2}{2}}, i_2, \ceil*{\frac{i'_1+i'_2}{2}}, i'_2}$,
            \end{itemize}
        \end{multicols}
        and exactly $2\ell$ original choice gadgets.
        Each original choice gadget is linked to the recursive part only through two copy gadgets,
        that is, through four copy vertices.
        Removing all these copy vertices disconnects the gadget into the four recursive adjacency gadgets
        and a collection of components of constant tree-depth.
        Since we remove $8\ell$ vertices, we obtain
        \[
            T(\ell) \leq 8\ell + T(\ell/2).
        \]
        Consequently,
        \[
            T(k) \le 8 \sum_{r=0}^{\log k} \frac{k}{2^r} = \bO(k),
        \]
        and the statement follows.
    \end{claimproof}

    \begin{claim}\label{claim:td_lb_cor1}
        If $G$ contains a $k$-clique,
        then $(H, \capacity, k')$ is a Yes instance of \CVC.
    \end{claim}

    \begin{claimproof}
        Consider $s \colon [k] \to [n]$ such that $\mathcal{V} = \setdef{v^i_{s(i)}}{i \in [k]} \subseteq V(G)$ is a $k$-clique of $G$.
        Let $S \subseteq V(H)$ be the subset of vertices of $H$ such that
        \begin{itemize}
            \item $S$ contains all vertices of $M$,

            \item for all $i \in [k]$, $S$ contains the vertex $v^i_{s(i)}$ of every instance of the choice gadget $\hat{C}_i$,

            \item for every ordered pair $(i,i') \in [k]^2$, if $e_{i,i'}$ denotes the edge
            $\{v^i_{s(i)}, v^{i'}_{s(i')}\}$ of $G$, then $S$ contains the vertex $v_{e_{i,i'}}$
            of the unique leaf gadget $\hat{A}(i,i,i',i')$.
        \end{itemize}
        By \Cref{claim:td_lb_helper} it follows that $|S| = \delta + \gamma + k^2 = k'$.
        It suffices to argue that $S$ is a capacitated vertex cover of $H$.

        First, notice that every edge incident to vertices $\hat{x}_i$ and $\hat{x}_{i,i'}$ is indeed covered:
        the corresponding marked vertex covers all but one of its incident edges, and $S$ contains one of its neighbors
        whose capacity is equal to its degree.

        Next we argue about edges present in choice gadgets.
        Let $\mathcal{I}_1$ and $\mathcal{I}_2$ denote two instances of the choice gadget $\hat{C}_i$,
        connected by a copy gadget which contains vertices $\hat{g}_1$ and $\hat{g}_2$ (see \Cref{fig:td_lb_copy_choice_gadget}).
        Recall that $\hat{g}_1$ is connected via a $j$-edge (resp., $(n-j)$-edge) with vertex $v^i_j$ of $\mathcal{I}_1$ (resp., of $\mathcal{I}_2$).
        Let $Z \subseteq M$ denote the set of the $n$ marked vertices that are neighbors of $\hat{g}_1$ and are introduced
        due to the $j$-edge and the $(n-j)$-edge
        that connects $\hat{g}_1$ with the vertex $v^i_{s(i)}$ of instances $\mathcal{I}_1$ and $\mathcal{I}_2$ respectively.
        Let $\hat{g}_1$ cover all of its incident edges apart from those connecting it with vertices in $Z$;
        in that case, the vertices of $N_H(\hat{g}_1) \setminus Z$ can cover the edge that connects them with vertices in the respective choice gadget.
        As for the edges between $v^i_{s(i)}$ and $Z$, recall that $v^i_{s(i)} \in S$ with $\capacity(v^i_{s(i)}) = \deg_H (v^i_{s(i)})$.
        Similarly, the edges concerning $\hat{g}_2$ or its incident vertices are also covered.
        Lastly, in an analogous way one can argue that all edges present in an adjacency gadget $\hat{A}(i,i,i',i')$ are covered.
    \end{claimproof}

    \begin{claim}\label{claim:td_lb_cor2}
        If $(H, \capacity, k')$ is a Yes instance of \CVC,
        then $G$ contains a $k$-clique.
    \end{claim}

    \begin{claimproof}
        Let $S \subseteq V(H)$ be a capacitated vertex cover of $H$ of cardinality $|S| \le k'$.
        We first argue that $M \subseteq S$.
        Indeed, let $v \in M \setminus S$, in which case it follows that $N_H(v) \subseteq S$,
        implying that $|S| \ge k'+1$, a contradiction.
        It follows that $S$ contains $\delta$ marked vertices.

        Let $\hat{X}$ be the set containing
        \begin{itemize}
            \item the marked vertex of every instance of a choice gadget, and
            \item the marked vertex $\hat{x}_{i,i'}$ of every leaf gadget $\hat{A}(i,i,i',i')$.
        \end{itemize}
        Notice that for all $\hat{x} \in \hat{X}$ it holds that $\capacity(\hat{x}) = \deg_H (\hat{x}) - 1$,
        while the neighborhoods of any distinct $\hat{x}, \hat{x}' \in \hat{X}$ are disjoint and composed of non-marked vertices.
        Consequently, it follows that $|S \cap N_H(\hat{x})| \ge 1$, for a total of at least $\gamma + k^2$ vertices due to \Cref{claim:td_lb_helper}.
        Since $|S| \le k' = \delta + \gamma + k^2$, and due to the previous paragraph,
        it follows that $|S \cap N_H(\hat{x})| = 1$, for all $\hat{x} \in \hat{X}$.

        Now we show that there exists a function $s \colon [k] \to [n]$ such that,
        for $i \in [k]$, for every instance of the choice gadget $\hat{C}_i$
        it holds that $v^i_{s(i)} \in S$.
        Assume there exists a copy gadget $(\mathcal{I}_1, \mathcal{I}_2)$,
        where $\mathcal{I}_1$ and $\mathcal{I}_2$ are instances of $\hat{C}_i$
        and $\hat{g}_1, \hat{g}_2$ denote the marked vertices introduced due to the copy gadget.
        Assume for the sake of contradiction that there exist vertices $v^i_{j_1}$ in $\mathcal{I}_1$
        and $v^i_{j_2}$ in $\mathcal{I}_2$ with $v^i_{j_1}, v^i_{j_2} \in S$, and no other vertex of those choice gadgets belongs to $S$.
        Let $Z_1$ (resp., $Z_2$) denote the set of neighbors of $\hat{g}_1$ (resp., $\hat{g}_2$) that are neighboring
        \begin{itemize}
            \item either the vertices of $\mathcal{I}_1$ except from $v^i_{j_1}$,
            \item or the vertices of $\mathcal{I}_2$ except from $v^i_{j_2}$.
        \end{itemize}
        It follows that
        \[
            |Z_1| = 2\sum_{w=1}^n w - (j_1 + n - j_2) \qquad
            \text{and} \qquad
            |Z_2| = 2\sum_{w=1}^n w - (n - j_1 + j_2).
        \]
        For every vertex $z \in Z_1 \cup Z_2$ it holds that its neighbor in the corresponding choice gadget does not belong to $S$,
        thus $z$ has to use its capacity in order to cover (i) its attached leaves, and (ii) the edge that it has with the vertex in the choice gadget.
        Consequently, it follows that $\hat{g}_1$ (resp., $\hat{g}_2$) has to use its capacity in order to cover all its incident edges with vertices
        in $Z_1$ (resp., $Z_2$), as well as its edges towards its leaves.
        It follows that
        \begin{itemize}
            \item $\capacity(\hat{g}_1) \ge |Z_1| + k'+1 \iff n \le j_1 + n - j_2 \iff j_2 \le j_1$, and
            \item $\capacity(\hat{g}_2) \ge |Z_2| + k'+1 \iff n \le n - j_1 + j_2 \iff j_1 \le j_2$.
        \end{itemize}
        If $j_1 \neq j_2$ at least one of the previous two inequalities cannot hold, leading to a contradiction.
        Hence any two choice gadget instances of type $\hat{C}_i$ that are connected by a copy gadget select the same index.
        Since every instance of $\hat{C}_i$ is connected to the original copy of $\hat{C}_i$ in the root gadget by a path of copy gadgets,
        there exists a well-defined function $s \colon [k] \to [n]$ such that,
        in every instance of $\hat{C}_i$, the unique selected vertex is $v^i_{s(i)}$.

        Let $\mathcal{V} = \setdef{v^i_{s(i)}}{i \in [k]} \subseteq V(G)$,
        where $|\mathcal{V} \cap V_i| = 1$, for all $i \in [k]$.
        We will prove that $\mathcal{V}$ is a clique.
        Let $i, i' \in [k]$ and consider the adjacency gadget $\hat{A}(i,i,i',i')$.
        Analogously to the previous paragraph, one can argue that the unique vertex $v_e \in S \cap N_H(\hat{x}_{i,i'})$
        must be connected via
        (i) a $(n-s(i))$-edge with $\hat{\alpha}_{i,i'}$,
        (ii) a $s(i)$-edge with $\hat{\beta}_{i,i'}$,
        (iii) a $(n-s(i'))$-edge with $\hat{\kappa}_{i,i'}$,
        and (iv) a $s(i')$-edge with $\hat{\lambda}_{i,i'}$.
        By construction of the leaf gadget, this is possible only when
        $e = \{ v^i_{s(i)}, v^{i'}_{s(i')} \} \in E(G)$.
        Since this holds for any two vertices belonging to~$\mathcal{V}$, it follows that $G$ has a $k$-clique.
    \end{claimproof}
    Therefore, in polynomial time we can construct a capacitated graph $H$ of tree-depth $\td = \bO(k)$ due to \Cref{claim:td_lb_td},
    such that, due to \Cref{claim:td_lb_cor1,claim:td_lb_cor2},
    deciding whether $H$ has a capacitated vertex cover of size at most $k'$
    is equivalent to deciding whether $G$ has a $k$-clique.
\end{proof}

\section{Clique-width}\label{sec:cw}

A \emph{$q$-expression} constructs a labeled graph using labels from $[q]$
and four operations: creating a labeled vertex, taking disjoint unions,
adding all edges between two distinct label classes, and changing all
occurrences of one label to another. The \emph{clique-width} $\cw(G)$
is the minimum $q$ for which a $q$-expression constructs $G$, ignoring labels.
An expression is \emph{linear} if every disjoint union adds a single vertex
to the graph built so far; the minimum number of labels under this restriction
is the \emph{linear clique-width} $\lcw(G)$.

Our main result in this section is that {\CVC} is NP-complete even when restricted to graphs of bounded clique-width.
Before presenting the proof, let us briefly give some intuition.
A natural approach for obtaining a polynomial-time algorithm on graphs of constant clique-width would be dynamic
programming over label classes, storing for each class some amount of information bounded by a polynomial of $n$.
For instance, one could try to remember, for each label class, the \emph{aggregate} remaining capacity of the selected vertices in a class.
Unfortunately, respecting the aggregate capacity of a label class is necessary but not sufficient for feasibility:
one must ensure that the capacity of each \emph{individual} vertex is respected.
This is precisely the feature exploited by our reduction,
where we use vertex capacities to encode the incidence of literals in clauses of a \OneInThreeSAT\ instance.

We remark that the phenomenon whereby a problem that is W[1]-hard (but in XP)
parameterized by treewidth becomes paraNP-complete parameterized by clique-width has been observed before.
Similar behavior is exhibited, for example, by \textsc{Grundy Coloring}~\cite{siamdm/BelmonteKLMO22},
\textsc{List Coloring}~\cite{JansenS97}, and \textsc{Defective Coloring}~\cite{BelmonteLM22}.
Our reduction is somewhat similar in spirit to the one of~\cite{siamdm/BelmonteKLMO22},
as we also reduce from a variant of \textsc{SAT} and use extra problem-specific
information to encode literal-clause incidences while representing most of the instance by essentially a large complete bipartite graph.
However, our reduction is technically more involved, since encoding this information through capacities
is subtler than doing so through colors.

\begin{theoremrep}[\appsymb]\label{thm:cw:np-hard}
    {\CVC} is NP-hard even for graphs of linear clique-width at most $6$.
\end{theoremrep}

\begin{proof}
    We reduce from \OneInThreeSAT, which is NP-complete~\cite[LO4]{books/GareyJ79}.
    Recall that an instance of {\OneInThreeSAT} is a CNF formula in which every clause has exactly three literals,
    and the question is whether there exists a truth assignment that makes exactly one literal true in every clause.

    Let $\psi$ be an instance of {\OneInThreeSAT} with variables $X=\{x_1,\ldots,x_n\}$ and clauses $c_1,\ldots,c_m$.
    For each $j \in [m]$, write $c_j=(\ell_{j,i_1}\lor \ell_{j,i_2}\lor \ell_{j,i_3})$,
    where $\ell_{j,i}$ is either $x_{i}$ or $\neg x_{i}$.

    We construct an equivalent instance $\mathcal I=(G,\capacity,k)$ of \CVC.
    Throughout the proof, saying that a vertex has \emph{demand} $d$ means that $\deg_G(v)-\capacity(v)=d$.
    Also, when we introduce a \emph{marked} vertex, we additionally attach $k+1$ pendant neighbors to it.
    Consequently, every marked vertex must belong to any solution of size at most $k$.

    We now describe the graph $G$.
    For each clause we create a positive side and a negative side.
    Intuitively, the positive side tracks the unique true literal of the clause,
    while the negative side tracks the two false literals.
    For each variable $x_i$ we create two selector vertices $v_i$ and $\overline{v}_i$,
    corresponding to setting $x_i$ to \true\ and \false, respectively.

    \proofsubparagraph*{Clause gadgets.}
    For each clause $c_j = (\ell_{j,i_1} \lor \ell_{j,i_2} \lor \ell_{j,i_3})$, where $j \in [m]$, we introduce:
    \begin{itemize}
        \item a marked vertex $c^+_j$ of demand $3(j-1)+1$;
        \item three marked vertices $\ell^+_{j,i_1}, \ell^+_{j,i_2}, \ell^+_{j,i_3}$, each of demand $j$;
        \item a marked vertex $c^-_j$ of demand $3(j-1)+2$;
        \item three marked vertices $\ell^-_{j,i_1}, \ell^-_{j,i_2}, \ell^-_{j,i_3}$, each of demand $j$.
    \end{itemize}
    We let
    $C^+ = \setdef{c^+_j}{j \in [m]}$,
    $C^- = \setdef{c^-_j}{j \in [m]}$,
    $L^+$ to be the set of all vertices $\ell^+_{j,i}$ introduced, and
    $L^-$ to be the set of all vertices $\ell^-_{j,i}$ introduced.
    For each $j \in [m]$, let
    $L^+_j = \{ \ell^+_{j,i_1}, \ell^+_{j,i_2}, \ell^+_{j,i_3} \}$ and
    $L^-_j = \{ \ell^-_{j,i_1}, \ell^-_{j,i_2}, \ell^-_{j,i_3} \}$.
    We add all edges between $C^+$ and $L^+$, and all edges between $C^-$ and $L^-$.
    See \cref{fig:cw_clause_gadget} for an illustration of the two sides of the construction on two consecutive clauses.

    \begin{figure}[ht]
        \centering
        \begin{tikzpicture}[scale=0.8, transform shape, decorate]
            % positive side
            \node[black_vertex] (cpj) at (1.5,3.5) {};
            \node[] () at (1.5,4.0) {$c_j^+$};
            \node[draw=black!35, rounded corners, fill=white, inner sep=1.5pt] () at (2.8,4) {\scriptsize $d=3(j-1)+1$};
            \node[black_vertex] (cpjp) at (5.9,3.5) {};
            \node[] () at (5.9,4.0) {$c_{j+1}^+$};
            \node[draw=black!35, rounded corners, fill=white, inner sep=1.5pt] () at (7,4) {\scriptsize $d=3j+1$};

            \node[black_vertex] (lp1) at (0.0,1.15) {};
            \node[black_vertex] (lp2) at (1.5,1.15) {};
            \node[black_vertex] (lp3) at (3.0,1.15) {};

            \node[black_vertex] (lq1) at (4.4,1.15) {};
            \node[black_vertex] (lq2) at (5.9,1.15) {};
            \node[black_vertex] (lq3) at (7.4,1.15) {};

            \foreach \u in {cpj,cpjp}{
                \foreach \v in {lp1,lp2,lp3,lq1,lq2,lq3}{
                    \draw[thin,gray!70] (\u)--(\v);
                }
            }

            % representative intended orientation on positive side
            \draw[->] (cpj) -- (lp2);
            \draw[->] (cpjp) -- (lp1);
            \draw[->] (cpjp) -- (lp2);
            \draw[->] (cpjp) -- (lp3);
            \draw[->] (lq1) -- (cpj);
            \draw[->] (lq2) -- (cpj);
            \draw[->] (lq3) -- (cpj);

            \draw[decorate,decoration={brace,mirror,amplitude=4pt}] (-0.25,0.45) -- (3.25,0.45);
            \node[] () at (1.5,-0.05) {$L_j^+$};
            \node[] () at (1.5,-0.55) {\scriptsize each demand $j$};

            \draw[decorate,decoration={brace,mirror,amplitude=4pt}] (4.15,0.45) -- (7.65,0.45);
            \node[] () at (5.9,-0.05) {$L_{j+1}^+$};
            \node[] () at (5.9,-0.55) {\scriptsize each demand $j+1$};

            \node[] () at (3.7,4.55) {\small positive side};

            % negative side
            \begin{scope}[shift={(9,0)}]
                \node[black_vertex] (cmj) at (1.5,3.5) {};
                \node[] () at (1.5,4.0) {$c_j^-$};
                \node[draw=black!35, rounded corners, fill=white, inner sep=1.5pt] () at (2.8,4) {\scriptsize $d=3(j-1)+2$};
                \node[black_vertex] (cmjp) at (5.9,3.5) {};
                \node[] () at (5.9,4.0) {$c_{j+1}^-$};
                \node[draw=black!35, rounded corners, fill=white, inner sep=1.5pt] () at (7,4) {\scriptsize $d=3j+2$};

                \node[black_vertex] (lm1) at (0.0,1.15) {};
                \node[black_vertex] (lm2) at (1.5,1.15) {};
                \node[black_vertex] (lm3) at (3.0,1.15) {};

                \node[black_vertex] (ln1) at (4.4,1.15) {};
                \node[black_vertex] (ln2) at (5.9,1.15) {};
                \node[black_vertex] (ln3) at (7.4,1.15) {};

                \foreach \u in {cmj,cmjp}{
                    \foreach \v in {lm1,lm2,lm3,ln1,ln2,ln3}{
                        \draw[thin,gray!70] (\u)--(\v);
                    }
                }

                % representative intended orientation on negative side
                \draw[->] (lm2) -- (cmj);
                \draw[->] (cmjp) -- (lm1);
                \draw[->] (cmjp) -- (lm2);
                \draw[->] (cmjp) -- (lm3);
                \draw[->] (ln1) -- (cmj);
                \draw[->] (ln2) -- (cmj);
                \draw[->] (ln3) -- (cmj);

                \draw[decorate,decoration={brace,mirror,amplitude=4pt}] (-0.25,0.45) -- (3.25,0.45);
                \node[] () at (1.5,-0.05) {$L_j^-$};
                \node[] () at (1.5,-0.55) {\scriptsize each demand $j$};

                \draw[decorate,decoration={brace,mirror,amplitude=4pt}] (4.15,0.45) -- (7.65,0.45);
                \node[] () at (5.9,-0.05) {$L_{j+1}^-$};
                \node[] () at (5.9,-0.55) {\scriptsize each demand $j+1$};

                \node[] () at (3.7,4.55) {\small negative side};
            \end{scope}
        \end{tikzpicture}
        \caption{The positive and negative sides on two consecutive clauses.
        On each side, every clause vertex is adjacent to every literal vertex.
        The arrows indicate the intended orientation pattern used in the proof.
        On clause $j$, the figure shows one representative edge directed out of $c_j^+$ and one representative edge directed into $c_j^-$.
        Between consecutive clauses, the edges from $c_{j+1}^{\pm}$ to $L_j^{\pm}$ are oriented away from $c_{j+1}^{\pm}$,
        while the edges from $L_{j+1}^{\pm}$ to $c_j^{\pm}$ are oriented towards $c_j^{\pm}$.
        All depicted vertices are marked, and pendant leaves are omitted.}
        \label{fig:cw_clause_gadget}
    \end{figure}

    \proofsubparagraph*{Variable gadgets.}
    For each variable $x_i$, where $i \in [n]$, we introduce two vertices $v_i$ and $\overline{v}_i$,
    both of demand $0$, and connect them by an edge.
    We refer to $v_i$ and $\overline{v}_i$ as the \emph{selector} vertices of $x_i$.

    \proofsubparagraph*{Incidence edges.}
    The selector vertices are connected to literal vertices according to whether they make the corresponding literal true or false.
    More precisely, for each variable $x_i$, we add:
    \begin{itemize}
        \item edges from $v_i$ to every vertex $\ell^+_{j,i}$ such that $\ell_{j,i}=x_i$;
        \item edges from $v_i$ to every vertex $\ell^-_{j,i}$ such that $\ell_{j,i}=\neg x_i$;
        \item edges from $\overline{v}_i$ to every vertex $\ell^+_{j,i}$ such that $\ell_{j,i}=\neg x_i$;
        \item edges from $\overline{v}_i$ to every vertex $\ell^-_{j,i}$ such that $\ell_{j,i}=x_i$.
    \end{itemize}
    Thus, a selector vertex is adjacent on the positive side exactly to the literal occurrences it satisfies,
    and on the negative side exactly to the literal occurrences it falsifies.
    This incidence pattern is depicted in \cref{fig:cw_incidence}.

    \begin{figure}[ht]
        \centering
        \begin{tikzpicture}[scale=1, transform shape, decorate]
            \node[vertex,label=left:$v_i$] (v) at (1.4,2.35) {};
            \node[vertex,label=left:$\overline{v}_i$] (vb) at (1.4,0.95) {};
            \draw (v)--(vb);

            \node[black_vertex,label={[yshift=0.1cm]above:{\scriptsize $\ell^+_{j,i}$}}] (pp) at (5.2,2.9) {};
            \node[black_vertex,label={[yshift=0.1cm]above:{\scriptsize $\ell^+_{j',i}$}}] (pn) at (8.5,2.9) {};

            \node[black_vertex,label={[yshift=0.1cm]above:{\scriptsize $\ell^-_{j,i}$}}] (np) at (5.2,0.45) {};
            \node[black_vertex,label={[yshift=0.1cm]above:{\scriptsize $\ell^-_{j',i}$}}] (nn) at (8.5,0.45) {};

            \draw (v)--(pp);
            \draw (v)--(nn);
            \draw (vb)--(pn);
            \draw (vb)--(np);

            \node[draw=black!40, rounded corners, inner sep=5pt, fit=(pp)(pn)] {};
            \node[] () at (6.85,4.02) {\small positive copies};

            \node[draw=black!40, rounded corners, inner sep=5pt, fit=(np)(nn)] {};
            \node[] () at (6.85,1.58) {\small negative copies};

            \node[] () at (5.2,-0.65) {$\ell_{j,i}=x_i$};
            \node[] () at (8.5,-0.65) {$\ell_{j',i}=\neg x_i$};
        \end{tikzpicture}
        \caption{Incidence pattern for a variable $x_i$.
        The left column corresponds to an occurrence with $\ell_{j,i}=x_i$, and the right column to an occurrence with $\ell_{j',i}=\neg x_i$.
        Hence $v_i$ is adjacent to $\ell^+_{j,i}$ and $\ell^-_{j',i}$, while $\overline{v}_i$ is adjacent to $\ell^-_{j,i}$ and $\ell^+_{j',i}$.
        Pendant leaves of marked vertices are omitted.}
        \label{fig:cw_incidence}
    \end{figure}

    Finally, we set $k = 8m + n$.
    This completes the construction.

    \begin{claim}
        The graph $G$ has linear clique-width at most $6$.
    \end{claim}

    \begin{claimproof}
        We give a linear $6$-expression for $G$.
        Labels $1$ and $2$ are the final labels of vertices in $L^+$ and $L^-$, respectively;
        labels $3$ and $4$ are temporary labels for the selector vertices currently being processed;
        label $5$ is used for the pendant neighbors of the marked vertex currently under construction;
        and label $6$, denoted by $l_g$, is a garbage label for vertices whose neighborhoods are already complete.

        Whenever we say that we introduce a marked vertex with label $a$,
        we mean the following sequence of operations:
        create the marked vertex with label $a$,
        create its $k+1$ pendant neighbors with label $5$,
        join labels $a$ and $5$,
        and finally relabel $5$ to $l_g$.

        We first construct all selector vertices and all literal vertices.
        At the end of this phase, every vertex of $L^+$ has label $1$, every vertex of $L^-$ has label $2$,
        and every selector vertex has been relabeled to $l_g$.
        For each variable $x_i$, processed in increasing order of $i$, we perform the following steps:
        \begin{itemize}
            \item introduce $v_i$ with label $3$;
            \item for each neighbor of $v_i$ in $L^+$, introduce that literal vertex as a marked vertex with label $4$,
                join labels $3$ and $4$, and relabel $4$ to $1$;
            \item for each neighbor of $v_i$ in $L^-$, introduce that literal vertex as a marked vertex with label $4$,
                join labels $3$ and $4$, and relabel $4$ to $2$;
            \item introduce $\overline{v}_i$ with label $4$, join labels $3$ and $4$, and relabel $3$ to $l_g$;
            \item for each neighbor of $\overline{v}_i$ in $L^+$, introduce that literal vertex as a marked vertex with label $3$,
                join labels $4$ and $3$, and relabel $3$ to $1$;
            \item for each neighbor of $\overline{v}_i$ in $L^-$, introduce that literal vertex as a marked vertex with label $3$,
                join labels $4$ and $3$, and relabel $3$ to $2$;
            \item relabel $4$ to $l_g$.
        \end{itemize}
        This creates exactly the intended selector-literal edges: every literal vertex is first created with a temporary label, joined only to the selector label to which it should be adjacent, and only then relabeled to its final class $1$ or $2$. Once a selector vertex is moved to $l_g$, it never receives new edges.

        Finally, we create the clause vertices.
        For each vertex of $C^+$, introduce it as a marked vertex with label $3$,
        join labels $3$ and $1$, and relabel $3$ to $l_g$.
        Hence each vertex of $C^+$ becomes adjacent to all vertices of $L^+$ and to no other non-pendant vertices.
        Similarly, for each vertex of $C^-$, introduce it as a marked vertex with label $3$,
        join labels $3$ and $2$, and relabel $3$ to $l_g$.
        Hence each vertex of $C^-$ becomes adjacent to all vertices of $L^-$ and to no other non-pendant vertices.

        We have therefore constructed $G$ with at most six labels, and the claim follows.
    \end{claimproof}

    \begin{lemma}\label{lem:cw:SAT->CVC}
        If $\psi$ is a yes-instance of \OneInThreeSAT, then $\mathcal{I}$ is a yes-instance of \CVC.
    \end{lemma}

    \begin{nestedproof}
        Let $\alpha \colon X \to \{\true, \false\}$ be a truth assignment such that every clause of $\psi$ has exactly one true literal.
        Let
        \[
            S = \setdef{v_i}{i \in [n], \, \alpha(x_i) = \true} \cup
                \setdef{\overline{v}_i}{i \in [n], \, \alpha(x_i) = \false} \cup
                L^+ \cup L^- \cup C^+ \cup C^-,
        \]
        where $|S| = n + 8m = k$.

        We show that $S$ is a capacitated vertex cover of $G$ by defining an orientation in which
        (i) every vertex outside $S$ has in-degree $0$,
        and (ii) every vertex in $S$ has out-degree at least its demand.

        First orient every leaf-edge towards its marked endpoint.
        Next, for each variable $x_i$, orient every edge incident with the selector in $S$ towards that selector and every edge incident with the other selector away from it.
        Thus the unique selector vertex of $x_i$ contained in $S$ has in-degree equal to its degree,
        whereas the other selector has in-degree $0$.

        It remains to orient the edges between clause vertices and literal vertices.
        Fix a clause $c_j = (\ell_{j,i_1} \lor \ell_{j,i_2} \lor \ell_{j,i_3})$, and let $\ell_{j,t}$ be its unique true literal under $\alpha$.

        On the positive side, every literal vertex $\ell^+_{j,r}$ is treated identically with respect to clause vertices other than $c_j^+$:
        \[
            \ell^+_{j,r} \to c^+_1, \ldots, c^+_{j-1}
            \qquad\text{and}\qquad
            c^+_{j+1}, \ldots, c^+_m \to \ell^+_{j,r}.
        \]
        The only edge whose orientation depends on truth is the one incident with $c_j^+$:
        we orient $c_j^+ \to \ell^+_{j,t}$, and orient $\ell^+_{j,r} \to c_j^+$ for the two false literals.

        The negative side is symmetric. For every literal vertex $\ell^-_{j,r}$ we orient
        \[
            \ell^-_{j,r} \to c^-_1, \ldots, c^-_{j-1}
            \qquad\text{and}\qquad
            c^-_{j+1}, \ldots, c^-_m \to \ell^-_{j,r},
        \]
        and only the edge to $c_j^-$ depends on truth:
        we orient $\ell^-_{j,t} \to c_j^-$, and orient $c_j^- \to \ell^-_{j,r}$ for the two false literals.

        We verify feasibility by checking that every vertex of $S$ has out-degree at least its demand.
        Every selector in $S$ has demand $0$, so there is nothing to show.
        Consider a positive literal vertex of clause $c_j$.
        If it corresponds to the true literal, then it has one outgoing selector edge and $j-1$ outgoing clause-edges, for a total out-degree of $j$.
        If it corresponds to a false literal, then it has no outgoing selector edge and exactly $j$ outgoing clause-edges.
        Thus every vertex of $L_j^+$ has out-degree exactly $j$.
        The same argument on the negative side shows that every vertex of $L_j^-$ also has out-degree exactly $j$.

        Finally, consider the clause vertices.
        For each $j' < j$, the three vertices of $L_{j'}^+$ receive the edge from $c_j^+$, and inside clause $c_j$ only the edge to the true literal is oriented out of $c_j^+$.
        Hence
        \[
            \outdeg(c_j^+) = 3(j-1)+1.
        \]
        Symmetrically,
        \[
            \outdeg(c_j^-) = 3(j-1)+2.
        \]
        Therefore every clause vertex also has out-degree equal to its demand.

        We have thus defined an orientation in which every vertex outside $S$ has in-degree $0$ and every vertex in $S$ has out-degree at least its demand.
        Equivalently, every vertex in $S$ has in-degree at most its capacity, so $S$ is a capacitated vertex cover of $G$.
    \end{nestedproof}

    \begin{lemma}\label{lem:cw:CVC->SAT}
        If $\mathcal{I}$ is a yes-instance of \CVC, then $\psi$ is a yes-instance of \OneInThreeSAT.
    \end{lemma}

    \begin{nestedproof}
        Let $O$ be a feasible orientation of $G$ of size at most $k$,
        and let $S = \setdef{v \in V(G)}{\indeg_O(v) > 0}$ be the corresponding capacitated vertex cover,
        where $|S| \leq k$.

        We first identify the structure of $S$.
        Every marked vertex belongs to $S$.
        Indeed, if a marked vertex did not belong to $S$, then all its incident edges would be oriented away from it,
        so all its $k+1$ pendant neighbors would belong to $S$, contradicting $|S| \leq k$.
        Therefore $L^+ \cup L^- \cup C^+ \cup C^- \subseteq S$.
        Consider now a variable $x_i$.
        Since $v_i$ and $\overline{v}_i$ are adjacent, at least one of them must belong to $S$.
        Because there are $n$ variables and $|S| \leq 8m+n = k$, while the $8m$ marked vertices are already in $S$,
        we conclude that for every $i \in [n]$ the set $S$ contains exactly one of $v_i$ and $\overline{v}_i$.
        Hence $S$ consists precisely of the $8m$ marked vertices and one selector vertex per variable.
        In particular, no pendant vertex belongs to $S$, thus every leaf-edge is oriented towards its marked endpoint.

        Since the selector vertices have capacity equal to their degree, every selector in $S$ may cover all of its incident edges.
        Therefore we assume, without loss of generality, that every edge incident with a selector in $S$ is oriented towards that selector.

        \begin{claim}\label{claim:cw:orientation-structure}
            For every $j \in [m]$,
            exactly one vertex of $L_j^+$ has an outgoing edge
            to its neighboring selector vertex.
        \end{claim}

        \begin{claimproof}
            We will prove the following stronger statement, and the claim follows immediately from it.
            For every $j \in [m]$:
            \begin{itemize}
                \item There exists a single vertex in $L^+_j$ that has an outgoing edge towards its neighboring selector vertex,
                and an incoming edge from $c^+_j$.
                The remaining two vertices in $L^+_j$ have an incoming edge from their
                neighboring selector vertex, and an outgoing edge to $c^+_j$.

                \item Every edge between $L^+_j$ and $\{c^+_{j+1},\ldots,c^+_m\}$ is oriented towards $L^+_j$, and every edge between $c^+_j$ and $L^+_{j+1} \cup \cdots \cup L^+_m$ is oriented towards $c^+_j$.

                \item There exist two vertices in $L^-_j$ that have an outgoing edge towards their neighboring selector vertex,
                and an incoming edge from $c^-_j$.
                The remaining vertex in $L^-_j$ has an incoming edge from its
                neighboring selector vertex, and an outgoing edge to $c^-_j$.

                \item Every edge between $L^-_j$ and $\{c^-_{j+1},\ldots,c^-_m\}$ is oriented towards $L^-_j$, and every edge between $c^-_j$ and $L^-_{j+1} \cup \cdots \cup L^-_m$ is oriented towards $c^-_j$.
            \end{itemize}

            We prove the statement by induction on $j$.
            Let $j \in [m]$ and suppose that the statement holds for all $z \in [j+1,m]$
            (this is vacuous when $j=m$).
            We will show that the statement also holds for $j$.

            Let $c_j = (\ell_{j,i_1} \lor \ell_{j,i_2} \lor \ell_{j,i_3})$.
            For each variable index $i \in \{i_1,i_2,i_3\}$,
            exactly one of the two vertices $\ell^+_{j,i}$ and $\ell^-_{j,i}$ is adjacent to the unique selector in $S \cap \{v_i,\overline{v}_i\}$,
            because one of these two copies is adjacent to $v_i$ and the other to $\overline{v}_i$.
            Since every edge incident with a selector in $S$ is oriented towards that selector,
            it follows that exactly one of $\ell^+_{j,i}$ and $\ell^-_{j,i}$ has an outgoing edge to a selector vertex in $S$.
            Consequently, exactly three vertices of $L_j^+ \cup L_j^-$ obtain one unit of out-degree from a selector edge.

            Suppose first that no vertex of $L_j^+$ has an outgoing selector edge.
            By assumption, every edge between $L_j^+$ and $\{c_{j+1}^+,\ldots,c_m^+\}$ is oriented towards $L_j^+$.
            Consequently, every vertex of $L_j^+$ must use all of its remaining clause-edges to $\{c_1^+,\ldots,c_j^+\}$ as outgoing edges,
            so all three edges between $c_j^+$ and $L_j^+$ are oriented towards $c_j^+$.
            However,
            \[
                \capacity(c_j^+) = \deg(c_j^+) - (3(j-1)+1) = (k+1) + 3(m-j) + 2,
            \]
            and $c_j^+$ must cover its $k+1$ leaf-edges as well as, by assumption,
            $3(m-j)$ additional edges with $L^+_{j+1} \cup \cdots \cup L^+_m$,
            meaning it can only cover at most two additional edges from $L_j^+$, a contradiction.
            Therefore at least one vertex of $L_j^+$ has an outgoing selector edge.

            Suppose next that at most one vertex of $L_j^-$ has an outgoing selector edge.
            By assumption, every edge between $L_j^-$ and $\{c_{j+1}^-,\ldots,c_m^-\}$ is oriented towards $L_j^-$.
            Consequently, at least two vertices of $L_j^-$ must use all of their remaining clause-edges to $\{c_1^-,\ldots,c_j^-\}$ as outgoing edges,
            so at least two edges between $c_j^-$ and $L_j^-$ are oriented towards $c_j^-$.
            However,
            \[
                \capacity(c_j^-) = \deg(c_j^-) - (3(j-1)+2) = (k+1) + 3(m-j) + 1,
            \]
            and $c_j^-$ must cover its $k+1$ leaf-edges as well as, by assumption,
            $3(m-j)$ additional edges with $L^-_{j+1} \cup \cdots \cup L^-_m$,
            meaning it can only cover at most one additional edge from $L_j^-$, a contradiction.
            Therefore at least two vertices of $L_j^-$ have an outgoing selector edge.

            Since exactly three vertices of $L_j^+ \cup L_j^-$ have outgoing selector edges,
            it follows that exactly one vertex of $L_j^+$ and exactly two vertices of $L_j^-$ do.
            In that case, for the vertices of $L_j^+$ and $L_j^-$ that do not have outgoing selector edges,
            all their clause-edges to $\{c_1^+,\ldots,c_j^+\}$ and $\{c_1^-,\ldots,c_j^-\}$ respectively must be outgoing so that the orientation respects their demand.
            Therefore, $c_j^+$ has at least two incoming edges from $L_j^+$ and $c_j^-$ has at least one incoming edge from $L_j^-$.
            However, $c_j^+$ and $c_j^-$ have capacity to cover at most two and one additional edges apart from those of $L_{j+1}^+ \cup \cdots \cup L_m^+$ and $L_{j+1}^- \cup \cdots \cup L_m^-$, respectively,
            thus it follows that all edges between $c_j^+$ and $L_1^+ \cup \cdots \cup L_{j-1}^+$ and all edges between $c_j^-$ and $L_1^- \cup \cdots \cup L_{j-1}^-$ are oriented away from $c_j^+$ and $c_j^-$, respectively.

            From the previous paragraph, it follows that for the vertices of $L_j^+$ and $L_j^-$ that do have outgoing selector edges,
            they have an incoming edge from $c_j^+$ and $c_j^-$, respectively.
            Consequently, to respect their demand, all their remaining clause-edges to $\{c_1^+,\ldots,c_{j-1}^+\}$ and $\{c_1^-,\ldots,c_{j-1}^-\}$ respectively must be outgoing.
            This completes the proof.
        \end{claimproof}

        We define an assignment $\alpha \colon X \to \{\true,\false\}$ by setting
        $\alpha(x_i)=\true$ if $v_i \in S$, and $\alpha(x_i)=\false$ if $\overline{v}_i \in S$.
        We argue that $\alpha$ satisfies every clause of $\psi$ with exactly one true literal.
        Indeed, consider a clause $c_j$.
        By \Cref{claim:cw:orientation-structure}, exactly one vertex of $L_j^+$ has an outgoing edge to its neighboring selector vertex.
        By construction, a vertex $\ell^+_{j,i} \in L_j^+$ has an outgoing edge to the selected selector if and only if the literal $\ell_{j,i}$ is true under $\alpha$.
        Therefore exactly one literal of $c_j$ is true under $\alpha$.
    \end{nestedproof}
    Correctness now follows from \cref{lem:cw:SAT->CVC,lem:cw:CVC->SAT}.
    This completes the proof.
\end{proof}

\section{Vertex Cover Number and Vertex Integrity}\label{sec:vertex_cover_integrity}

In this section we consider the parameterizations by vertex cover number and vertex integrity.
We begin by establishing a barrier on the parameter dependence in the exponent for the parameterization by vertex cover number.
The recent work of Rohwedder and W\k{e}grzycki~\cite{innovations/RohwedderW25}
establishes a fine-grained equivalence between several problems and \textsc{Integer Linear Programming} (ILP) with bounded entries parameterized by the number of constraints,
including {\SMC} parameterized by the size of the universe.
In particular, by~\cite[Theorem~2]{innovations/RohwedderW25},
there exists a $2^{\bO(m^{2-\varepsilon})} n^{\bO(1)}$ algorithm for {\SMC} with $m$ elements in the universe
if and only if ILP with bounded values admits an algorithm of the same running time, where $m$ denotes the number of constraints, when each variable $x_i$ is additionally subject to bounds $\ell_i \le x_i \le r_i$ that do not count towards $m$.
Since in the reduction below the vertex cover number of the constructed graph is at most the universe size, any significantly faster algorithm for {\CVC} parameterized by $\vc$ would therefore transfer to this ILP variant.

\begin{theoremrep}[\appsymb]\label{thm:vc:lb}
    Let $\varepsilon > 0$ be any constant.
    {\CVC} cannot be solved in time $2^{\bO(\vc^{2 - \varepsilon})} n^{\bO(1)}$ where $\vc$ denotes the vertex cover number of the input graph,
    unless there exists a $2^{\bO(m^{2 - \varepsilon})} n^{\bO(1)}$ algorithm for {\SMC} where $m$ denotes the size of the universe.
\end{theoremrep}

\begin{proof}
    Let $\mathcal{I}$ be an instance of {\SMC} with universe $U = [m]$,
    sets $S_1,\ldots,S_n \subseteq U$,
    and integers $b,k \in \N$.
    We ask whether there exists a subfamily of at most $k$ sets such that
    every element of $U$ belongs to at least $b$ chosen sets.

    We construct an instance $\mathcal{J} = (G,\capacity,k')$ of {\CVC} as follows.
    For each element $i \in [m]$ we introduce an \emph{element} vertex $v_i$ and $k'+1$ unnamed leaves adjacent only to $v_i$.
    For each set $S_j$, $j \in [n]$, we introduce a \emph{set} vertex $u_j$.
    We add the edge $\{v_i,u_j\}$ whenever $i \in S_j$.
    For the capacities, we set $\capacity(v_i) = \deg_G(v_i) - b$ for all $i \in [m]$,
    and $\capacity(u_j) = \deg_G(u_j)$ for all $j \in [n]$.
    Finally, we let $k' = m + k$.

    \begin{claim}\label{claim:SMC->CVC}
        If $\mathcal{I}$ is a Yes instance of \SMC, then $\mathcal{J}$ is a Yes instance of \CVC.
    \end{claim}

    \begin{claimproof}
        Let $I \subseteq [n]$ be a family of at most $k$ sets such that every element $i \in [m]$
        belongs to at least $b$ sets $S_j$ with $j \in I$.
        Define $X = \setdef{v_i}{i \in [m]} \cup \setdef{u_j}{j \in I}$.
        Clearly, $|X| \leq m+k = k'$.

        We show that $X$ is a capacitated vertex cover.
        First, for every $i \in [m]$, $v_i$ covers the edges to its leaves.
        For each chosen set vertex $u_j \in X$, since $\capacity(u_j)=\deg_G(u_j)$,
        it may cover all edges incident with it.
        Consider now an element vertex $v_i$.
        By assumption, at least $b$ neighbors of $v_i$ among the set vertices belong to $X$.
        Hence at most $\deg_G(v_i)-b=\capacity(v_i)$ edges incident with $v_i$
        need to be covered by $v_i$ itself.
        Therefore the capacity bound of $v_i$ is respected, and all edges are covered.
    \end{claimproof}

    \begin{claim}\label{claim:CVC->SMC}
        If $\mathcal{J}$ is a Yes instance of \CVC, then $\mathcal{I}$ is a Yes instance of \SMC.
    \end{claim}

    \begin{claimproof}
        Let $X$ be a capacitated vertex cover of $G$ with $|X| \leq k'$.
        First, every element vertex $v_i$ must belong to $X$.
        Indeed, if $v_i \notin X$, then all of its $k'+1$ incident leaves must do so,
        contradicting the fact that $|X| \leq k'$.
        Furthermore, we may assume that no leaf belongs to $X$; if that is the case,
        exchange the leaf for a set vertex incident to the same element vertex.
        Let $I = \setdef{j \in [n]}{u_j \in X}$, where $|I| \leq k$.

        Fix any element $i \in [m]$.
        Because $v_i \in X$ and $\capacity(v_i)=\deg_G(v_i)-b$,
        at least $b$ edges incident with $v_i$ must be covered by its neighbors.
        Since its leaves do not belong to $X$, these $b$ neighboring vertices must all be set vertices.
        Therefore at least $b$ vertices $u_j \in X$ satisfy $i \in S_j$.
        Equivalently, $i$ belongs to at least $b$ sets $S_j$ with $j \in I$.
        Since $i \in [m]$ was arbitrary, the family $\setdef{S_j}{j \in I}$ is a feasible solution of size at most $k$
        for the {\SMC} instance.
    \end{claimproof}
    Correctness follows from \cref{claim:SMC->CVC,claim:CVC->SMC}.
    Observe that the set $\setdef{v_i}{i \in [m]}$ is a vertex cover of $G$,
    hence $\vc(G) \leq m$.
    Consequently, a $2^{\bO(\vc^{2-\varepsilon})} n^{\bO(1)}$ algorithm for {\CVC} would imply a
    $2^{\bO(m^{2-\varepsilon})} n^{\bO(1)}$ algorithm for \SMC.
    This completes the proof.
\end{proof}

We now turn to vertex integrity.
{\CVC} is already known to be FPT in that case~\cite{tcs/GimaHKKO22},
however the algorithm of Gima et al.~has a running time with double-exponential parametric dependence:
they formulate the problem as an ILP with an FPT number of variables ($2^{\bO(\vi^2)}$),
and then use standard tools such as Lenstra's algorithm~\cite{mor/Lenstra83} to solve it.
Our main algorithmic result in this section is a significant improvement over this bound,
obtained by expressing the problem as an $N$-fold Integer Program instead.
In particular, the running time of our algorithm matches the barrier of \cref{thm:vc:lb}.

\begin{theoremrep}[\appsymb]\label{thm:vi:algo}
    There is an algorithm that solves {\CVC} in time
    $\vi^{\bO(\vi^2)} n^{\bO(1)}$,
    where $\vi$ is the vertex integrity of the input graph.
\end{theoremrep}

\begin{proof}
    If $k \le \vi$, then we may simply use one of the known algorithms for the natural parameterization~\cite{iwpec/DomLSV08,sofsem/RooijR19},
    which runs in time $k^{\bO(k)} n^{\bO(1)} \le \vi^{\bO(\vi)} n^{\bO(1)}$.
    Hence, in the remainder, we may assume that $\vi < k$.

    Let $(G,\capacity,k)$ be an instance of {\CVC}, where $G=(V,E)$ and $\capacity \colon V \to \mathbb{N}$.
    Let $U \subseteq V$ be a set witnessing the vertex integrity of $G$, that is,
    for every connected component $C$ of $G-U$ we have $|U| + |V(C)| \le \vi(G)$.
    Such a set can be computed in time ${\vi}^{\bO(\vi)} n^{\bO(1)}$~\cite{algorithmica/DrangeDH16}.
    Write $U=\{u_1,\ldots,u_{|U|}\}$ and let $\cc(G-U)=\{C^1,\ldots,C^m\}$.
    By definition, each component $C^j$ has size at most $\vi-|U| \le \vi$.
    We guess the behavior of a solution on $U$, and then solve the remaining task on the components of $G-U$ by an $N$-fold IP.

    We now try all possible behaviors of a solution on the modulator $U$.
    More precisely, we guess a set $S \subseteq U$, intended to be the set of vertices of $U$ with positive in-degree,
    and an orientation $O_U$ of the edges of $G[U]$.
    We call such a guess \emph{valid} if $\indeg_{O_U}(u_i) \le \capacity(u_i)$ for every $u_i \in U$,
    and $\indeg_{O_U}(u_i) > 0$ implies $u_i \in S$.
    There are at most $2^{|U|}$ choices for $S$ and at most $2^{|E(G[U])|}$ orientations of $G[U]$,
    and since $|U| \le \vi$ and $|E(G[U])| \le \vi^2$ it follows that
    the number of valid guesses is at most $2^{\bO(\vi^2)}$.

    Fix one valid guess $(S,O_U)$.
    For every $i \in [|U|]$, define
    \[
        b_i =
        \begin{cases}
            \capacity(u_i) - \indeg_{O_U}(u_i) & \text{if } u_i \in S, \\
            0 & \text{if } u_i \notin S.
        \end{cases}
    \]
    Thus $b_i$ is the remaining capacity of $u_i$ available for edges coming from the components of $G-U$.

    \proofsubparagraph{Local constraints.}
    For each component $C^j$, let $F^j$ be the set of all edges with at least one endpoint in $V(C^j)$.
    A \emph{partial orientation} of $C^j$ is an orientation of all edges in $F^j$.
    Such a partial orientation is \emph{valid} if (i) it never orients an edge towards a vertex of $U \setminus S$,
    and (ii) it respects the capacity of every vertex of $C^j$.

    Once $(S,O_U)$ has been fixed, each component interacts with the rest of the graph only through the vertices of $U$.
    Therefore, for every valid partial orientation, it suffices to record two pieces of information:
    its contribution to the load of each vertex of $U$, and its contribution to the number of vertices of positive in-degree.

    For each $j \in [m]$, enumerate all valid partial orientations of $F^j$ and denote the resulting family by
    \[
        \mathcal{O}^j = \{O^j_1,\ldots,O^j_{\mu_j}\}.
    \]
    Since $|U|+|V(C^j)| \le \vi$ and all edges of $F^j$ lie in the graph induced by $U \cup V(C^j)$,
    we have $\mu_j \le 2^{\bO(\vi^2)}$.
    If $\mu_j=0$, then the current guess can be discarded immediately.
    For every $q \in [\mu_j]$ and every $i \in [|U|]$, let
    \[
        a^{j,q}_i
        =
        |\setdef{u_i v \in E(G)}{v \in V(C^j)\text{ and }O^j_q(u_i v)=v \to u_i}|,
    \]
    that is, $a^{j,q}_i$ is the number of additional edges covered by $u_i$ when the component $C^j$ uses the partial orientation $O^j_q$.
    We also define
    \[
        d^{j,q}
        =
        |\setdef{v \in V(C^j)}{\indeg_{O^j_q}(v) > 0}|,
    \]
    which is the number of vertices of $C^j$ that contribute to the solution size under $O^j_q$.

    We introduce, for every component $C^j$ and every $q \in [\mu_j]$, a binary variable $I^j_q \in \{0,1\}$.
    The intended meaning is that $I^j_q=1$ if and only if we choose the partial orientation $O^j_q$ for the component $C^j$.
    The local constraint of brick $j$ is simply
    \[
        \sum_{q \in [\mu_j]} I^j_q = 1,
    \]
    enforcing that exactly one valid partial orientation is selected for $C^j$.

    \proofsubparagraph{Global constraints.}
    The bricks are coupled through the vertices of the modulator and through the solution-size budget.
    For every $i \in [|U|]$, we require
    \[
        \sum_{j \in [m]} \sum_{q \in [\mu_j]} a^{j,q}_i I^j_q \le b_i,
    \]
    ensuring that the total additional load on $u_i$ does not exceed its remaining capacity.
    We also require
    \[
        \sum_{j \in [m]} \sum_{q \in [\mu_j]} d^{j,q} I^j_q \le k-|S|,
    \]
    ensuring that outside the modulator at most $k-|S|$ vertices obtain positive in-degree.

    The resulting formulation has one brick per connected component of $G-U$, one local equality per brick,
    at most $|U|+1 \le \vi+1$ global constraints, coefficients bounded by $\vi$,
    and at most $2^{\bO(\vi^2)}$ variables per brick.
    Hence it is an $N$-fold IP whose relevant parameters are functions of~$\vi$.

    \begin{claim}
        The current guess $(S,O_U)$ can be extended to a feasible orientation of size at most $k$
        if and only if the above integer program is feasible.
    \end{claim}

    \begin{claimproof}
        For the forward direction, suppose there exists a feasible orientation $O$ of size at most $k$
        whose restriction to $G[U]$ is $O_U$ and such that $S$ is the set of vertices of $U$ with positive in-degree.
        For each component $C^j$, let $q_j$ be such that the restriction of $O$ to $F^j$ is the partial orientation $O^j_{q_j}$.
        Since $O$ is feasible and no vertex of $U \setminus S$ has positive in-degree, every $O^j_{q_j}$ is valid.
        Set $I^{j}_{q_j}=1$ and all other variables of brick $j$ to $0$.
        Then every local constraint is satisfied.

        For every $i \in [|U|]$, the left-hand side of the corresponding global constraint is exactly the number of edges entering $u_i$
        from the components of $G-U$, which is at most $b_i$ because $O$ is feasible.
        Moreover,
        \[
            \sum_{j \in [m]} \sum_{q \in [\mu_j]} d^{j,q} I^j_q
        \]
        is precisely the number of vertices of $V \setminus U$ with positive in-degree in $O$,
        and this is at most $k-|S|$.
        Hence the integer program is feasible.

        For the reverse direction, suppose the integer program is feasible.
        Since every variable is binary and each brick satisfies $\sum_{q \in [\mu_j]} I^j_q=1$,
        each component $C^j$ selects exactly one valid partial orientation, say $O^j_{q_j}$.
        We define an orientation $O$ of $G$ by taking $O_U$ on $G[U]$ and, for every component $C^j$, taking $O^j_{q_j}$ on $F^j$.
        This is well defined because the sets $F^1,\ldots,F^m$ are pairwise edge-disjoint and disjoint from $E(G[U])$.

        By validity of each selected partial orientation, every vertex of $V \setminus U$ respects its capacity bound in $O$,
        and no edge is oriented towards a vertex of $U \setminus S$.
        For every $u_i \in S$, the number of additional incoming edges coming from the components of $G-U$ is
        \[
            \sum_{j \in [m]} \sum_{q \in [\mu_j]} a^{j,q}_i I^j_q,
        \]
        which is at most $b_i$ by the corresponding global constraint.
        Since $\indeg_{O_U}(u_i) + b_i = \capacity(u_i)$, the capacity of $u_i$ is also respected.
        Therefore $O$ is a feasible orientation.

        Finally, the number of vertices of $U$ with positive in-degree in $O$ is at most $|S|$,
        while the last global constraint guarantees that at most $k-|S|$ vertices of $V \setminus U$ have positive in-degree.
        Hence $\size(O) \le |S| + (k-|S|) = k$, so the current guess indeed yields a solution.
        This completes the proof.
    \end{claimproof}

    \proofsubparagraph{Running time.}
    There are at most $2^{\bO(\vi^2)}$ valid guesses $(S,O_U)$.
    For a fixed guess and a fixed component $C^j$, the set $F^j$ contains at most $\vi^2$ edges,
    hence all valid partial orientations of $F^j$ can be enumerated in time $2^{\bO(\vi^2)}$.
    Since there are at most $n$ connected components, the total time needed to construct all bricks for one fixed guess is
    $2^{\bO(\vi^2)} n^{\bO(1)}$.

    We now consider the optimization version of the above program:
    we drop the final inequality
    \[
        \sum_{j \in [m]} \sum_{q \in [\mu_j]} d^{j,q} I^j_q \le k-|S|
    \]
    and instead minimize its left-hand side.
    By the claim above, the current guess $(S,O_U)$ succeeds if and only if the optimum value of this objective is at most $k-|S|$.

    Let us now bound the parameters of this $N$-fold IP in the notation of \cref{prop:nfold}.
    In the optimization version, the number of global constraints is
    $r = |U| \le \vi$.
    Each brick has exactly one local equality, so
    $s = 1$.
    The number of variables in brick $j$ is
    $t_j = \mu_j \le 2^{\bO(\vi^2)}$,
    and thus
    $t = \max_j t_j \le 2^{\bO(\vi^2)}$.
    Moreover, every coefficient appearing in the constraint matrices is at most $\vi$:
    the local matrix contains only coefficients in $\{0,1\}$,
    while each value $a^{j,q}_i$ or $d^{j,q}$ is at most $|V(C^j)| \le \vi$.
    Hence
    $a = r \cdot s \cdot \max_{i \in [N]} \big\{ \max\{\|D_i\|_\infty,\|A_i\|_\infty\} \big\}
        = \bO(\vi^2)$.
    The exponent in \cref{prop:nfold} is
    $r^2 s + rs^2 = \bO(\vi^2)$,
    and therefore
    $a^{r^2 s + rs^2} = \vi^{\bO(\vi^2)}$.

    The dimension of the instance is
    $d = \sum_{j \in [m]} t_j \le n \cdot 2^{\bO(\vi^2)}$,
    so $d \log d = 2^{\bO(\vi^2)} n^{\bO(1)}$.
    Finally, the optimum value of the objective is at most $n$.
    Hence we may take $L \le n$.
    Plugging these bounds into \cref{prop:nfold}, one fixed guess can be processed in time
    \[
        \vi^{\bO(\vi^2)} \cdot 2^{\bO(\vi^2)} n^{\bO(1)}
        =
        \vi^{\bO(\vi^2)} n^{\bO(1)}.
    \]
    Multiplying by the $2^{\bO(\vi^2)}$ possible valid guesses preserves the same asymptotic bound.
\end{proof}

\section{Planar Approximation}\label{sec:planar}

In this section we present an EPTAS for {\CVC} on planar graphs with arbitrary
capacities.  Building on Baker's layering technique~\cite{jacm/Baker94},
we extend Becker's EPTAS for bounded capacities~\cite{waoa/Becker17} which runs
in time $C^{\bO(C/\varepsilon)}n^{\bO(1)}$, where $C$ is the maximum capacity.

Baker's technique uses \emph{BFS layers}, which group vertices by their
distance from a chosen root.  Periodic cuts divide the layers into
disjoint blocks, each inducing a \emph{slab}.  The offset of the cuts is
the \emph{shift}.  A \emph{patch} is the subgraph induced by a few layers
on either side of a cut.
Becker solves the slabs and patches exactly, combines their solutions,
and repairs capacity violations by augmenting paths, at a cost of at
most one additional vertex per unit of excess load, which is at most $C$.

We use a similar approach, but we avoid solving the patches exactly and instead
use a fixed $2$-approximate reference solution to orient crossing edges.
The excess load can then be charged to edges induced by reference vertices near the cuts.
Planarity and averaging over the shifts bound the loss by an
$\varepsilon$-fraction of the optimum with only $\bO(1/\varepsilon)$
layers per slab.  Solving the slabs exactly in $n^{\bO(\tw)}$ time
therefore gives a PTAS for arbitrary capacities, with running time
$n^{\bO(1/\varepsilon)}$. We note that our use of an initial constant-factor
approximate solution to control the cost of patching across layer boundaries
is similar to Demaine and Hajiaghayi's extension of Baker's technique for
\textsc{Connected Dominating Set}~\cite{soda/DemaineH05}.

In order to improve the running time to an EPTAS, we replace the exact slab solver
with the treewidth FPT-AS of Chu and Lin~\cite{isaac/ChuL23}, thus obtaining an
approximate solution while respecting all capacities in time $f(\tw,\varepsilon) n^{\bO(1)}$.
The rest of the analysis remains the same.  We believe this combination of parameterized approximation and
Baker's technique is of independent interest: for problems that are W[1]-hard parameterized by treewidth,
it may turn other PTASes into EPTASes.

We use the orientation formulation from \Cref{sec:preliminaries}.
\Cref{lem:planar-repair} formalizes Becker's augmenting-path repair
argument~\cite{waoa/Becker17}; we include a short proof for completeness.

\begin{lemmarep}[\appsymb]\label{lem:planar-repair}
    Given an orientation $O$ of a graph $G$ with capacities $\capacity$,
    one can in polynomial time either correctly report infeasibility or find a feasible
    orientation $O'$ with
    \[
        \size(O')\leq\size(O)+\sum_{v\in V(G)}
            \max\{0,\indeg_O(v)-\capacity(v)\}.
    \]
\end{lemmarep}

\begin{proof}
    A vertex is \emph{overloaded} if its indegree exceeds its capacity,
    and has \emph{spare capacity} if its indegree is below its capacity.
    While some vertex $v$ is overloaded, find a directed path
    from a vertex $u$ with spare capacity to $v$ in the current orientation.
    Reversing this path increases the indegree of $u$ by one, decreases
    that of $v$ by one, and leaves all other indegrees unchanged.  Thus
    the total excess decreases by one, and only $u$ can become newly
    selected.

    If the search fails, let $X$ be the set of vertices that can reach
    $v$, including $v$ itself.  No arc enters $X$, and every vertex of
    $X$ has indegree at least its capacity, with strict inequality at $v$.
    Hence, in the current orientation,
    \[
        |E(G[X])|=\sum_{x\in X}\indeg(x)
        >\sum_{x\in X}\capacity(x).
    \]
    This certifies infeasibility, since all edges of $G[X]$ must be
    covered by vertices of $X$.

    The number of reversals is at most the initial excess, which is at
    most $|E(G)|$.  Each reversal adds at most one selected vertex,
    proving the size bound and polynomial running time.
\end{proof}

\begin{theoremrep}[\appsymb]\label{thm:planar-eptas}
    For every $0<\varepsilon\leq1$, \CVC\ on planar graphs admits a
    $(1+\varepsilon)$-approximation algorithm running in time
    $(1/\varepsilon)^{\bO(1/\varepsilon)}n^{\bO(1)}$.
    The algorithm respects all capacities and reports infeasibility when
    no feasible solution exists.
\end{theoremrep}

\begin{proof}
    Let $G$ be the input graph, with $n=|V(G)|$ and capacities $\capacity$.
    Write $\OPT$ for its optimum and $\OPT_H$ for the optimum on a
    subgraph $H$ with the original capacities.

    We first prove the PTAS guarantee with exact slab solutions, then
    replace the slab solver to obtain the EPTAS.
    Test feasibility by applying \Cref{lem:planar-repair} to any orientation.
    Edgeless graphs are trivial, and disconnected graphs are handled
    componentwise, since the optimum is additive.  Hence assume that the
    instance is feasible and $G$ is connected with at least one edge.
    Fix an integer $k\geq4$, to be chosen later.
    Compute a $2$-approximate orientation $O_0$ in polynomial time~\cite{jcss/GandhiHKKS06,soda/Wong17},
    and let $R$ denote the vertices belonging to the corresponding capacitated vertex cover,
    that is, $R \coloneq\setdef{v\in V(G)}{\indeg_{O_0}(v)>0}$,
    where $|R|\leq2\OPT$.

    \subparagraph{Baker's technique.}
    Fix a BFS root, and let $L_0,L_1,\ldots,L_h$ be its layers, with
    out-of-range layers being empty.  Every edge joins vertices
    in the same or consecutive layers.  For each shift $r\in[0,k-1]$,
    let
    \[
        I_r\coloneq\setdef{t\in[1,h]}{t\equiv r\pmod{k}}.
    \]
    We cut between $L_{t-1}$ and $L_t$ for every $t\in I_r$.
    Each resulting block contains at most $k$ consecutive layers and
    induces a \emph{slab}; the first and last blocks may be shorter.
    Around a cut before $L_t$, the two adjacent layers form the boundary,
    and the four layers $L_{t-2},\ldots,L_{t+1}$ induce a \emph{patch}.
    Let $B_r$ be the union of the boundaries and $T_r$ the union of the
    patch vertex sets:
    \[
        B_r\coloneq\bigcup_{t\in I_r}(L_{t-1}\cup L_t),
        \qquad
        T_r\coloneq\bigcup_{t\in I_r}
            (L_{t-2}\cup L_{t-1}\cup L_t\cup L_{t+1}).
    \]
    Edges joining different slabs are called \emph{crossing edges}.
    They have both endpoints in $B_r$, and every neighbor
    of a vertex in $B_r$ belongs to $T_r$; see \Cref{fig:planar-layers}.
    The patches are used only to analyze the loss; we do not solve
    separate instances on them.

    \begin{figure}[htb]
        \centering
        \begin{subfigure}[t]{.48\textwidth}
            \centering
            \begin{tikzpicture}[
                x=.66cm, y=.68cm,
                every node/.style={font=\scriptsize},
                windowbar/.style={rounded corners=1.5pt, line width=.55pt},
            ]
                \foreach \x in {-.5,.5,...,7.5}
                    \draw[black!13] (\x,-.35)--(\x,2.35);
                \draw[windowbar, draw=RoyalBlue!80!black, fill=RoyalBlue!20]
                    (-.48,1.68) rectangle (3.46,2.32);
                \draw[windowbar, draw=RoyalBlue!80!black, fill=RoyalBlue!20]
                    (3.54,1.68) rectangle (7.48,2.32);
                \node at (1.5,2) {left slab};
                \node at (5.5,2) {right slab};
                \draw[windowbar, draw=BurntOrange!85!black, fill=BurntOrange!24]
                    (2.52,.68) rectangle (4.48,1.32);
                \node at (3.5,1) {$B_r$};
                \draw[windowbar, draw=black!65, fill=black!10]
                    (1.52,-.32) rectangle (5.48,.32);
                \node at (3.5,0) {$T_r$};
                \draw[densely dashed, black!65] (3.5,1.42)--(3.5,2.62);
                \node[above] at (3.5,2.62) {cut before $L_t$};
                \foreach \x/\r in {
                    0/{-4},1/{-3},2/{-2},3/{-1},4/0,5/1,6/2,7/3
                }
                    \node at (\x,-.67) {$\r$};
                \node at (3.5,-1.14) {offset $j$ in layer $L_{t+j}$};
            \end{tikzpicture}
            \caption{The boundary around one cut.}
        \end{subfigure}
        \hfill
        \begin{subfigure}[t]{.48\textwidth}
            \centering
            \begin{tikzpicture}[
                x=1cm, y=.68cm,
                every node/.style={font=\scriptsize},
                windowbar/.style={rounded corners=1.5pt, line width=.55pt},
            ]
                \path[use as bounding box] (-.2,-1.4) rectangle (5.6,3.1);
                \foreach \x/\j in {0/1,1.8/2,3.6/3} {
                    \draw[windowbar, draw=RoyalBlue!80!black, fill=RoyalBlue!20]
                        ({\x+.03},-.32) rectangle ({\x+1.77},.32);
                    \node at ({\x+.9},0) {slab $\j$};
                }
                \foreach \x in {1.8,3.6} {
                    \draw[windowbar, draw=BurntOrange!85!black, fill=BurntOrange!24]
                        ({\x-.2},.68) rectangle ({\x+.2},1.32);
                    \draw[windowbar, draw=black!65, fill=black!10]
                        ({\x-.4},1.68) rectangle ({\x+.4},2.32);
                    \draw[densely dashed, black!65] (\x,-.5)--(\x,2.58);
                }
                \node at (2.7,1) {$B_r$};
                \node at (2.7,2) {$T_r$};
                \node[above] at (1.8,2.62) {$t$};
                \node[above] at (3.6,2.62) {$t+k$};
                \draw[decorate, decoration={brace, mirror, amplitude=3pt}]
                    (1.83,-.55)--(3.57,-.55);
                \node at (2.7,-1.05) {$k$ layers};
            \end{tikzpicture}
            \caption{A shifted partition into disjoint slabs.}
        \end{subfigure}
        \caption{The slabs (blue) are solved independently.  Each cut has
        a two-layer boundary (orange) and a four-layer patch (gray).
        Their unions form $B_r$ and $T_r$, respectively.
        Panel~(a) shows only the portions adjacent to one cut; panel~(b)
        shows consecutive cuts schematically.}
        \label{fig:planar-layers}
    \end{figure}

    Each layer belongs to $T_r$ for at most four shifts $r \in [0,k-1]$.
    Choose $r \in [0,k-1]$ minimizing $|R\cap T_r|$.  Averaging gives
    \begin{equation}\label{eq:planar-shift}
        |R\cap T_r|\leq\frac{1}{k}\sum_{s=0}^{k-1}|R\cap T_s|
        \leq\frac{4|R|}{k}\leq\frac{8}{k}\OPT.
    \end{equation}
    Fix this shift and let $\mathcal{H}$ be its family of slabs.
    Each $H\in\mathcal{H}$ has a tree decomposition of width $\bO(k)$,
    constructible in polynomial
    time~\cite[Theorem~7.33 and Corollary~7.34]{books/CyganFKLMPPS15}.

    \subparagraph{The PTAS.}
    Solve each slab $H\in\mathcal{H}$ exactly in $n^{\bO(k)}$ time using
    the standard treewidth dynamic program~\cite[Section~3]{isaac/ChuL23},
    obtaining an orientation $O_H$ of size $\OPT_H$.
    Combine the slab orientations with the crossing edges oriented as in
    $O_0$, giving an orientation $\widehat O$.
    Repair $\widehat O$ using \Cref{lem:planar-repair} and return the
    feasible orientation $O$.  Repair succeeds because the input is
    feasible.

    \subparagraph{Approximation guarantee.}
    Restricting an optimum orientation to the vertex-disjoint slabs gives
    $\sum_{H\in\mathcal{H}}\OPT_H\leq\OPT$.  Orienting the crossing
    edges as in $O_0$ adds selected vertices only in $R\cap B_r$, so
    \begin{equation}\label{eq:planar-total-cost}
        \size(\widehat O)
        \leq\sum_{H\in\mathcal{H}}\size(O_H)+|R\cap B_r|
        \leq\OPT+|R\cap B_r|.
    \end{equation}

    Before repair, only vertices in $R\cap B_r$ can be overloaded.
    Since $O_0$ is feasible, the excess at each such vertex $v$ is at most
    the number of edges it does not cover in $O_0$, that is, $\outdeg_{O_0}(v)$.
    Every arc of $O_0$ with its tail in $R\cap B_r$ has its head in
    $R\cap T_r$: all heads belong to $R$, and all neighbors of $B_r$
    lie in $T_r$.
    These arcs are distinct edges of the planar graph
    $G[R\cap T_r]$.  Thus repair adds at most
    \[
        \sum_{v\in R\cap B_r}\outdeg_{O_0}(v)
        \leq |E(G[R\cap T_r])|
        \leq 3|R\cap T_r|
    \]
    selected vertices, where we used the fact that a planar graph with $n$ vertices has at most $3n-6$ edges~\cite[Corollary 4.2.10.]{books/Diestel25}.
    Combining this with \eqref{eq:planar-total-cost},
    $B_r\subseteq T_r$, and \eqref{eq:planar-shift}, we obtain
    \[
        \size(O)
        \leq\OPT+4|R\cap T_r|
        \leq\left(1+\frac{32}{k}\right)\OPT.
    \]
    Taking $k=\lceil32/\varepsilon\rceil$ gives a PTAS for arbitrary
    capacities with running time $n^{\bO(1/\varepsilon)}$.

    \subparagraph{From PTAS to EPTAS.}
    Take $\delta=\varepsilon/2$ and $k=\lceil64/\varepsilon\rceil$.
    We change only the slab solver.  Edgeless slabs are trivial;
    for any other slab $H$, write $n_H=|V(H)|$ and let $w\geq1$ be the
    width of its decomposition.
    Chu and Lin~\cite[Theorem~1 and Appendix~B]{isaac/ChuL23} find a feasible
    orientation with approximation ratio $1+\bO(1/(w^2\log n_H))$ in time
    $(w\log n_H)^{\bO(w)}n_H^{\bO(1)}=w^{\bO(w)}n_H^{\bO(1)}$.%
    \footnote{It is a standard argument that $(\log n_H)^{\bO(w)}\leq w^{\bO(w)}n_H$.}
    Thus, if $w^2\log n_H\geq a/\delta$ for a sufficiently large absolute
    constant $a$, we use their algorithm to obtain a $(1+\delta)$-approximation.
    Otherwise, solving the slab exactly takes
    \[
        n_H^{\bO(w)}\leq2^{\bO(1/(\delta w))}
        \leq2^{\bO(1/\delta)}.
    \]
    The slab solutions remain feasible, and their total size is at most
    $(1+\delta)\OPT$.  The repair bound is unchanged, so
    \[
        \size(O)\leq\left(1+\delta+\frac{32}{k}\right)\OPT
        \leq(1+\varepsilon)\OPT.
    \]
    Since $w=\bO(k)=\bO(1/\varepsilon)$ and
    $1/\delta=\bO(1/\varepsilon)$, solving all slabs takes
    $(1/\varepsilon)^{\bO(1/\varepsilon)}n^{\bO(1)}$ time.
    All other steps take $k^{\bO(1)}n^{\bO(1)}$ time, which is absorbed
    in this bound.
\end{proof}

\section{Conclusion}\label{sec:conclusion}

We revisited the parameterized complexity of {\CVC} and obtained a significantly sharper picture of its exact complexity across several standard parameterizations.
A few natural questions remain open.
First and foremost, our lower bound for the parameterization by vertex cover number is conditional on the fine-grained equivalences of Rohwedder and W\k{e}grzycki~\cite{innovations/RohwedderW25}.
Can one replace this barrier by an ETH-based lower bound ruling out, say, a $2^{o(\vc^2)} n^{\bO(1)}$ algorithm?
Alternatively, can one perhaps show \emph{equivalence} to the hypothesis of Rohwedder and W\k{e}grzycki?
Another interesting open question concerns the EPTAS on planar graphs, and whether we can obtain
one of running time $2^{\bO(1/\varepsilon)} n^{\bO(1)}$.

\newpage

\bibliography{refs}

@book{books/Diestel25,
  author    = {Reinhard Diestel},
  title     = {Graph Theory},
  publisher = {Springer},
  year      = {2025},
  series    = {Graduate Texts in Mathematics},
  volume    = {173},
  edition   = {6},
  doi       = {10.1007/978-3-662-70107-2},
  isbn      = {978-3-662-70106-5}
}

@article{tcs/JansenN19,
  author = {Bart M. P. Jansen and Jesper Nederlof},
  title = {Computing the chromatic number using graph decompositions via matrix rank},
  journal = {Theoretical Computer Science},
  volume = {795},
  pages = {520--539},
  year = {2019},
  doi = {10.1016/j.tcs.2019.08.006}
}

@book{books/CyganFKLMPPS15,
  author    = {Marek Cygan and
               Fedor V. Fomin and
               Lukasz Kowalik and
               Daniel Lokshtanov and
               D{\'{a}}niel Marx and
               Marcin Pilipczuk and
               Michal Pilipczuk and
               Saket Saurabh},
  title     = {Parameterized Algorithms},
  publisher = {Springer},
  year      = {2015},
  doi       = {10.1007/978-3-319-21275-3},
  isbn      = {978-3-319-21274-6}
}

@article{toct/LampisV24,
  author       = {Michael Lampis and
                  Manolis Vasilakis},
  title        = {Structural Parameterizations for Two Bounded Degree Problems Revisited},
  journal      = {{ACM} Trans. Comput. Theory},
  volume       = {16},
  number       = {3},
  pages        = {17:1--17:51},
  year         = {2024},
  doi          = {10.1145/3665156}
}

@inproceedings{iwpec/DomLSV08,
  author       = {Michael Dom and
                  Daniel Lokshtanov and
                  Saket Saurabh and
                  Yngve Villanger},
  title        = {Capacitated Domination and Covering: {A} Parameterized Perspective},
  booktitle    = {Parameterized and Exact Computation, Third International Workshop,
                  {IWPEC} 2008. Proceedings},
  series       = {Lecture Notes in Computer Science},
  volume       = {5018},
  pages        = {78--90},
  publisher    = {Springer},
  year         = {2008},
  doi          = {10.1007/978-3-540-79723-4_9}
}

@article{algorithmica/BodlaenderGJJL25,
  author       = {Hans L. Bodlaender and
                  Carla Groenland and
                  Hugo Jacob and
                  Lars Jaffke and
                  Paloma T. Lima},
  title        = {{XNLP}-Completeness for Parameterized Problems on Graphs with a Linear
                  Structure},
  journal      = {Algorithmica},
  volume       = {87},
  number       = {4},
  pages        = {465--506},
  year         = {2025},
  doi          = {10.1007/S00453-024-01274-9}
}

@article{tcs/GimaHKKO22,
  author       = {Tatsuya Gima and
                  Tesshu Hanaka and
                  Masashi Kiyomi and
                  Yasuaki Kobayashi and
                  Yota Otachi},
  title        = {Exploring the gap between treedepth and vertex cover through vertex
                  integrity},
  journal      = {Theor. Comput. Sci.},
  volume       = {918},
  pages        = {60--76},
  year         = {2022},
  doi          = {10.1016/J.TCS.2022.03.021}
}

@inproceedings{soda/Lampis25,
  author       = {Michael Lampis},
  title        = {The Primal Pathwidth {SETH}},
  booktitle    = {Proceedings of the 2025 Annual {ACM-SIAM} Symposium on Discrete Algorithms,
                  {SODA} 2025},
  pages        = {1494--1564},
  publisher    = {{SIAM}},
  year         = {2025},
  doi          = {10.1137/1.9781611978322.47}
}

@article{mst/GuoNW07,
  author       = {Jiong Guo and
                  Rolf Niedermeier and
                  Sebastian Wernicke},
  title        = {Parameterized Complexity of Vertex Cover Variants},
  journal      = {Theory Comput. Syst.},
  volume       = {41},
  number       = {3},
  pages        = {501--520},
  year         = {2007},
  doi          = {10.1007/S00224-007-1309-3}
}

@inproceedings{isaac/ChuL23,
  author       = {Huairui Chu and
                  Bingkai Lin},
  title        = {{FPT} Approximation Using Treewidth: Capacitated Vertex Cover, Target
                  Set Selection and Vector Dominating Set},
  booktitle    = {34th International Symposium on Algorithms and Computation, {ISAAC} 2023},
  series       = {LIPIcs},
  volume       = {283},
  pages        = {19:1--19:20},
  publisher    = {Schloss Dagstuhl - Leibniz-Zentrum f{\"{u}}r Informatik},
  year         = {2023},
  doi          = {10.4230/LIPICS.ISAAC.2023.19}
}

@inproceedings{soda/LokshtanovS0S025,
  author       = {Daniel Lokshtanov and
                  Abhishek Sahu and
                  Saket Saurabh and
                  Vaishali Surianarayanan and
                  Jie Xue},
  title        = {Parameterized Approximation for Capacitated \emph{d}-Hitting Set with
                  Hard Capacities},
  booktitle    = {Proceedings of the 2025 Annual {ACM-SIAM} Symposium on Discrete Algorithms,
                  {SODA} 2025},
  pages        = {1565--1592},
  publisher    = {{SIAM}},
  year         = {2025},
  doi          = {10.1137/1.9781611978322.48}
}

@inproceedings{sofsem/RooijR19,
  author       = {Sebastiaan B. van Rooij and
                  Johan M. M. van Rooij},
  title        = {Algorithms and Complexity Results for the Capacitated Vertex Cover
                  Problem},
  booktitle    = {{SOFSEM} 2019: Theory and Practice of Computer Science - 45th International
                  Conference on Current Trends in Theory and Practice of Computer Science, Proceedings},
  series       = {Lecture Notes in Computer Science},
  volume       = {11376},
  pages        = {473--489},
  publisher    = {Springer},
  year         = {2019},
  doi          = {10.1007/978-3-030-10801-4_37}
}

@inproceedings{iwpec/BodlaenderGJPP22,
  author       = {Hans L. Bodlaender and
                  Carla Groenland and
                  Hugo Jacob and
                  Marcin Pilipczuk and
                  Michal Pilipczuk},
  title        = {On the Complexity of Problems on Tree-Structured Graphs},
  booktitle    = {17th International Symposium on Parameterized and Exact Computation,
                  {IPEC} 2022},
  series       = {LIPIcs},
  volume       = {249},
  pages        = {6:1--6:17},
  publisher    = {Schloss Dagstuhl - Leibniz-Zentrum f{\"{u}}r Informatik},
  year         = {2022},
  doi          = {10.4230/LIPICS.IPEC.2022.6}
}

@article{dam/BodlaenderS26,
  author       = {Hans L. Bodlaender and
                  Krisztina Szil{\'a}gyi},
  title        = {{XALP}-completeness of parameterized problems on planar graphs},
  journal      = {Discret. Appl. Math.},
  volume       = {386},
  pages        = {156--174},
  year         = {2026},
  doi          = {10.1016/J.DAM.2026.01.021}
}

@inproceedings{waoa/Becker17,
  author       = {Amariah Becker},
  title        = {Capacitated Domination Problems on Planar Graphs},
  booktitle    = {Approximation and Online Algorithms - 15th International Workshop,
                  {WAOA} 2017},
  series       = {Lecture Notes in Computer Science},
  volume       = {10787},
  pages        = {1--16},
  publisher    = {Springer},
  year         = {2017},
  doi          = {10.1007/978-3-319-89441-6_1}
}

@inproceedings{innovations/RohwedderW25,
  author       = {Lars Rohwedder and
                  Karol Wegrzycki},
  title        = {Fine-Grained Equivalence for Problems Related to Integer Linear Programming},
  booktitle    = {16th Innovations in Theoretical Computer Science Conference, {ITCS}
                  2025},
  series       = {LIPIcs},
  pages        = {83:1--83:18},
  publisher    = {Schloss Dagstuhl - Leibniz-Zentrum f{\"{u}}r Informatik},
  year         = {2025},
  doi          = {10.4230/LIPICS.ITCS.2025.83}
}

@article{cmb/Lindstrom65,
  author       = {Bernt Lindstr{\"o}m},
  title        = {On a Combinatorial Problem in Number Theory},
  journal      = {Canadian Mathematical Bulletin},
  volume       = {8},
  number       = {4},
  pages        = {477--490},
  year         = {1965},
  doi          = {10.4153/CMB-1965-034-2}
}

@book{books/GareyJ79,
  author       = {M. R. Garey and
                  David S. Johnson},
  title        = {Computers and Intractability: {A} Guide to the Theory of {NP}-Completeness},
  publisher    = {W. H. Freeman},
  year         = {1979},
  isbn         = {0-7167-1044-7}
}

@article{siamdm/BelmonteKLMO22,
  author       = {R{\'{e}}my Belmonte and
                  Eun Jung Kim and
                  Michael Lampis and
                  Valia Mitsou and
                  Yota Otachi},
  title        = {Grundy Distinguishes Treewidth from Pathwidth},
  journal      = {{SIAM} J. Discret. Math.},
  volume       = {36},
  number       = {3},
  pages        = {1761--1787},
  year         = {2022},
  doi          = {10.1137/20M1385779}
}

@article{mor/Lenstra83,
  author       = {Hendrik W. Lenstra Jr.},
  title        = {Integer Programming with a Fixed Number of Variables},
  journal      = {Math. Oper. Res.},
  volume       = {8},
  number       = {4},
  pages        = {538--548},
  year         = {1983},
  doi          = {10.1287/MOOR.8.4.538}
}

@article{jal/GuhaHKO03,
  author       = {Sudipto Guha and
                  Refael Hassin and
                  Samir Khuller and
                  Einat Or},
  title        = {Capacitated vertex covering},
  journal      = {J. Algorithms},
  volume       = {48},
  number       = {1},
  pages        = {257--270},
  year         = {2003},
  doi          = {10.1016/S0196-6774(03)00053-1}
}

@article{toct/BonamyKPSW19,
  author       = {Marthe Bonamy and
                  Lukasz Kowalik and
                  Michal Pilipczuk and
                  Arkadiusz Socala and
                  Marcin Wrochna},
  title        = {Tight Lower Bounds for the Complexity of Multicoloring},
  journal      = {{ACM} Trans. Comput. Theory},
  volume       = {11},
  number       = {3},
  pages        = {13:1--13:19},
  year         = {2019},
  doi          = {10.1145/3313906}
}

@article{siamdm/Bar-YehudaFMR10,
  author       = {Reuven Bar{-}Yehuda and
                  Guy Flysher and
                  Juli{\'{a}}n Mestre and
                  Dror Rawitz},
  title        = {Approximation of Partial Capacitated Vertex Cover},
  journal      = {{SIAM} J. Discret. Math.},
  volume       = {24},
  number       = {4},
  pages        = {1441--1469},
  year         = {2010},
  doi          = {10.1137/080728044}
}

@article{siamcomp/GrandoniKPS08,
  author       = {Fabrizio Grandoni and
                  Jochen K{\"{o}}nemann and
                  Alessandro Panconesi and
                  Mauro Sozio},
  title        = {A Primal-Dual Bicriteria Distributed Algorithm for Capacitated Vertex
                  Cover},
  journal      = {{SIAM} J. Comput.},
  volume       = {38},
  number       = {3},
  pages        = {825--840},
  year         = {2008},
  doi          = {10.1137/06065310X}
}

@article{siamcomp/ChuzhoyN06,
  author       = {Julia Chuzhoy and
                  Joseph Naor},
  title        = {Covering Problems with Hard Capacities},
  journal      = {{SIAM} J. Comput.},
  volume       = {36},
  number       = {2},
  pages        = {498--515},
  year         = {2006},
  doi          = {10.1137/S0097539703422479}
}

@article{jcss/GandhiHKKS06,
  author       = {Rajiv Gandhi and
                  Eran Halperin and
                  Samir Khuller and
                  Guy Kortsarz and
                  Aravind Srinivasan},
  title        = {An improved approximation algorithm for vertex cover with hard capacities},
  journal      = {J. Comput. Syst. Sci.},
  volume       = {72},
  number       = {1},
  pages        = {16--33},
  year         = {2006},
  doi          = {10.1016/J.JCSS.2005.06.004}
}

@article{algorithmica/Kao21,
  author       = {Mong{-}Jen Kao},
  title        = {Iterative Partial Rounding for Vertex Cover with Hard Capacities},
  journal      = {Algorithmica},
  volume       = {83},
  number       = {1},
  pages        = {45--71},
  year         = {2021},
  doi          = {10.1007/S00453-020-00749-9}
}

@inproceedings{soda/Wong17,
  author       = {Sam Chiu{-}wai Wong},
  title        = {Tight Algorithms for Vertex Cover with Hard Capacities on Multigraphs
                  and Hypergraphs},
  booktitle    = {Proceedings of the Twenty-Eighth Annual {ACM-SIAM} Symposium on Discrete
                  Algorithms, {SODA} 2017},
  pages        = {2626--2637},
  publisher    = {{SIAM}},
  year         = {2017},
  doi          = {10.1137/1.9781611974782.173}
}

@article{tcs/KaoSLL19,
  author       = {Mong{-}Jen Kao and
                  Jia{-}Yau Shiau and
                  Ching{-}Chi Lin and
                  D. T. Lee},
  title        = {Tight approximation for partial vertex cover with hard capacities},
  journal      = {Theor. Comput. Sci.},
  volume       = {778},
  pages        = {61--72},
  year         = {2019},
  doi          = {10.1016/J.TCS.2019.01.027}
}

@article{siamdm/GanianKS22,
  author       = {Robert Ganian and
                  Eun Jung Kim and
                  Stefan Szeider},
  title        = {Algorithmic Applications of Tree-Cut Width},
  journal      = {{SIAM} J. Discret. Math.},
  volume       = {36},
  number       = {4},
  pages        = {2635--2666},
  year         = {2022},
  doi          = {10.1137/20M137478X}
}

@inproceedings{soda/CheungGW14,
  author       = {Wang Chi Cheung and
                  Michel X. Goemans and
                  Sam Chiu{-}wai Wong},
  title        = {Improved Algorithms for Vertex Cover with Hard Capacities on Multigraphs
                  and Hypergraphs},
  booktitle    = {Proceedings of the Twenty-Fifth Annual {ACM-SIAM} Symposium on Discrete
                  Algorithms, {SODA} 2014},
  pages        = {1714--1726},
  publisher    = {{SIAM}},
  year         = {2014},
  doi          = {10.1137/1.9781611973402.124}
}

@article{JansenS97,
  author       = {Klaus Jansen and
                  Petra Scheffler},
  title        = {Generalized Coloring for Tree-like Graphs},
  journal      = {Discret. Appl. Math.},
  volume       = {75},
  number       = {2},
  pages        = {135--155},
  year         = {1997},
  doi          = {10.1016/S0166-218X(96)00085-6}
}

@article{BelmonteLM22,
  author       = {R{\'{e}}my Belmonte and
                  Michael Lampis and
                  Valia Mitsou},
  title        = {Defective Coloring on Classes of Perfect Graphs},
  journal      = {Discret. Math. Theor. Comput. Sci.},
  volume       = {24},
  number       = {1},
  year         = {2022},
  doi          = {10.46298/DMTCS.4926}
}

@article{algorithmica/DrangeDH16,
  author       = {P{\aa}l Gr{\o}n{\aa}s Drange and
                  Markus S. Dregi and
                  Pim van 't Hof},
  title        = {On the Computational Complexity of Vertex Integrity and Component
                  Order Connectivity},
  journal      = {Algorithmica},
  volume       = {76},
  number       = {4},
  pages        = {1181--1202},
  year         = {2016},
  doi          = {10.1007/S00453-016-0127-X},
}

@misc{corr/EisenbrandHKKLO19,
  author       = {Friedrich Eisenbrand and
                  Christoph Hunkenschr{\"o}der and
                  Kim-Manuel Klein and
                  Martin Kouteck{\'y} and
                  Asaf Levin and
                  Shmuel Onn},
  title        = {An Algorithmic Theory of Integer Programming},
  year         = {2019},
  eprint       = {1904.01361},
  archivePrefix= {arXiv}
}

@inproceedings{swat/KnopMV26,
  author       = {Dusan Knop and
                  Nikolaos Melissinos and
                  Manolis Vasilakis},
  title        = {Parameterized Critical Node Cut Revisited},
  booktitle    = {20th Scandinavian Symposium on Algorithm Theory, {SWAT} 2026},
  series       = {LIPIcs},
  volume       = {370},
  pages        = {25:1--25:17},
  publisher    = {Schloss Dagstuhl - Leibniz-Zentrum f{\"{u}}r Informatik},
  year         = {2026},
  doi          = {10.4230/LIPICS.SWAT.2026.25}
}

@inproceedings{mfcs/BlazejKPS24,
  author       = {V{\'{a}}clav Blazej and
                  Dusan Knop and
                  Jan Pokorn{\'{y}} and
                  Simon Schierreich},
  title        = {Equitable Connected Partition and Structural Parameters Revisited:
                  N-Fold Beats {Lenstra}},
  booktitle    = {49th International Symposium on Mathematical Foundations of Computer
                  Science, {MFCS} 2024},
  series       = {LIPIcs},
  pages        = {29:1--29:16},
  publisher    = {Schloss Dagstuhl - Leibniz-Zentrum f{\"{u}}r Informatik},
  year         = {2024},
  doi          = {10.4230/LIPICS.MFCS.2024.29}
}

@article{disopt/GavenciakKK22,
  author       = {Tomas Gavenciak and
                  Martin Kouteck{\'{y}} and
                  Dusan Knop},
  title        = {Integer programming in parameterized complexity: Five miniatures},
  journal      = {Discret. Optim.},
  volume       = {44},
  number       = {Part},
  pages        = {100596},
  year         = {2022},
  doi          = {10.1016/J.DISOPT.2020.100596}
}

@article{toct/KnopPW20,
  author       = {Dusan Knop and
                  Michal Pilipczuk and
                  Marcin Wrochna},
  title        = {Tight Complexity Lower Bounds for Integer Linear Programming with
                  Few Constraints},
  journal      = {{ACM} Trans. Comput. Theory},
  volume       = {12},
  number       = {3},
  pages        = {19:1--19:19},
  year         = {2020},
  doi          = {10.1145/3397484}
}

@article{mor/JansenR23,
  author       = {Klaus Jansen and
                  Lars Rohwedder},
  title        = {On Integer Programming, Discrepancy, and Convolution},
  journal      = {Math. Oper. Res.},
  volume       = {48},
  number       = {3},
  pages        = {1481--1495},
  year         = {2023},
  doi          = {10.1287/MOOR.2022.1308}
}

@inproceedings{esa/RandolphW24,
  author       = {Tim Randolph and
                  Karol W{\k{e}}grzycki},
  title        = {Parameterized Algorithms on Integer Sets with Small Doubling: Integer
                  Programming, Subset Sum and k-SUM},
  booktitle    = {32nd Annual European Symposium on Algorithms, {ESA} 2024},
  series       = {LIPIcs},
  pages        = {96:1--96:19},
  publisher    = {Schloss Dagstuhl - Leibniz-Zentrum f{\"{u}}r Informatik},
  year         = {2024},
  doi          = {10.4230/LIPICS.ESA.2024.96}
}

@book{books/KorteV18,
  author    = {Bernhard Korte and
               Jens Vygen},
  title     = {Combinatorial Optimization: Theory and Algorithms},
  publisher = {Springer},
  edition   = {6},
  year      = {2018},
  doi       = {10.1007/978-3-662-56039-6},
  isbn      = {978-3-662-56038-9}
}

@inproceedings{icalp/Lampis14,
  author       = {Michael Lampis},
  title        = {Parameterized Approximation Schemes Using Graph Widths},
  booktitle    = {Automata, Languages, and Programming - 41st International Colloquium,
                  {ICALP} 2014},
  series       = {Lecture Notes in Computer Science},
  volume       = {8572},
  pages        = {775--786},
  publisher    = {Springer},
  year         = {2014},
  doi          = {10.1007/978-3-662-43948-7_64}
}

@article{jacm/Baker94,
  author       = {Brenda S. Baker},
  title        = {Approximation Algorithms for {NP}-Complete Problems on Planar Graphs},
  journal      = {J. {ACM}},
  volume       = {41},
  number       = {1},
  pages        = {153--180},
  year         = {1994},
  doi          = {10.1145/174644.174650}
}

@inproceedings{soda/DemaineH05,
  author       = {Erik D. Demaine and
                  Mohammad Taghi Hajiaghayi},
  title        = {Bidimensionality: new connections between {FPT} algorithms and PTASs},
  booktitle    = {Proceedings of the Sixteenth Annual {ACM-SIAM} Symposium on Discrete
                  Algorithms, {SODA} 2005},
  pages        = {590--601},
  publisher    = {{SIAM}},
  year         = {2005},
  url          = {http://dl.acm.org/citation.cfm?id=1070432.1070514}
}

@article{talg/LokshtanovMS18,
  author       = {Daniel Lokshtanov and
                  D{\'{a}}niel Marx and
                  Saket Saurabh},
  title        = {Known Algorithms on Graphs of Bounded Treewidth Are Probably Optimal},
  journal      = {{ACM} Trans. Algorithms},
  volume       = {14},
  number       = {2},
  pages        = {13:1--13:30},
  year         = {2018},
  doi          = {10.1145/3170442}
}

@phdthesis{thesis/Schepper25,
  author       = {Philipp Schepper},
  title        = {Faster, higher, easier:
                  Toward a systematic study of parameterized
                  vertex and edge selection problems},
  school       = {Saarland University},
  year         = {2025},
  doi          = {10.22028/D291-46499}
}

@inproceedings{LampisV25,
  author       = {Michael Lampis and
                  Manolis Vasilakis},
  title        = {Structural Parameterizations for Induced and Acyclic Matching},
  booktitle    = {Graph-Theoretic Concepts in Computer Science - 51st International
                  Workshop, {WG} 2025, Revised
                  Selected Papers},
  series       = {Lecture Notes in Computer Science},
  volume       = {16124},
  pages        = {360--376},
  publisher    = {Springer},
  year         = {2025},
  doi          = {10.1007/978-3-032-11835-6_26}
}

@inproceedings{esa/LampisL26,
  author       = {Michael Lampis and
                  Yiren Lu},
  title        = {Pure Nash Equilibria in Graphical Games of Bounded Width Revisited},
  booktitle    = {34th Annual European Symposium on Algorithms, {ESA} 2026},
  series       = {LIPIcs},
  volume       = {388},
  pages        = {74:1--74:23},
  publisher    = {Schloss Dagstuhl - Leibniz-Zentrum f{\"{u}}r Informatik},
  year         = {2026},
  doi          = {10.4230/LIPICS.ESA.2026.74}
}

@inproceedings{soda/Lampis26a,
  author       = {Michael Lampis},
  title        = {Circuits and Backdoors: Five Shades of the {SETH}},
  booktitle    = {Proceedings of the 2026 Annual {ACM-SIAM} Symposium on Discrete Algorithms,
                  {SODA} 2026},
  pages        = {1945--2001},
  publisher    = {{SIAM}},
  year         = {2026},
  doi          = {10.1137/1.9781611978971.71}
}

\end{document}